%% file: main.tex
\documentclass[11pt]{article}

\usepackage[T1]{fontenc}
\usepackage[utf8]{inputenc}
\usepackage{lmodern}
\usepackage{microtype}
\usepackage{amsmath,amssymb,amsthm}
\usepackage{mathtools}
\usepackage{graphicx}
\usepackage{placeins}
\usepackage{needspace}
\usepackage[round,authoryear]{natbib}
\usepackage{geometry}
\usepackage{xcolor}
\usepackage{hyperref}
\hypersetup{
  colorlinks=true,
  linkcolor=blue!50!black,
  citecolor=blue!50!black,
  urlcolor=blue!50!black,
  pdftitle={Strategic Disclosure and the Design of Attention},
  pdfauthor={Qian Cao and Yifei Sun},
  pdfsubject={Strategic disclosure and attention design},
  pdfkeywords={verifiable disclosure, competing advocates, attention design,
    selection-naive inference}
}

\newtheorem{theorem}{Theorem}
\newtheorem{proposition}{Proposition}
\newtheorem{lemma}{Lemma}
\newtheorem{corollary}{Corollary}
\newtheorem{assumption}{Assumption}
\theoremstyle{definition}

\theoremstyle{remark}
\newtheorem{remark}{Remark}
\newtheorem*{examplebox}{Example}

\newcommand{\R}{\mathbb{R}}
\newcommand{\E}{\mathbb{E}}
\newcommand{\Prob}{\mathbb{P}}
\newcommand{\ind}[1]{\mathbf{1}\{#1\}}

\begin{document}
\input{main-body}

\input{bibliography}
\clearpage
\input{appendix-body}

\end{document}

%% file: main-body.tex
\title{\textbf{Strategic Disclosure and the Design of Attention}%
\thanks{\raggedright Yifei Sun gratefully acknowledges financial support from
the National Natural Science Foundation of China (Grant No.~72273029). Qian Cao
gratefully acknowledges financial support from the Fundamental Research Funds
for the Central Universities at the University of International Business and
Economics (Grant No.~25QD43). All remaining errors are our own.}}
\author{Qian Cao\thanks{\raggedright Global Value Chain, University of International
Business and Economics, No.~10 Huixin East Street, Chaoyang District,
Beijing 100029, China. Email:
\texttt{caoqian\_econ@163.com}.}
\and
Yifei Sun\thanks{\raggedright Corresponding author. School of Economics, University of
International Business and Economics, No.~10 Huixin East Street, Chaoyang
District, Beijing 100029, China. Email:
\texttt{sunyifei@uibe.edu.cn}.}}
\date{September 2026}

\maketitle

\begin{abstract}
\noindent
An institution that cannot compel disclosure can still allocate exposure.
We study two opposed advocates who select verifiable observations and a
receiver who averages the selected record without correcting for selection.
The institution commits exposure before the state and evidence are known;
each withheld item's assigned weight expires. With commonly known
pools and strictly opposed monotone preferences, all finite-game Nash
equilibria produce one action and common cutoff.
In large pools with a vanishing anchor, directional tail moments determine
the implementing exposure share. Under common evidence laws, equal exposure
implements every state. The institutional restriction matters: full
within-camp reallocation induces extreme selection and, on a fixed support,
eliminates state responsiveness. In a finite-state common-law benchmark,
small reallocation separates disclosure cutoffs and changes implementing
shares according to favorable-tail frequencies. Variation in these
frequencies across states creates a positive local welfare cost even after
exposure is optimized. Ex ante results bound the cost of commitment and
identify information sufficient for limiting implementation. Separately,
under linear payoffs, uniform item weights and continuous evidence laws on
a compact state space, all Bayesian equilibria converge to the baseline
outcome when advocates observe only their own samples and learn the state.
Finite-pool results characterize a precision cost under common laws,
balanced exposure, uniform weights and equal limiting camp sizes, and first
corrections to exact risk-minimizing shares under additional regularity.
\end{abstract}



\section{Introduction}
\label{sec:intro}

Public decisions often depend on evidence selected by interested parties.
Litigants choose exhibits, political organizations circulate favorable
studies, and research organizations publicize findings that support their
positions. The disclosed observations can all be true while the resulting
record is misleading. An institution that cannot compel disclosure may
nevertheless control each side's opportunity to be heard. How should it
allocate that opportunity when advocates will adjust what they reveal?

We study an institution that divides potential exposure between two advocates
who want to move the same decision in opposite directions. Each advocate
controls a separate, exogenous pool of verifiable observations. Before the
state and observations are known, the institution commits a share of exposure
to each camp. Within a camp, an item's relative weight is predetermined and
expires if the item is withheld. The receiver averages disclosed observations
with a fixed benchmark, treating the selected sample as exogenous. Thus the
institution chooses opportunities, the advocates choose disclosures, and the
receiver acts on the selected record. We refer to the institution as a
platform; Section~\ref{sec:environment} gives a hypothetical interpretation as a
policy evidence panel with study-specific release rights and prominence.

The institutional distinction is between authority over presentation and
authority over release. A panel can announce how much prominence each
registered study will receive while its owner retains the right to withhold
the result. A sender that withholds an item loses that item's assigned
weight; it cannot transfer the weight to a more favorable finding. This
restriction makes the frequency of favorable evidence matter for its
influence: a rare favorable finding carries only its own assigned weight.
The analysis asks how to allocate exposure within this protocol and how
allowing unused weights to migrate changes selection and accuracy.

The distinction between potential exposure and realized speech is essential.
\begin{examplebox}
Suppose the camps have equally many observations and uniform within-camp
weights. Both draw from the same population: an observation equals $9$ with
probability $1/5$ and $4$ with probability $4/5$. The full-evidence mean and receiver benchmark are
$5$; the benchmark's weight vanishes in the limit. If the upward advocate presents the $9$s and the downward advocate the
$4$s, a display with equal numbers of the two values has mean $6.5$.

Instead, give the camps equal potential exposure. At a candidate decision
$c\in(4,9)$, their expected upward and downward pressures balance when
\[
 \tfrac12\cdot\tfrac15(9-c)
 =\tfrac12\cdot\tfrac45(c-4).
\]
The solution is $c=5$. In the large-pool limit, one fifth of the upward
camp's opportunities and four fifths of the downward camp's opportunities
are used. The visible record converges to $20\%$ nines and $80\%$ fours,
with mean $5$. Equal opportunities produce unequal speech and the correct
limiting decision.

If each camp can instead concentrate its entire exposure on its retained
observations, the upward camp presents only $9$s and the downward camp only
$4$s. Equal camp shares now give the two values equal total weight, yielding
$6.5$. The rarity of $9$s no longer limits their influence.
\end{examplebox}
The number $5$ is not an imposed screening standard. Adding an observation
above the current average raises it; adding one below lowers it. Each
advocate therefore selects relative to the decision that both advocates
jointly produce. The example suppresses a vanishing anchor and uses a
discrete law for transparency; the Appendix provides a smooth version.

Our first result solves this feedback for every finite evidence pool.
With common knowledge of the realized pools, all Nash equilibria, including
mixed equilibria, produce
the same action and common cutoff. The upward advocate reveals observations
above it and the downward advocate observations below it. Strategies can
differ only at cutoff ties, which do not affect the decision. This entire
equilibrium characterization holds for strictly increasing and strictly
decreasing sender preferences, beyond the baseline linear normalization. The
institution can therefore evaluate a policy without selecting among
equilibrium outcomes.

In large pools, increasing the upward camp's exposure raises the common
cutoff. That camp discloses a smaller fraction of its eligible evidence,
and its disclosed observations have a higher conditional mean. The opposing
camp discloses more, and the public action rises. This response holds with
different evidence laws across camps. The fraction disclosed and the total
visible weight are distinct margins: more opportunity can mean less
frequent disclosure and greater influence.

The implementation and design results concern large pools with diffuse
item weights and a vanishing anchor. At a given state, the implementing
share balances the camps' expected favorable deviations from the correct
action. Under a common evidence law, equal exposure implements that action
even when the law is skewed, consistently with the familiar identity that
the one-half expectile is the mean. With camp-specific laws, the two
directional moments need not agree. The state-specific implementing share
is a benchmark for the institution's commitment problem; the institution
chooses exposure before learning the state.

Our central institutional comparison asks what happens when unused exposure
can be reassigned. Under common evidence laws on a fixed support, equal
nontransferable exposure implements every state in the limit. Full
within-camp reallocation induces selection of the most extreme observations;
the limiting action then depends only on the support endpoints and the
exposure share. Its minimum limiting squared error is at least the prior
variance of the state. Binding exposure to individual items therefore has
a strict welfare value in this benchmark with a nondegenerate prior.

Partial reallocation reveals the mechanism behind this comparison. Even a
small reassigned fraction separates the advocates' disclosure cutoffs: the
upward camp withholds some observations above the public action, and the
downward camp withholds some observations below it. In a finite-state
common-law benchmark, the implementing share changes according to the
fractions of observations on either side of the state's mean. When those
fractions vary across states, the required policy adjustments conflict.
The best committed share then leaves a strictly positive welfare loss at
second order in the reassigned fraction. This is a local result; the paper
does not claim a global welfare ranking of all reallocation rules.

The broader commitment problem allows the camps' evidence laws to differ.
One fixed share implements every state exactly when the state-specific
shares coincide. Otherwise every optimum lies strictly between the extreme
targets. Bounds on the action's sensitivity to exposure translate the
variance of these targets into bounds on optimal limiting risk. If a signal
is available before commitment, the corresponding bounds use the expected
conditional variance of the target share. Information identifying the
required exposure balance suffices for implementation even if it does not
identify the state. A two-state economy illustrates the optimal compromise.

A separate extension removes common knowledge of the evidence pools. With
linear sender payoffs, uniform item weights, compactness and continuity,
every Bayesian equilibrium converges to the baseline action and cutoff
when each advocate observes only its own sample and learns the unknown
state. Neither advocate needs to observe the state or the other pool.
This result concerns the nontransferable institution and is separate from
the partial-reallocation analysis.

Finally, eliminating limiting bias does not recover the precision of
complete disclosure. Under common evidence laws, equal exposure, uniform
item weights and asymptotically equal camp sizes, the excess asymptotic
variance is $\operatorname{Var}_{\theta}(|X-\theta|)$. Under additional
smoothness and curvature conditions, we derive the first adjustment to
every exact finite-pool risk-minimizing share. The adjustment accounts for
cutoff bias, the receiver's benchmark and sampling variance; its policy
displacement is at most of order $1/n$ and its welfare gain at most of order
$1/n^2$. These are refinements of the institution's limiting performance.

The contribution is an exposure-design problem in which presentation
authority and disclosure rights are held separately. The finite equilibrium
characterization makes policy outcomes determinate. The institutional
comparison and its local extension show why reallocating unused exposure
changes the evidence selected and can prevent one committed policy from
correcting every state. Private learning and finite-pool analysis establish,
respectively, robustness to sender information and limits to statistical
precision. Opposed selection and its potential benefits have close
precedents; our claims concern this allocation problem under a fixed,
selection-naive receiver rule and the stated exposure instruments.

Short corollaries are proved in the text; the Appendix contains the
remaining proofs, the receiver's subjective working model, and extensions
to personalized exposure and processing of empty slots. An earlier version
of this paper appeared as
\href{https://arxiv.org/abs/2608.27903v1}{arXiv:2608.27903v1}.

\Needspace{6\baselineskip}
\subsection{Related literature}
\citet{froebkobayashi1996} provide a close precedent for beneficial opposed
selection. Litigants produce costly evidence for a jury that treats the
selected record as random; when both litigants produce evidence, the
full-information decision results. Our additional object of study is an
institution's allocation of exposure over exogenous evidence pools before
the state and observations are known. The senders face no direct disclosure
or concealment cost, but withholding an item forfeits its assigned exposure.
Characterizing the finite disclosure game allows us to evaluate that
allocation and compare it with full within-camp reallocation.

\citet{kartikleesuen2026} study privately informed senders, each of whom
can disclose a verifiable signal to a Bayesian receiver, with disclosure
and concealment costs. Under their conditionally independent signal
structure and linear sender preferences, zero-cost disclosure thresholds
are strategically independent. With message costs, disclosures are strategic substitutes under disclosure
costs and complements under concealment costs, as characterized by their
extremal best-response thresholds \citep[Section~4]{kartikleesuen2026}.
In our baseline, the realized pools are commonly known between advocates, each of
whom chooses a subset from the pool it controls. The averaging rule makes
the equilibrium action the disclosure cutoff for both advocates. The
institution changes both through its prior commitment of exposure.

\citet{ditillioottavianisorensen2021} derive strategic sample selection
when an evaluator anticipates selection. Their design analysis varies
sample and presample sizes; their application to jury selection also
examines opposed parties' challenges from a common pool
\citep[Sections~5.2--5.3 and~6.2]{ditillioottavianisorensen2021}.
Our institution holds the evidence technologies fixed and allocates
potential exposure across separately controlled pools. Each item's weight
expires upon suppression, and advocates choose both how much and which
evidence to disclose to a receiver following a fixed average. The design
problem concerns how this exposure allocation changes the common cutoff
and accuracy across states.

\citet{gao2026} studies deletion from a large dataset when a Bayesian
receiver evaluates its disclosed size and composition. The analysis
characterizes receiver-optimal imitation outcomes, derives a welfare
ordering of data-generating experiments, and establishes finite-data
convergence under the truth-leaning refinement
\citep[Sections~2.1 and~3--5]{gao2026}.
We hold the evidence technologies and the receiver's averaging rule fixed,
and allocate exposure across two opposed evidence owners.
The policy operates through their equilibrium subset choices rather than
through a choice of experiment or a change in the receiver's inference rule.

\citet{mohseni2022} show how journalistic selection of true observations
can distort the beliefs of learners who do not correct for curation.
Their specified extremity and fairness rules illustrate why reporting
unusual events or balancing opposite sides of a fixed neutral point can
mislead. In our baseline game, disclosure follows from the mutual best responses
of opposed evidence owners, and the dividing standard is the action their
disclosures jointly produce. Committed exposure changes this standard,
which is why equal opportunities need not produce equal realized speech.

The broader disclosure literature emphasizes skeptical inference from silence
\citep{grossman1981,milgrom1981,rappoport2025}, competition among information
holders \citep{milgromroberts1986,lipmanseppi1995}, and evidence acquisition
and presentation to a Bayesian arbiter \citep{dewatriponttirole1999,shin1998}.
Competitive persuasion instead permits senders to choose information
structures \citep{gentzkowkamenica2017,aukawai2020}. Here the institution
commits exposure and the advocates select verifiable realizations.
The selected-average receiver isolates this channel. Related models study behavioral receivers in
communication \citep{ottavianisquintani2006,chen2011,eysterrabin2005} or limited
processing of disclosed information \citep{hirshleiferteoh2003}. Our Gaussian
working-model representation identifies a particular misspecification,
rather than correctly specified equilibrium inference.

\citet{farina2024} combine a model of selective disclosure under a
communication-capacity constraint with an experiment. Their analysis varies
the quantity of available evidence and the number of signals that can be
disclosed, showing how these affect selection and informativeness. Our
instrument assigns weights to individual evidence items, with unused
weights expiring. The findings in \citet{jin2021,farina2024} motivate
incomplete correction for selection: they support insufficient skepticism,
not the exact absence of selection inference imposed by our averaging rule.

Under a common evidence law the limiting cutoff is an $r$-expectile
\citep{neweypowell1987}; the one-half expectile equals the mean. We use these
identities and the associated asymptotic literature as benchmarks
\citep{ehm2016,leeullahwang2019}. The economic questions concern the
endogenous cutoff with camp-specific laws and the choice of disclosure
opportunities before the state is known.

\section{Model}
\label{sec:game}

\subsection{Players, evidence, and timing}
\label{sec:environment}

The game has a designer, two senders, and a receiver. The platform is the
designer: it allocates potential exposure before the state and evidence are
known. The two senders are opposed advocates who control separate pools of
verifiable evidence. After the evidence is realized, each sender chooses
which of its assigned display opportunities to fill. The receiver sees the
displayed observations and takes a one-dimensional action; a withheld item
leaves its promised slot empty. The designer therefore controls exposure,
the senders control disclosure, and the receiver acts on the selected record.

\paragraph{A running interpretation.}
Consider a policy evidence panel presenting estimates of a common policy
effect. Two research organizations favor higher and lower assessments.
Before results are known, the panel registers a fixed set of studies, their
owners and the relative prominence of their entries within each
organization. Its administrator chooses the division of total potential
prominence between the two organizations. The eligible studies and their
within-organization weights are fixed by the panel's charter, so the
administrator's remaining instrument is this division.

Each entry accepts only its listed study's authenticated result, at most
once. Reposting a result, splitting its presentation into several entries,
or substituting a different study creates no additional eligible item.
Authentication enforces the link between an entry and its evidence; the
published prominence schedule prevents an organization from moving an
unused entry's weight to a favorable finding. If a study is withheld, its
assigned weight expires. The receiver normalizes the weights of released
results and a fixed benchmark, so a displayed item's normalized coefficient
can still increase when another item is withheld.

The owners retain release rights under the participation agreement. The
administrator can specify how an authorized result is presented but cannot
publish it without the owner's consent. In the baseline, the organizations
have confidential access to each other's completed results, while each
retains exclusive control over public release of its own studies. The
private-learning extension dispenses with this mutual access. In either
case, access to evidence and authority to disclose it are distinct.

Here attention means assigned influence in the displayed record; consumer
processing time is outside the model. The audience averages released
estimates using their assigned weights and an outside benchmark, without
correcting for selection. This hypothetical protocol fixes the study set,
release rights, within-organization prominence and audience rule. We
optimize the camp share within that protocol. Its enforcement requirements
are substantive: allowing new entries, substitutes or freely reassigned
prominence would change the game. The reallocation comparisons isolate the
last of these changes while holding the evidence technology fixed.

The receiver's decision concerns an unknown state $\theta$ in the interior
of $\mathcal X=[\underline x,\overline x]$. Denote the receiver's action by
$a\in\mathcal X$. The state is the correct action: if $\theta$ were known,
the receiver would choose $a=\theta$. Depending on the application, $a$ may
be a policy choice, forecast, rating, damages award, or another
one-dimensional decision whose accuracy is evaluated against the state.

\paragraph{Preferences.}
The senders are advocates rather than neutral experts. Index them by
$k\in\{1,2\}$ and let $b_k\in\{-1,1\}$ record the direction of sender
$k$'s bias. We set $b_1=1$ and $b_2=-1$: Sender 1 prefers a higher receiver
action, whereas Sender 2 prefers a lower one. Neither sender values
disclosure in itself. Evidence matters only through its effect on the
receiver's action. We normalize sender $k$'s payoff to
\[
  u_k(a)=b_k a,
  \qquad b_1=1,\quad b_2=-1.
\]
The linear specification is convenient, but the entire finite Nash
characterization with commonly known pools holds whenever $u_1(a,\theta)$
is strictly increasing in $a$ and $u_2(a,\theta)$ is strictly decreasing
(Proposition~\ref{prop:ordinal} in the Appendix). This is an
invariance result for the equilibrium set, not for rankings of arbitrary
lotteries. The private-evidence analyses retain linear sender payoffs.
The underlying preference restriction is directional conflict over the
relevant action range. It fits,
for example, litigants arguing for higher or lower damages, research desks
advocating expansion or contraction, and interest groups pushing a policy in
opposite directions.

The receiver and the platform instead value accuracy. We use
$-(a-\theta)^2$ to evaluate their ex post welfare. The platform's design
problem is therefore not to choose the evidence, the receiver's prior, or
the detailed ranking of items within a camp. It chooses the share of total
potential exposure assigned to each of two sources that will select evidence
in opposite directions.

\paragraph{Evidence and control.}
Information about the state comes from a finite pool of observations
$X_1,\ldots,X_n$, with every $X_i$ taking a value in $\mathcal X$. Let
$I_k$ denote the set of observations controlled by sender $k$; the two sets
$I_1$ and $I_2$ partition the evidence pool. Ownership determines disclosure
rights: a sender may reveal or suppress an observation in its own pool, but
it cannot alter the observation, fabricate a new one, or reveal evidence
owned by the other sender. The finite disclosure game treats this realized
pool as arbitrary. Section~\ref{sec:largepool} introduces a stochastic
evidence environment only when we study large pools and ex ante design.

\paragraph{The attention policy.}
The institutional environment fixes the relative prominence of observations
within each camp. Write $\omega_{ni}^k>0$ for observation $i$'s within-camp
weight, collect these weights in the schedule $\omega_n$, and normalize
$\sum_{i\in I_k}\omega_{ni}^k=1$ for $k\in\{1,2\}$. Before observing the
state or any evidence value, the platform chooses one scalar policy: Sender
1's share $r_n\in(0,1)$ of total potential exposure. The resulting display
weight of observation $i$ is
\begin{equation}
 v_{ni}
 =
 \begin{cases}
   r_n\omega_{ni}^1,&i\in I_1,\\
   (1-r_n)\omega_{ni}^2,&i\in I_2.
 \end{cases}
 \label{eq:attention-schedule}
\end{equation}
Thus $\sum_i v_{ni}=1$ and $r_n=\sum_{i\in I_1}v_{ni}$. If observation
$i$ is disclosed, $v_{ni}$ determines its influence on the receiver; if it
is withheld, that opportunity expires. The platform chooses the division
between camps, not the within-camp schedule $\omega_n$.

The receiver's decision rule contains two additional fixed primitives: a
benchmark $a_0\in\mathrm{int}\,\mathcal X$ and its weight $\eta_n>0$.
They can represent prior information or a default and keep the receiver's
action well defined even if both senders suppress their entire pools. The
platform does not choose either $a_0$ or $\eta_n$; the finite-pool analysis
shows how their exogenous pull affects its optimal attention share.

\paragraph{Timing, disclosure, and the receiver's action.}
The within-camp schedule $\omega_n$, the benchmark $a_0$, and its weight
$\eta_n$ are fixed first. The platform then commits to the attention share
$r_n$. Nature next draws the state and the evidence pool. Both
senders observe the full realized evidence profile; ownership restricts
what they can disclose, not what evidence they know. They need not observe
$\theta$. With directional preferences, their disclosure choices depend on
whether an item moves the receiver action up or down, not on knowledge of the
correct action. They simultaneously choose disclosure sets
$S_1\subseteq I_1$ and $S_2\subseteq I_2$. Finally, the receiver observes
the displayed set $S=S_1\cup S_2$ and its realized values
$\{X_i:i\in S\}$.

The receiver gives a disclosed observation its promised weight and gives an
empty slot no weight. In particular, exposure attached to a suppressed item
is neither reassigned within its camp nor transferred to the opponent. The
resulting receiver action, denoted $a_n(S)$, is the anchored average
\begin{equation}
  a_n(S)
  =
  \frac{\eta_n a_0+\sum_{i\in S}v_{ni}X_i}
       {\eta_n+\sum_{i\in S}v_{ni}},
  \qquad S=S_1\cup S_2.
  \label{eq:action}
\end{equation}
The numerator combines the benchmark with the displayed values, while the
denominator normalizes only the weight that is actually filled. Consequently,
$a_n(S)\in\mathcal X$ even when little evidence is shown. For now,
\eqref{eq:action} can be read as a reduced-form decision rule.
The Appendix shows that it is also the
quadratic-loss Bayes action under a subjective Gaussian model in which the
displayed sample is treated as exogenous. The receiver responds to the
evidence it sees but draws no inference from which slots are empty.

The senders' disclosure incentives can be seen from one calculation. For a
current display $S$ and an undisclosed observation $i$, adding that observation
changes the receiver's action by
\begin{equation}
  a_n(S\cup\{i\})-a_n(S)
  =
  \frac{v_{ni}\bigl(X_i-a_n(S)\bigr)}
       {\eta_n+\sum_{j\in S}v_{nj}+v_{ni}}.
  \label{eq:marginal-item}
\end{equation}
Because the denominator is positive, the sign of the change depends only on
whether $X_i$ lies above or below the current action. Sender 1
therefore wants favorable observations above the action in the display and
unfavorable observations below it out; Sender 2 has the reverse
incentive. The relevant standard is endogenous: the action used to classify
an observation as favorable is itself produced by both senders' disclosure
choices. Their mutual best responses therefore generate the common cutoff
characterized below.

For a realized evidence pool, a pure strategy is a disclosure set, and a
mixed strategy is a probability distribution over the finite collection of
such sets. Unless stated otherwise, equilibrium means Nash equilibrium, pure
or mixed, of this simultaneous, complete-information disclosure subgame.
Theorem~\ref{thm:finite} characterizes the support of every mixed equilibrium;
the private-evidence extension separately states its behavioral Bayesian
equilibrium concept. Because unused
slots expire, $r_n$ is Sender 1's share of potential exposure, not
its realized share of the visible record. Our benchmark is the
\emph{complete-disclosure action} $a_n(I_1\cup I_2)$. It holds fixed the
platform's chosen share, the predetermined within-camp schedule, and the
receiver's benchmark, and removes only the senders' strategic suppression.
Section~\ref{sec:missing-processing} allows unused slots to be partially
backfilled with neutral content at $a_0$; the Appendix gives the full
analysis. Section~\ref{sec:partial-reallocation} studies partial
within-camp reallocation, with its finite-game foundation in the appendix.

The finite disclosure game places no distributional restriction on the
realized evidence profile. In the baseline, that profile is common knowledge
between the two professional senders (litigants after discovery, opposing
research desks monitoring a common professional record), while the receiver
observes only what is disclosed. The senders need not know the latent state.
Section~\ref{sec:largepool} introduces stochastic assumptions only when they
are needed for population limits and attention design.
The Appendix removes common knowledge of the evidence profile under
camp-uniform attention. Proposition~\ref{prop:private} first treats a
known state. Proposition~\ref{prop:unknown-state} then lets both senders
learn an unknown state from their own samples. With a compact prior support
and conditional evidence laws continuous in the state, every Bayesian
equilibrium converges to the baseline outcome along any joint large-pool,
vanishing-anchor sequence.

Four modeling choices deserve comment. First, attention is committed
capacity rather than an ex post ranking choice. In the evidence-panel
interpretation, the editor preassigns prominence to source-tagged study
slots, and a suppressed result leaves its promised position empty. Reallocation
of unused capacity defines a different game, studied in
Section~\ref{sec:partial-reallocation} and the Appendix;
partial neutral backfilling is
introduced in Proposition~\ref{prop:missing-processing} and developed there.
Second, the strategic players are
senders who control positive masses of evidence rather than atomistic
speakers. When the total weight in the averaging denominator is bounded away from
zero, an atomistic speaker changes the average by $O(v_i)$. A sender
internalizes the effect of its entire disclosure policy. Third,
evidence is verifiable, so the strategic choice is suppression rather than
fabrication. Fourth, the receiver processes displayed values while treating
the selected set as exogenous. This is a tractable limiting case motivated by
experimental evidence of insufficient skepticism and partial neglect of
selection \citep{jin2021,farina2024}. Those findings do not establish the
particular anchored-average rule maintained here. The Appendix gives
a Bayesian interpretation under a misspecified selection model.

\subsection{Equilibrium disclosure}

We now solve the disclosure game for an arbitrary finite evidence pool. The
marginal calculation in~\eqref{eq:marginal-item} suggests a common cutoff,
but it does not by itself prove that the senders' best responses meet or
that randomization cannot generate additional equilibrium outcomes. Both
questions can be settled with one scalar balance condition.

For a candidate public action $c\in\mathcal X$, define
\begin{equation}
  H_n(c)=\eta_n(a_0-c)+\sum_{i\in I_1}v_{ni}(X_i-c)^+
           -\sum_{i\in I_2}v_{ni}(c-X_i)^+,
  \label{eq:Hn}
\end{equation}
where $z^+=\max\{z,0\}$. This function measures net pressure on the
receiver's action at $c$. The anchor and Sender 1's observations above $c$
pull the action upward, whereas Sender 2's observations below $c$ pull it
downward. Theorem~\ref{thm:finite} shows that these forces balance in
equilibrium. The receiver's action and both senders' marginal disclosure
standard are therefore the same number---the unique zero of $H_n$.

For a candidate cutoff $c$, let $\mathcal T_1(c)$ and $\mathcal T_2(c)$
denote the disclosure sets that contain the strict favorable tail and differ
from it, if at all, only at observations equal to $c$:
\begin{equation}
 \begin{aligned}
 \mathcal T_1(c)
 &=\left\{S_1:\{i\in I_1:X_i>c\}\subseteq S_1
          \subseteq\{i\in I_1:X_i\geq c\}\right\},\\
 \mathcal T_2(c)
 &=\left\{S_2:\{i\in I_2:X_i<c\}\subseteq S_2
          \subseteq\{i\in I_2:X_i\leq c\}\right\}.
 \end{aligned}
 \label{eq:tail-families}
\end{equation}

\begin{theorem}[Common-cutoff equilibrium]
\label{thm:finite}
For every finite evidence pool and every realized evidence profile, there is
a unique equilibrium action. Specifically, $H_n$ has a unique zero
$c_n\in(\underline x,\overline x)$. A mixed-strategy profile
$(\sigma_1,\sigma_2)$ is a Nash equilibrium if and only if
\begin{equation}
 \sigma_1\bigl(\mathcal T_1(c_n)\bigr)=1
 \quad\text{and}\quad
 \sigma_2\bigl(\mathcal T_2(c_n)\bigr)=1.
 \label{eq:mixed-equilibrium-support}
\end{equation}
Consequently, every equilibrium produces $a_n=c_n$ with probability one and
the same disclosure record away from observations equal to $c_n$. In
particular, the strict-tail profile is a pure-strategy equilibrium.
\end{theorem}

The theorem does more than exhibit a threshold equilibrium. It closes the
equilibrium-selection problem for the baseline game. Multiplicity in the
treatment of observations exactly at the cutoff is economically irrelevant:
even mixed equilibria induce the same decision in every realization and the
same visible evidence away from ties. Randomization creates no additional
action risk, so the platform can evaluate an attention policy without an
equilibrium-selection rule. If no observation equals $c_n$, both tail
families are singletons and the equilibrium strategy profile is itself
unique.

The proof here uses the baseline zero-sum normalization. Proposition~\ref{prop:ordinal}
in the Appendix extends the entire characterization to strictly
opposed monotone sender preferences. Its proof uses the additive structure
of the action's centered numerator to show that every supported pure
profile in any mixed equilibrium must give the same action. Hence the
baseline design and welfare calculations also apply under these preferences.

The balance function $H_n$ has a direct economic meaning. At a candidate
action $c$, Sender 1 can add the weighted surplus
$\sum v_{ni}(X_i-c)^+$, Sender 2 can add the weighted shortfall
$\sum v_{ni}(c-X_i)^+$, and the anchor contributes $\eta_n(a_0-c)$. Their net
pressure vanishes at the equilibrium action. This single equation carries
the comparative statics below.

\section{How Attention Changes Disclosure}
\label{sec:largepool}

Theorem~\ref{thm:finite} did not require a probability model for evidence.
To study systematic disclosure and design across states, we now specify how
the two pools are generated.

\begin{assumption}[Evidence]
\label{ass:evidence}
Conditional on $\theta\in\mathrm{int}\,\mathcal X$, the items
$\{X_{i}\}_{i\in I_1}$ are i.i.d.\ $F_{1,\theta}$ and
$\{X_{i}\}_{i\in I_2}$ are i.i.d.\ $F_{2,\theta}$, mutually independent.
Each $F_{k,\theta}$, $k\in\{1,2\}$, is atomless with a density
$f_{k,\theta}$ that is continuous and strictly positive on
$\mathrm{int}\,\mathcal X$. If $X_k$ denotes a generic draw from
$F_{k,\theta}$, then $\E_{\theta}[X_k]=\theta$.
\end{assumption}

The common-mean restriction makes the state the full-disclosure benchmark:
if either pool were shown without selection, its average would converge to
$\theta$. The camps may nevertheless have different variances, skewness, and
tail behavior. Atomlessness and positive interior densities avoid gaps and
ties at the limiting disclosure standard; they are not needed for the exact
finite-pool equilibrium.

For a candidate public action $c$, the population counterpart of the finite
balance function is
\begin{equation}
  H(c;r,\theta)
  =r\,\E_\theta[(X_1-c)^+]
   -(1-r)\,\E_\theta[(c-X_2)^+].
  \label{eq:H}
\end{equation}
The two terms are the expected upward and downward pressure generated by the
tail observations each sender is willing to disclose. We will also use their
values at the correct action,
\[
  \Gamma_1(\theta)=\E_\theta[(X_1-\theta)^+]>0,
  \qquad
  \Gamma_2(\theta)=\E_\theta[(\theta-X_2)^+]>0.
\]
Thus $\Gamma_1(\theta)$ is the expected favorable deviation available to
Sender 1, and $\Gamma_2(\theta)$ is the corresponding favorable deviation
available to Sender 2.

The finite-pool result is exact, but its cutoff depends on the realized
record. Large pools reveal the systematic force created by attention. We let
the anchor vanish, require no single observation to carry appreciable weight,
and hold fixed Sender 1's limiting attention share. Formally,
$\eta_{n}\to0$, $\max_{i}v_{ni}\to0$, and $r_{n}\to r\in(0,1)$.

\begin{proposition}[Large-pool strategic narrative]
\label{thm:largepool}
Under Assumption~\ref{ass:evidence}, $H(\cdot\,;r,\theta)$ is strictly
decreasing with a unique interior zero $c^{*}(r,\theta)$, and
$c_{n}\xrightarrow{p}c^{*}(r,\theta)$. The disclosed fractions of
Sender 1's and Sender 2's pools converge in probability to
$1-F_{1,\theta}(c^{*})\in(0,1)$ and $F_{2,\theta}(c^{*})\in(0,1)$.
Moreover $c^{*}$ is continuously differentiable in $r$ and
\begin{equation}
 \frac{\partial c^{*}(r,\theta)}{\partial r}
 =
 \frac{\E_{\theta}[(X_1-c^{*})^{+}]
       +\E_{\theta}[(c^{*}-X_2)^{+}]}
      {r\,\Prob_{\theta}(X_1>c^{*})
       +(1-r)\,\Prob_{\theta}(X_2<c^{*})}
 >0.
 \label{eq:cs}
\end{equation}
\end{proposition}

The limit $c^*(r,\theta)$ is the public narrative generated by the two
selection rules. Sender 1 shows observations above it and
Sender 2 shows observations below it; both continue to suppress a positive
fraction of their pools. Equation~\eqref{eq:cs} gives the strategic response
to attention. Its numerator is the total directional surplus available on
the two sides of the cutoff. Its denominator is the mass of evidence visible
at the margin. Attention therefore has more leverage when the public record
is thin.

Implementing the correct action does not reconstruct the missing record. The
Appendix shows that, under a mild identification condition,
equilibrium disclosure is not sufficient for the state: two full
pools can generate the same visible evidence while implying different
likelihoods for what was withheld. Countervailing attention restores the
decision, not the information that strategic selection removed.

\section{Implementing the Complete-Evidence Action}
\label{sec:countervailing}
Before confronting the
platform's information constraint, ask which attention share would implement
the truth if the state were known. Because senders change their disclosure
standards when attention changes, the answer balances their favorable tail
moments rather than their realized speech.

\begin{theorem}[State-contingent optimal attention]
\label{thm:design}
Suppose the platform knows the state $\theta$ when it assigns attention.
There is a unique share that implements the receiver's correct action in the
large-pool limit. It is characterized by
\[
  r\,\Gamma_1(\theta)=(1-r)\,\Gamma_2(\theta),
  \qquad\text{i.e.}\qquad
  r=r^{*}(\theta)
   =\frac{\Gamma_2(\theta)}{\Gamma_1(\theta)+\Gamma_2(\theta)} .
\]
At this share, and only at this share, the equilibrium and
complete-disclosure actions have the same limit $\theta$. Consequently,
$r^{*}(\theta)$ uniquely minimizes the receiver's limiting loss
$(c^{*}(r,\theta)-\theta)^2$. Whenever
$\Gamma_2(\theta)/\Gamma_1(\theta)$ is independent of $\theta$,
the implementing share is state independent and can be set ex ante.
\end{theorem}

Equal attention is optimal only when $\Gamma_1(\theta)=\Gamma_2(\theta)$, which
holds under a common evidence distribution however skewed
(Corollary~\ref{thm:countervailing}), but fails when the favorable-tail moments differ: the camp with the
larger such moment must
receive \emph{less} than half of attention. Equal numbers of sources,
equal realized speaker counts, and equal visible mass are generally the
wrong design targets---the Example in the Introduction shows realized
voice balance failing exactly where attention balance is exact.

To isolate the offsetting force, first suppose that the camps draw
from the same evidence distribution. This benchmark removes differences in
evidence quality: any asymmetry in the public narrative is generated by
attention and endogenous disclosure alone. Formally,
$F_{1,\theta}=F_{2,\theta}=F_{\theta}$. We write $c^{*}(r)$ for
$c^{*}(r,\theta)$, $q(c)=1-F_{\theta}(c)$, and
$M_{0}(\theta)=\E_{\theta}(X-\theta)^{+}=\E_{\theta}(\theta-X)^{+}>0$,
where the equality of the two truncated moments is mean-unbiasedness.

\begin{corollary}[Countervailing selection]
\label{thm:countervailing}
Under a common evidence distribution,
\[
  c^{*}(1/2)=\theta,
  \qquad
  c^{*}(r)>\theta\ \text{ if }r>1/2,
  \qquad
  c^{*}(r)<\theta\ \text{ if }r<1/2 .
\]
Equal attention uniquely gives the equilibrium and complete-disclosure
actions the same large-pool limit, even though both senders withhold
nonvanishing shares and the receiver does not infer from suppression.
\end{corollary}

At equal attention the senders select opposite tails around the truth.
Their disclosed masses need not be equal, and the distribution need not be
symmetric. Mean-unbiasedness is enough: the expected upward deviation from
the mean equals the expected downward deviation. Equal attention gives
those deviations equal force. Moving attention away from one half changes
both a camp's weight and the common disclosure standard, producing a
strictly biased narrative.

Under a common evidence law, the limiting narrative is the
$r$-expectile of that law. We use this familiar statistical identity only as
a benchmark: the economic content is that strategic tail disclosure
endogenously generates the expectile and that attention chooses its index.
At equal attention, the two weighted visible tails also fit together into the
uncensored marginal distribution, although the identities of the suppressed
items remain missing. The Appendix states these comparisons precisely.

\Needspace{17\baselineskip}
\subsection{The representation paradox}

Return to possibly different camp laws. Fix the state and abbreviate
$c=c^*(r,\theta)$, $p_1(r)=1-F_{1,\theta}(c)$, and
$p_2(r)=F_{2,\theta}(c)$. Let $\mu_1(r)=\E_\theta[X_1\mid X_1>c]$
and $\mu_2(r)=\E_\theta[X_2\mid X_2<c]$ be the disclosed conditional means.

\begin{proposition}[More attention, less speech, more influence]
\label{prop:paradox}
Under Assumption~\ref{ass:evidence}, for every $r\in(0,1)$,
\[
 c_r>0,\qquad p_1'=-f_{1,\theta}(c)c_r<0,\qquad
 p_2'=f_{2,\theta}(c)c_r>0,
\]
and
\[
 \mu_1'=\frac{f_{1,\theta}(c)}{p_1}(\mu_1-c)c_r>0,
 \qquad
 \mu_2'=\frac{f_{2,\theta}(c)}{p_2}(c-\mu_2)c_r>0.
\]
Increasing Sender 1's exposure raises the action, reduces its per-item
disclosure rate, and raises its disclosed conditional mean. Sender 2
discloses more frequently and its disclosed conditional mean also rises.
\end{proposition}

The favored advocate becomes more selective while its opponent becomes less
selective. The public action nevertheless moves toward the favored camp,
as Proposition~\ref{thm:largepool} establishes. The decline in $p_1$ concerns
the fraction of eligible items disclosed; it does not imply a decline in
the visible exposure mass $rp_1$. Nor does this result require the two camps
to share an evidence distribution.

\paragraph{Empirical implication.}
\label{rem:empirical-content}
Let $z$ be an experimental or conditionally exogenous shifter of the upward
attention share, with $r'(z)>0$, and hold fixed the state and both
evidence laws. Define the limiting action
$a(z)=c^*(r(z),\theta)$, the per-item disclosure rates
$s_1(z)=p_1(r(z))$ and $s_2(z)=p_2(r(z))$, and the mean value of disclosed
upward evidence
\[
 e_1(z)=\E_\theta[X_1\mid X_1>c^*(r(z),\theta)].
\]
Proposition~\ref{prop:paradox} implies the joint sign restrictions
\[
 \frac{da}{dz}>0,\qquad
 \frac{ds_1}{dz}<0,\qquad
 \frac{ds_2}{dz}>0,\qquad
 \frac{de_1}{dz}>0.
\]
The monotonicities also survive averaging these state-specific quantities
using weights on states that are fixed across values of $z$, with evidence
technologies held fixed. Such averages differ from pooling all disclosed
observations, whose state composition can change with disclosure.
A mechanical benchmark with
attention-invariant disclosure rules instead has
$ds_1/dz=ds_2/dz=0$.

These restrictions are not point identification of strategic curation.
An editorial rule that responds directly to $z$, or an attention change
that alters the state, evidence composition, or the receiver's processing
rule, can
mimic them. Moreover, estimating $s_1$ and $s_2$ requires the platform to
observe the eligible item pool and each disclosure decision, not merely the
selected public record. Thus failure of a sign restriction rejects the
joint maintained model under a valid attention intervention; satisfaction
of all four restrictions is evidence consistent with the mechanism, not a
unique test of it.

Under a common evidence law, the mechanism has two useful further
implications. First, endogenous
movement of the disclosure standard amplifies the bias created by imbalanced
attention relative to a benchmark with standards fixed at the truth. Second,
as one camp's attention vanishes, the narrative approaches the favored
boundary. If the common density extends continuously to strictly positive
values at both support endpoints, a small opposing share pulls the narrative
back at a square-root rate.
These results sharpen the mechanism without changing the design problem, so
their formal statements and the uniform example are collected in the Online
Appendix.

\subsection{Processing empty slots: a robustness boundary}
\label{sec:missing-processing}

The baseline receiver gives an empty slot zero weight. We now ask whether the
attention-to-disclosure feedback depends on complete neglect of missingness.
Suppose every assigned slot is publicly known and the receiver imputes a
suppressed item at the fixed benchmark $a_0$, but gives that imputation only
a fraction $\lambda\in[0,1]$ of the weight of displayed evidence. Equivalently,
an institution fills the fraction $\lambda$ of every unused slot with neutral
content valued at $a_0$. Given the displayed set $S$, the induced action is
\begin{equation}
 a_{n,\lambda}(S)
 =
 \frac{\eta_n a_0+\sum_{i\in S}v_{ni}X_i
       +\lambda a_0\sum_{i\notin S}v_{ni}}
      {\eta_n+\sum_{i\in S}v_{ni}
       +\lambda\sum_{i\notin S}v_{ni}}.
 \label{eq:partial-action}
\end{equation}
Thus $\lambda=0$ is the baseline and $\lambda=1$ gives every empty slot the
same weight as a displayed item.

The sender's marginal comparison remains one-dimensional. When the current
action is $y$, replacing an imputed empty slot by a displayed value $x$ raises
the action if and only if
\[
 x>t\equiv\lambda a_0+(1-\lambda)y.
\]
The common disclosure bar is therefore $t$. With partial imputation it lies
between the receiver's benchmark and action; only in the baseline does the
bar coincide with the action itself.

\begin{proposition}[Robustness to processing empty slots]
\label{prop:missing-processing}
Fix $\lambda\in[0,1]$ and maintain the baseline sender preferences.

\smallskip
\noindent\textbf{(i) Partial imputation.}
For every $\lambda<1$ and every finite realized pool, a pure-strategy
equilibrium exists, all pure equilibria induce the same action, and---apart
from cutoff ties---Sender 1 discloses observations above a unique common bar
$t_{n,\lambda}$ while Sender 2 discloses observations below it. Under the
large-pool conditions of Proposition~\ref{thm:largepool}, both the limiting
action and limiting disclosure bar are strictly increasing in Sender 1's
attention share. More upward attention therefore makes Sender 1 disclose
less often and more extremely, makes Sender 2 disclose more often, and moves
the action upward.

\smallskip
\noindent\textbf{(ii) Full-imputation boundary.}
As $\lambda$ approaches one, the limiting disclosure bar approaches $a_0$
and its response to attention vanishes. At $\lambda=1$, the unique bar is
$a_0$ for every attention share. The action remains strictly increasing in
Sender 1's share, but disclosure frequencies no longer respond to attention.
\end{proposition}

The result identifies a sharp boundary. The representation paradox survives
whenever displayed evidence receives strictly more weight than an imputed
empty slot; complete neglect of missingness is not required. Full imputation
removes the endogenous participation response, leaving only the mechanical
effect of changing attention weights. Partial imputation also gives the
benchmark persistent influence, so an implementing attention share may fail
to exist when $a_0$ is too far from the state. The Appendix gives the
exact finite-pool cutoff,
comparative statics, implementation formula, and proof.

This is robustness to processing publicly observed empty slots, not a model
of correctly specified inference from strategic silence. The parameter
$\lambda$ measures the relative weight placed on a known missing observation.
A receiver who uses the equilibrium selection likelihood faces a different
problem, because disclosure then changes beliefs as well as the displayed
average.

\section{Attention Before the State Is Known}
\label{sec:design}

The conditional balancing share is a useful benchmark, not generally a
feasible policy. We now impose the paper's actual timing: attention is
assigned before the state is learned. When state-contingent shares disagree,
the platform must trade errors across states rather than eliminate them
state by state.

\subsection{Ex ante attention under an unknown state}
\label{sec:exante}

Let the state be drawn from a prior $\pi$ on a compact set
$\Theta\subset\mathrm{int}\,\mathcal X$, and let the platform choose one
share from a compact interval
$\mathcal R=[\underline r,\overline r]\subset(0,1)$ before observing
$\theta$. Define its ex ante large-pool loss by
\begin{equation}
  \mathcal L(r)
  =\int_{\Theta}\big[c^{*}(r,\theta)-\theta\big]^{2}\,d\pi(\theta).
  \label{eq:exanteloss}
\end{equation}

\begin{assumption}[Design regularity]
\label{ass:design}
The functions $\Gamma_1(\theta)$ and $\Gamma_2(\theta)$ are continuous on
$\Theta$; the tail probabilities and truncated moments entering
$H(c;r,\theta)$ are jointly continuous in $(c,r,\theta)$; and
$r^{*}(\theta)\in\mathrm{int}\,\mathcal R$ for every $\theta\in\Theta$.
\end{assumption}

These conditions make $c^{*}$ and its derivative in \eqref{eq:cs}
continuous on the compact design domain. They are satisfied by smooth
parametric families whose densities vary continuously with the state.

\begin{theorem}[Optimal attention before the state is known]
\label{thm:exante}
Under Assumptions~\ref{ass:evidence} and~\ref{ass:design}:

\smallskip
\noindent\textbf{(i) When one policy implements every state.}
An ex ante optimum exists. The minimum loss is zero if and only if
$r^{*}(\theta)=\bar r$ for $\pi$-almost every $\theta$ and some constant
$\bar r\in\mathcal R$; then $\bar r$ is the unique optimum.

\smallskip
\noindent\textbf{(ii) When the platform must compromise across states.}
Let
\[
  r_{L}=\operatorname*{ess\,inf}_{\theta\sim\pi}r^{*}(\theta),
  \qquad
  r_{U}=\operatorname*{ess\,sup}_{\theta\sim\pi}r^{*}(\theta).
\]
Every ex ante optimum lies in $[r_{L},r_{U}]$; if $r^{*}(\theta)$ is not
constant $\pi$-almost surely, every optimum lies strictly in $(r_L,r_U)$.

\smallskip
\noindent\textbf{(iii) How the optimal compromise is determined.}
Every interior optimum $r^{EA}$ satisfies
\begin{equation}
  \int_{\Theta}
  \big[c^{*}(r^{EA},\theta)-\theta\big]
  \frac{\partial c^{*}(r^{EA},\theta)}{\partial r}
  \,d\pi(\theta)=0.
  \label{eq:exanteFOC}
\end{equation}
If $\mathcal L$ is strictly convex, this equation identifies the unique
ex ante optimum.

\smallskip
\noindent\textbf{(iv) The value of information to the platform.}
If the platform observes a signal $Z$ before choosing attention, it solves
the same problem under $\pi(\cdot\mid Z)$, signal by signal. A
Blackwell-more-informative signal weakly lowers attainable ex ante loss.
\end{theorem}

The theorem separates conditional correction from implementable policy.
It gives a simple test: one fixed share removes limiting bias on a set
of states if and only if the moment ratio
$\Gamma_2(\theta)/\Gamma_1(\theta)$ is constant on that set. Common within-state evidence laws
pass this test even when the common law changes with the state, because
$\Gamma_2(\theta)/\Gamma_1(\theta)=1$. When this ratio varies across
states, no fixed share can eliminate large-pool bias in every state. This is a characterization of what every optimum must satisfy, not a
general convexity result. Strict convexity is an additional sufficient
condition, not a consequence established by
Assumptions~\ref{ass:evidence} and~\ref{ass:design}. The next subsection
supplies primitive conditions under which the ex ante optimum is unique and
solves it exactly. For small dispersion of $r^{*}(\theta)$, a first-order expansion
around the state-contingent shares gives the useful approximation
\[
  r^{EA}
  \simeq
  \frac{\int \kappa(\theta)^{2}r^{*}(\theta)\,d\pi(\theta)}
       {\int \kappa(\theta)^{2}\,d\pi(\theta)},
  \qquad
  \kappa(\theta)
  =\left.\frac{\partial c^{*}(r,\theta)}{\partial r}
   \right|_{r=r^{*}(\theta)} .
\]
States in which the narrative is more sensitive to attention receive more
weight in the ex ante compromise.

\input{new-design-results}

\input{partial-reallocation}

\subsection{A fully solved two-state design problem}
\label{sec:two-state}

The general theorem deliberately does not assume that the ex ante loss is
strictly convex. We now give primitive evidence technologies for which the
state-contingent policies conflict but the second-best problem nevertheless
has a unique solution.

There are two states, centered around $\theta_0$, with $\Delta>0$:
$\theta_L=\theta_0-\Delta$ and $\theta_H=\theta_0+\Delta$, with prior
probability $p$ on the high state. In either state, each camp observes a
mean-unbiased uniform signal. The asymmetry is in dispersion. Sender 1 has
the more dispersed evidence technology in the low state, while Sender 2 has
the more dispersed technology in the high state. Thus the camp that the
platform should restrain changes with the state.

Formally, let $\mathcal X=[\theta_0-M,\theta_0+M]$, fix $s>0$ and
$1<k<\sqrt3$ such that $\Delta+k^2s<M$, and write
\[
 X_{1,j}=\theta_j+\sigma_{1,j}Z_1,
 \qquad
 X_{2,j}=\theta_j+\sigma_{2,j}Z_2,
\]
where $Z_1$ and $Z_2$ are independent draws from $U[-1,1]$. The scale pairs
reverse across states:
\[
 (\sigma_{1,L},\sigma_{2,L})=(k^2s,s),
 \qquad
 (\sigma_{1,H},\sigma_{2,H})=(s,k^2s).
\]
The parameter $k$ measures the cross-state reversal in relative dispersion.

\begin{proposition}[Unique ex ante attention in a two-state economy]
\label{prop:primitiveexante}
In the two-state economy above, suppose $p\in(0,1)$ and $\mathcal R$ contains
$[1/(1+k^2),k^2/(1+k^2)]$ in its interior, and define the attention-odds
coordinate $t=\sqrt{r/(1-r)}$. Then:

\smallskip
\noindent\textbf{(i) The state-contingent targets conflict.}
For every $r\in[1/(1+k^2),k^2/(1+k^2)]$,
\begin{align}
 c^{*}(r,\theta_L)-\theta_L
 &=d_L(t)=sk\frac{kt-1}{t+k},
 & r^{*}(\theta_L)&=\frac{1}{1+k^2},
 \label{eq:primitiveDL}\\
 c^{*}(r,\theta_H)-\theta_H
 &=d_H(t)=sk\frac{t-k}{kt+1},
 & r^{*}(\theta_H)&=\frac{k^2}{1+k^2}.
 \label{eq:primitiveDH}
\end{align}
No fixed attention share implements both states.

\smallskip
\noindent\textbf{(ii) The ex ante second best is unique.}
The platform's unique optimum is
$r^{EA}_p=t_p^2/(1+t_p^2)$, where $t_p\in(1/k,k)$ is the unique solution to
\begin{equation}
 \frac{p}{1-p}
 =R_k(t_p)
 \equiv
 \frac{(kt_p-1)(kt_p+1)^3}
      {(k-t_p)(t_p+k)^3}.
 \label{eq:primitiveFOC}
\end{equation}
Both $t_p$ and $r^{EA}_p$ are strictly increasing in the prior probability
of the high state. As $p$ approaches zero or one, the ex ante policy
approaches the corresponding state-contingent share; under the symmetric
prior, $r^{EA}_{1/2}=1/2$.

\smallskip
\noindent\textbf{(iii) Information has strictly positive value.}
A platform that observes the state uses the two state-contingent shares and
attains zero limiting loss. Without that information, its loss is strictly
positive. At $p=1/2$, the value of perfect state information is
\begin{equation}
 \mathcal L(1/2)
 =s^2k^2\left(\frac{k-1}{k+1}\right)^2,
 \label{eq:primitiveVOI}
\end{equation}
which is strictly increasing in the dispersion reversal $k$.

\pagebreak[3]
\smallskip
\noindent\textbf{(iv) The conclusions are robust under the maintained evidence assumptions.}
Fix any compact subset $K$ of the interior of the stated parameter domain.
There is a sequence of smooth densities, strictly positive on $\mathcal X$
and satisfying Assumption~\ref{ass:evidence}, whose induced truncated-moment
maps converge in $C^{2}$, on a common neighborhood of the relevant
equilibrium paths, to those of the uniform economy. For every sufficiently
close member of this sequence, the ex ante optimum is unique. Uniformly on
$K$, the state-contingent shares, equilibrium actions, unique ex ante
optima, and attainable losses converge to the expressions above. Hence the
strict policy ordering and positive value of information persist in a
$C^{2}$-open set of smooth, full-support economies satisfying the maintained
assumptions, where the topology is on the truncated-moment maps over that
common neighborhood of the equilibrium paths.
\end{proposition}

The proof is in the Appendix. Part (iv) connects
the transparent uniform calculation to the maintained evidence assumptions;
$C^2$ control is used to preserve uniqueness, not merely convergence of the
optimal value.

The function $R_k$ maps the attention odds into the prior odds. It rises
from zero to infinity on $(1/k,k)$, which is why the first-order condition
selects one and only one compromise. The policy moves toward the high-state
target as that state becomes more likely, but it reaches the target only when
the other state vanishes from the prior.

\paragraph{A policy illustration.}
Let $\theta_0=5$, $\Delta=1$, $s=1$, $k=3/2$, and $p=0.8$.
The true effect is either $4$ or $6$. The upward camp's estimates are more
dispersed in the low state; the downward camp's are more dispersed in the
high state. The conditional implementing shares are $4/13\simeq0.308$ and
$9/13\simeq0.692$. Table~\ref{tab:two-state-policy} compares equal exposure
with the optimal fixed policy.

\begin{table}[htbp]
\centering
\caption{Exposure chosen before the state is known: uniform benchmark}
\label{tab:two-state-policy}
\small
\begin{tabular}{lccc}
\hline
Upward exposure share & Action at $\theta=4$ & Action at $\theta=6$ & Expected loss\\
\hline
Equal: $0.5000$ & $4.300$ & $5.700$ & $0.0900$\\
Optimal: $0.6098$ & $4.477$ & $5.870$ & $0.0592$\\
\hline
\end{tabular}
\end{table}

The optimal allocation accepts more error in the less likely state to reduce
error in the more likely state, lowering expected squared loss by about
$34.3\%$. More exposure makes the upward advocate disclose a smaller fraction
of its eligible studies in either state. The prior probability $0.8$ is
the chance of the high state, not the fraction of favorable observations or
a measure of a camp's credibility. The policy tilt depends on the specified
reversal in relative dispersion; under common camp laws, equal exposure
would implement either state regardless of the prior. These numbers describe
the uniform large-pool prototype with a vanishing anchor. The smooth
full-support economies in part (iv) approach this comparison.

\subsection{Finite evidence pools}
\label{sec:finitesample}

Moment balance removes limiting bias, but a finite selected record is not
equivalent to complete disclosure. The equilibrium cutoff is an implicit
estimator: sampling error changes both the tail moments and the observations
classified as visible. A risk-minimizing platform must therefore trade three
forces---endogenous cutoff bias, the anchor, and sampling variance.

For this section the predetermined within-camp schedule is uniform. Let
$n_k=|I_k|$, let $n=n_1+n_2$, and suppose
$n_k/n\to\alpha_k\in(0,1)$ for $k\in\{1,2\}$, with
$\alpha_1+\alpha_2=1$. At share $r$ the item weights are
\[
 v_{ni}=\frac{r}{n_1}\quad(i\in I_1),
 \qquad
 v_{ni}=\frac{1-r}{n_2}\quad(i\in I_2).
\]
The fixed anchor may remain relevant at the same order as finite-pool bias:
\begin{equation}
 n\eta_n\longrightarrow\nu\in[0,\infty).
 \label{eq:anchorrate}
\end{equation}
Restrict the platform's share to a compact interval
$\mathcal R_0\subset(0,1)$, and let $\mu$ be either a known state or a
compactly supported prior. Its exact and large-pool risks are
\begin{align*}
 \mathcal W_{n,\mu}(r)
 &=\int \E_\theta[(c_n(r)-\theta)^2]\,d\mu(\theta),\\
 \mathcal L_\mu(r)
 &=\int[c^*(r,\theta)-\theta]^2\,d\mu(\theta).
\end{align*}
The theorem below will require a unique, well-curved interior minimizer of
$\mathcal L_\mu$; denote that candidate by $r_\mu$. Before stating the
theorem, three economically interpretable objects are enough. Fix
$(r,\theta)$ and abbreviate $c^*=c^*(r,\theta)$. Let
$p_1=\Prob_\theta(X_1>c^*)$ and
$p_2=\Prob_\theta(X_2<c^*)$ be the fractions of the two pools that are
visible at the limiting cutoff. Their attention-weighted sum,
\[
 D^*(r,\theta)=rp_1+(1-r)p_2,
\]
is the effective slope of the balance equation: a small $D^*$ makes the
equilibrium action sensitive to sampling error. In particular, its
asymptotic variance is
\begin{equation}
 V^C(r,\theta)
 =\frac{1}{D^*(r,\theta)^2}
 \left\{
 \frac{r^2}{\alpha_1}\operatorname{Var}_\theta((X_1-c^*)^+)
 +\frac{(1-r)^2}{\alpha_2}\operatorname{Var}_\theta((c^*-X_2)^+)
 \right\}.
 \label{eq:curatedvariance-main}
\end{equation}

Sampling also shifts the mean cutoff. Under the regularity conditions below,
write $B^{C,\nu}(r,\theta)$ for the coefficient in
\[
 \E_\theta[c_n(r)-c^*(r,\theta)]
 =\frac{B^{C,\nu}(r,\theta)}{n}+o(n^{-1}).
\]
This bias has two distinct sources:
\begin{equation}
 B^{C,\nu}(r,\theta)
 =B^{C,0}(r,\theta)
  +\frac{\nu[a_0-c^*(r,\theta)]}{D^*(r,\theta)}.
 \label{eq:biasdecomposition-main}
\end{equation}
The first term is generated endogenously because sampling error changes both
the empirical balance equation and which observations cross the cutoff. The
second is the residual pull of the receiver's fixed benchmark. The Online
Appendix gives the primitive formula for
$B^{C,0}$. Finally, with $d(r,\theta)=c^*(r,\theta)-\theta$, define
\[
 Q_\nu(r,\theta)
 =V^C(r,\theta)+2d(r,\theta)B^{C,\nu}(r,\theta),
 \qquad
 \mathcal Q_{\mu,\nu}(r)=\int Q_\nu(r,\theta)\,d\mu(\theta).
\]
Thus $\mathcal Q_{\mu,\nu}$ is the first finite-pool adjustment to the
platform's large-pool loss.

\input{precision-cost}

The following conditions keep the limiting cutoff away from support
boundaries, ensure that a nonnegligible record remains visible, and permit us
to differentiate the exact risk as observations cross the moving cutoff.

\begin{assumption}[Finite-pool regularity]
\label{ass:finitepool}
Let $\mathcal K=\mathcal R_0\times\operatorname{supp}\mu$ and fix an
interior candidate $r_\mu$.

\smallskip
\noindent\textbf{(i) Risk expansion.}
On $\mathcal K$, the limiting cutoffs lie in a common compact subset of the
interior of $\mathcal X$, the visible attention mass $D^*$ is uniformly
bounded away from zero, and the two evidence densities are locally
continuously differentiable with uniform bounds. The camp fractions satisfy
$n_k/n\to\alpha_k\in(0,1)$ and the anchor satisfies~\eqref{eq:anchorrate}.

\smallskip
\noindent\textbf{(ii) Policy derivative.}
On a compact policy neighborhood $\mathcal N_\mu$ of $r_\mu$, the densities
are locally three times continuously differentiable in the evidence
realization, jointly continuous in the state, and uniformly bounded together
with those derivatives. In addition,
$\eta_n^{-1}=O(n^\rho)$ for some finite $\rho$.
\end{assumption}

The Appendix records the corresponding uniform neighborhoods and
bounds used in the proofs. Economically, the
assumption rules out boundary cutoffs, vanishing visible records, and
excessively rough density paths. It is not needed for the finite disclosure
equilibrium, large-pool moment balance, or ex ante design theorem.

The Appendix derives the sampling distribution and the uniform
expansion
\begin{equation}
 \mathcal W_{n,\mu}(r)
 =\mathcal L_\mu(r)+\frac{\mathcal Q_{\mu,\nu}(r)}{n}+o(n^{-1}),
 \label{eq:integratedrisk-main}
\end{equation}
where $\mathcal Q_{\mu,\nu}$ collects the three forces defined above. Let
$\widetilde r_{n,\mu}$ minimize the displayed second-order criterion
$\mathcal L_\mu+\mathcal Q_{\mu,\nu}/n$. A value expansion alone does not
imply that the optimizer of exact risk has the same displacement: two
policies separated by order $1/n$ differ in welfare only at order $1/n^2$.
The next theorem establishes that the exact and approximate policies do in
fact agree to first order.

\begin{theorem}[Exact finite-pool attention correction]
\label{thm:exactpolicy}
Under Assumptions~\ref{ass:evidence} and~\ref{ass:finitepool}, suppose the
large-pool loss $\mathcal L_\mu$ has a unique interior minimizer $r_\mu$ with
$\mathcal L_\mu''(r_\mu)>0$, and $\mathcal L_\mu$ and
$\mathcal Q_{\mu,\nu}$ are continuous on $\mathcal R_0$, and suppose
$\mathcal L_\mu$ and $\mathcal Q_{\mu,\nu}$ are twice continuously
differentiable on $\mathcal N_\mu$. Then every sequence of exact risk-minimizing
shares
$\widehat r_{n,\mu}\in\arg\min_{r\in\mathcal R_0}
\mathcal W_{n,\mu}(r)$ satisfies
\begin{equation}
 n(\widetilde r_{n,\mu}-r_\mu),\ 
 n(\widehat r_{n,\mu}-r_\mu)
 \longrightarrow
 -\frac{\mathcal Q_{\mu,\nu}'(r_\mu)}
        {\mathcal L_\mu''(r_\mu)},
 \qquad
 n(\widehat r_{n,\mu}-\widetilde r_{n,\mu})\longrightarrow0.
 \label{eq:exactpolicyshift}
\end{equation}
The conclusion does not require eventual uniqueness: every exact optimizer
has the same first-order displacement.

If the state $\theta$ is known to the platform, then $r_\mu=r^*(\theta)$.
Let
$\kappa(\theta)=\partial_r c^*(r^*(\theta),\theta)>0$ denote the sensitivity
of the limiting narrative to attention at that share. The correction
separates into three terms:
\begin{equation}
 n[\widehat r_n(\theta)-r^*(\theta)]
 \longrightarrow
 \underbrace{-\frac{B^{C,0}(r^*,\theta)}{\kappa(\theta)}}_{
   \text{selection-induced cutoff bias}}
 -\underbrace{\frac{\nu(a_0-\theta)}
  {\kappa(\theta)D^*(r^*,\theta)}}_{
   \text{fixed receiver benchmark}}
 -\underbrace{\frac{\partial_r V^C(r^*,\theta)}
                         {2\kappa(\theta)^2}}_{
   \text{sampling variance}}.
 \label{eq:exactbiasvarianceshift}
\end{equation}
\end{theorem}

The first term offsets the bias created by estimating a cutoff from the same
sample that the cutoff selects. The second offsets the receiver's fixed
benchmark: a benchmark below the state calls for more upward attention, and
a benchmark above the state for less. The third reallocates attention toward
policies with lower sampling variance. These forces need not point in the
same direction, even when the large-pool policy implements the state exactly.

\input{new-finite-result}

The difficulty is not the Taylor expansion itself. Sample cutoff paths are
only piecewise smooth because individual observations enter and leave the
disclosed tails. The Appendix proves the
uniform $C^1$ risk expansion behind the theorem. Its proof localizes
the full and leave-one-out roots and their attention derivatives, then
integrates the moving-boundary contributions explicitly. For the centered
slope remainder, a direct threshold identity reduces the boundary term to
a leave-one-out level effect. Thus the exact-policy conclusion does not
come from formally differentiating the value expansion
in~\eqref{eq:integratedrisk-main}.

The correction can be nonzero even when equal attention implements the
state. In the two-point limit used in the Introduction---values $9$ and $4$
with probabilities $1/5$ and $4/5$, common to both camps, equal camp sizes,
and $\theta=5$---direct calculation at $r^*=1/2$ gives
\[
 \kappa=\frac{16}{5},\qquad
 B^{C,0}=-\frac{24}{25},\qquad
 \partial_rV^C=\frac{3264}{125},
\]
Consequently, the policy correction itself reads
\[
 n(\widetilde r_n-\tfrac12)
 \longrightarrow
 \underbrace{\frac{3}{10}}_{\text{selection bias}}
 -\underbrace{\frac{5\nu}{8}(a_0-5)}_{\text{receiver benchmark}}
 -\underbrace{\frac{51}{40}}_{\text{sampling variance}}
 =-\frac{39}{40}-\frac{5\nu}{8}(a_0-5).
\]
Thus when $\nu=0$ or $a_0=\theta$, the second-order criterion gives the upward
camp slightly less than half of attention despite exact limiting moment
balance. Selection bias by itself calls for slightly more upward attention,
but the variance motive is more than four times as large and points the other
way. An order-$1/n$ benchmark below $5$ offsets that negative correction,
while one above $5$ reinforces it. Arbitrarily close smooth approximations
preserve these strict comparisons; for every such approximation satisfying
Assumption~\ref{ass:finitepool}, Theorem~\ref{thm:exactpolicy} gives the same
displacement for every exact finite-pool optimizer. By contrast, under a
common distribution symmetric around $\theta$ and asymptotically equal camp
sizes, the endogenous-bias, benchmark, and variance motives all vanish at
equal attention when $a_0=\theta$; if $a_0\neq\theta$, only the benchmark
term remains.

\paragraph{Transparent finite-pool benchmark.}
The second-order correction is quantitatively accurate at moderate pool
sizes in the same two-point benchmark.

\begin{figure}[t]
 \centering
 \includegraphics[width=\textwidth]{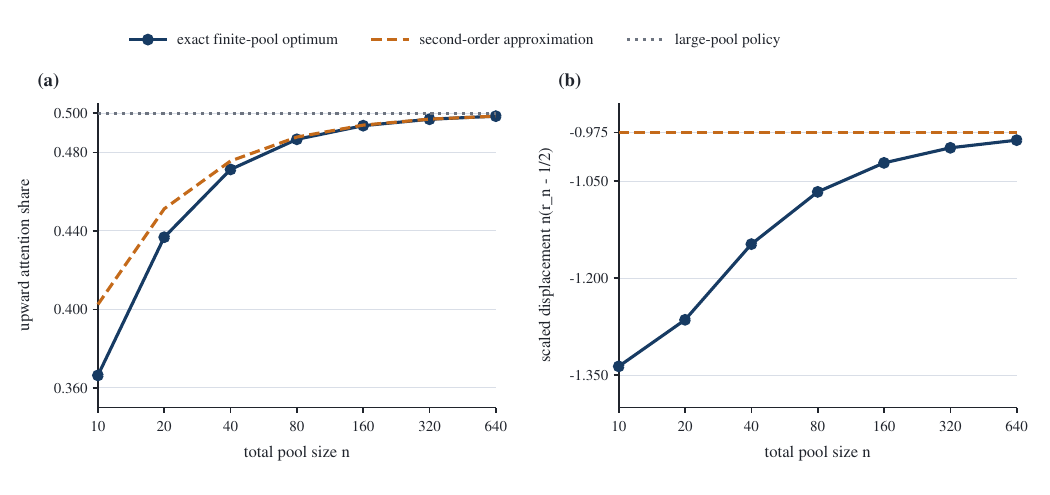}
 \caption{The exact finite-pool risk-minimizing share and its second-order approximation.
 Panel (a) plots the exact risk-minimizing upward attention share, the
 second-order policy $1/2-39/(40n)$, and the large-pool policy $1/2$.
 Panel (b) rescales the exact displacement by $n$ and compares it with its
 limit $-39/40$. Each exact point evaluates the full double-binomial sum;
 there is no simulation error.}
 \label{fig:finite-policy-validation}
\end{figure}

Figure~\ref{fig:finite-policy-validation} evaluates the full
double-binomial risk, rather than a simulation, and compares its numerical
minimizer with the large-pool policy $1/2$ and the second-order policy
$1/2-39/(40n)$. At $n=80$, the exact and second-order shares are $0.48666$
and $0.48781$; at $n=320$, they are $0.496879$ and $0.496953$. The scaled
exact displacement also converges to $-39/40$, as the expansion predicts.
The same finite-sum calculation puts the risk reduction relative to equal
exposure at $2.47\%$ for $n=80$ and $0.57\%$ for $n=320$, consistent
with a declining relative benefit from the finite-pool refinement.
The Appendix derives the exact cutoff and finite-sum risk used in the
figure.

The policy correction refines accuracy within the institution. Its
order-$1/n^2$ risk gain does not remove the order-$1/n$ precision cost in
the common-law, equal-camp benchmark of Proposition~\ref{prop:varcost}.
The Appendix also gives a smooth validation and a heterogeneous-law
example in which the variance ranking against complete disclosure reverses.

\FloatBarrier
\section{Discussion and conclusion}
\label{sec:discussion}

The paper studies a division of authority: an institution commits
presentation weights, and advocates retain the right to withhold evidence.
Their disclosure choices jointly determine the standard against which an
additional observation is favorable. All finite Nash equilibria give the
same action under strictly opposed monotone preferences, making the
consequences of an exposure policy well defined. In the large-pool limit, more upward
exposure raises the action while reducing the upward camp's per-item
disclosure rate and raising the mean of its disclosed evidence.

This equilibrium response connects the disclosure game to the institution's
ex ante problem. Directional tail moments identify the state-specific
implementing share. Variation in that share bounds the loss from using a
common policy, while information identifying the required share suffices
for limiting implementation. Within the common-law, fixed-support benchmark,
binding exposure to individual items preserves the state dependence of
decisions that full within-camp reallocation eliminates. These comparisons
concern the large-pool, vanishing-anchor limit and the specified exposure
rules.

Partial reallocation shows how this mechanism changes between the two
institutions. The advocates' cutoffs separate, and the correcting attention
share depends locally on the frequencies of favorable observations. If
those frequencies vary across states, no single first-order policy
adjustment corrects every state, leaving a positive local welfare cost.
Separately, private learning removes the need for the advocates to know the
state or one another's evidence: under linear payoffs, uniform item
weights, and the stated continuity and compactness conditions, every
Bayesian equilibrium has the same baseline limit.

These conclusions concern limiting accuracy. Strategic selection still
has a precision cost in the common-law benchmark with balanced exposure,
uniform item weights, and asymptotically equal camp sizes. The finite-pool
correction accounts for cutoff bias, the receiver's anchor, and variance;
its welfare gain is at most second order in inverse pool size. The disclosure game,
the institutional comparison, and the finite-pool refinement thus address
distinct parts of the design problem.

\paragraph{Behavioral and institutional scope.}
The Appendix derives the averaging rule as the
quadratic-loss Bayes action under a Gaussian working model in which display
is exogenous. This pins down the receiver's misspecification, but does not
describe correctly specified inference from strategic missingness. If the
receiver uses the equilibrium selection likelihood, a change in disclosure
affects both the visible sample and the likelihood of the selection event;
the tail best responses used here then need not follow.
Section~\ref{sec:missing-processing} shows that the representation paradox
survives operational neutral backfilling and benchmark imputation whenever
empty slots receive less weight than displayed evidence. This extension
assigns weight to observable missingness but does not model inference from
the senders' equilibrium selection likelihoods. The
baseline common-cutoff characterization uses nontransferable item weights.
Full neutral backfill
preserves that feature, whereas full within-camp reallocation produces
extremal rather than tail-moment selection, as shown in the Appendix.

\paragraph{Empirical and network scope.}
The comparative statics give joint sign restrictions under exogenous
variation in assigned attention: the favored camp speaks less often but uses
more extreme evidence even as its influence rises. Testing these predictions
would require attention to be assigned before content is realized, a stable
state and evidence technology, and data on the eligible evidence pool rather
than only the public record. The Appendix maps diffuse deterministic
attention rows into static listener actions. It does not model repeated
learning, endogenous links, or homophily.

\paragraph{Remaining limitations.}
The analysis assumes verifiable, exclusively controlled evidence, opposed
monotone preferences, and a benevolent platform. Private learning is covered
when each camp's growing unbiased sample identifies the state and the
evidence laws vary continuously. We do not combine that extension with
reallocation, nor characterize private-pool equilibria exactly at finite
sample sizes. We also use squared loss for one receiver and leave platform
engagement motives outside the model. Within
this scope, the broader distinction is simple: institutions often control
access before content, and advocates adjust content to the access they
receive. A policy can therefore balance voices and still unbalance evidence.

\section*{Code availability}
Replication code for the numerical calculations and figures is provided in
the ancillary files.

\section*{Declaration of competing interest}
The authors declare that they have no known competing financial interests or
personal relationships that could have appeared to influence the work reported
in this article.

\section*{Declaration of generative AI use}
During the preparation of this work, the authors used OpenAI Codex
(GPT-6, September 2026) to assist
with manuscript organization, exposition, literature searches, LaTeX
preparation, development and checking of mathematical derivations, and
numerical code and calculations.
After using this tool, the authors reviewed and edited the content as needed and
take full responsibility for the content of the article.

%% file: new-design-results.tex
\paragraph{The value of committing exposure to individual items.}
The nontransferability of exposure is economically consequential. Consider
instead full within-camp reallocation: a camp's entire attention share is
spread across the evidence it submits, with an empty submission imputed at
$a_0$, as in \eqref{eq:reallocated-action} in the Appendix. Each
sender then concentrates its submission on its most favorable observation,
or submits nothing if the fixed benchmark is more favorable. On a common
fixed support, the resulting limiting action is independent of the
state. This yields a sharp comparison when the two camps have the same
evidence law within each state.

\begin{corollary}[The welfare value of nontransferable exposure]
\label{cor:exposure-regimes}
Maintain the conditions of Theorem~\ref{thm:exante} and the large-pool
regime with diffuse item weights and a vanishing anchor. Suppose
$F_{1,\theta}=F_{2,\theta}=F_\theta$ for every $\theta$, with the common,
state-independent support
$\mathcal X=[\underline x,\overline x]$. Compare the two institutions on the
same feasible interval $\mathcal R$, which contains $1/2$. Let
$\mathcal L^R(r)$ be the expected squared error of the limiting action under
full within-camp reallocation, and set
\[
 J=\{\underline x+r(\overline x-\underline x):r\in\mathcal R\}.
\]
Then the optimal limiting losses satisfy
\begin{equation}
 \min_{r\in\mathcal R}\mathcal L(r)=0,
 \qquad
 \min_{r\in\mathcal R}\mathcal L^R(r)
 =\operatorname{Var}_\pi(\theta)
  +\operatorname{dist}\bigl(\E_\pi\theta,J\bigr)^2.
 \label{eq:exposure-regime-losses}
\end{equation}
Thus nontransferable exposure yields strictly lower optimal limiting loss
whenever the prior is nondegenerate.
\end{corollary}

\begin{proof}
Corollary~\ref{thm:countervailing} gives $c^*(1/2,\theta)=\theta$ in every
state. Under full reallocation, Proposition~\ref{prop:fungible} gives
$c^R(r)=\underline x+r(\overline x-\underline x)$. For any constant action
$a\in J$,
\[
 \E_\pi[(a-\theta)^2]
 =\operatorname{Var}_\pi(\theta)+(a-\E_\pi\theta)^2.
\]
Minimizing over $a\in J$ proves the claim.
\end{proof}

Full reallocation pushes disclosed observations toward the two fixed
support endpoints. Binding exposure to individual items preserves the
response of the limiting action to the state. The comparison concerns
optimal limiting losses under these two institutions; it does not rank
arbitrary reallocation rules or establish a finite-pool welfare ordering.
The fixed-support restriction matters: if the support endpoints vary with
the state, their mixture can itself respond to the state.

\paragraph{Policy disagreement and useful information.}
Within the nontransferable institution, the obstacle to a common policy is
dispersion in the state-contingent balancing share. Write
$T=r^*(\theta)$. This policy target, rather than the state itself, is the
information that the designer needs to eliminate limiting error.

\begin{corollary}[The cost of policy disagreement]
\label{cor:policy-dispersion}
Under the conditions of Theorem~\ref{thm:exante}, there exist constants
$0<m\leq M<\infty$ such that
\[
 m\leq\partial_r c^*(r,\theta)\leq M
 \quad\text{for every }(r,\theta)\in\mathcal R\times\Theta.
\]
For every feasible share,
\begin{equation}
 m^2\E_\pi[(r-T)^2]\leq\mathcal L(r)
 \leq M^2\E_\pi[(r-T)^2].
 \label{eq:policy-disagreement-bound}
\end{equation}
Consequently,
\begin{equation}
 m^2\operatorname{Var}_\pi(T)
 \leq\min_{r\in\mathcal R}\mathcal L(r)
 \leq M^2\operatorname{Var}_\pi(T).
 \label{eq:policy-variance-bound}
\end{equation}
More generally, suppose the platform observes a signal $Z$ before choosing
attention, and the conditional evidence laws given $(\theta,Z)$ remain
$F_{k,\theta}$. Define the optimal remaining limiting loss by
\[
 V(Z)=\inf_g\E\bigl[(c^*(g(Z),\theta)-\theta)^2\bigr],
\]
where the infimum is over measurable $\mathcal R$-valued policies. Then
\begin{equation}
 m^2\E[\operatorname{Var}(T\mid Z)]
 \leq V(Z)
 \leq M^2\E[\operatorname{Var}(T\mid Z)].
 \label{eq:signal-policy-variance-bound}
\end{equation}
In particular, zero loss is attainable if and only if $T$ is determined by
$Z$ almost surely.
\end{corollary}

\begin{proof}
By \eqref{eq:cs} and design regularity, $\partial_r c^*$ is positive and
continuous on the compact design domain, giving the constants $m,M$.
Since $c^*(T,\theta)=\theta$,
\[
 c^*(r,\theta)-\theta
 =\int_T^r\partial_s c^*(s,\theta)\,ds.
\]
Bounding this integral in absolute value and squaring gives
\eqref{eq:policy-disagreement-bound}. The least mean squared distance from
$T$ to a constant is its variance, attained at the feasible share
$\E_\pi T\in\mathcal R$; this gives
\eqref{eq:policy-variance-bound}. Applying the same pointwise bounds to
$g(Z)$ and using
\[
 \E[(g(Z)-T)^2]
 =\E[\operatorname{Var}(T\mid Z)]
  +\E[(g(Z)-\E[T\mid Z])^2]
\]
proves \eqref{eq:signal-policy-variance-bound}: the conditional mean is a
feasible policy. The lower bound makes zero loss impossible unless the
conditional variance vanishes almost surely. When it vanishes, choosing
$g(Z)=\E[T\mid Z]=T$ gives zero loss.
\end{proof}

Uncertainty about the state can be large even when the required attention
share is constant. Conversely, variation in that share forces the platform
to compromise across states. A signal can therefore suffice for perfect
limiting decisions without revealing the full state, provided it identifies
the balancing share. The constants in the bounds measure how strongly the
equilibrium action responds to attention.

%% file: partial-reallocation.tex
\subsection{Partial reallocation of unused exposure}
\label{sec:partial-reallocation}

The comparison with full reallocation has a local counterpart. We study a
population game to isolate it. Appendix Lemma~\ref{lem:partial-reallocation-finite}
shows that this game is the limit of finite disclosure games with uniform
within-camp weights. As in the baseline, the realized pools are observed
by both advocates; this extension is separate from private learning. The state belongs to a finite set
$\Theta\subset(\underline x,\overline x)$, with strictly positive prior
probabilities $\pi_\theta$. In state $\theta$, both camps have the same
evidence law $F_\theta$ on the fixed support
$[\underline x,\overline x]$, with mean $\theta$. Its density is $C^2$
and strictly positive on the interior of this interval; it may vanish at
the endpoints. A camp chooses a measurable
disclosure set. Let $p_k$ be the probability of this set and $m_k$ its
conditional mean when $p_k>0$. The institutional parameter
$\tau\in[0,1]$ is the fraction of unused exposure reassigned
proportionally to the camp's disclosed observations. Its realized exposure
per unit of committed camp share is therefore
\[
 w_\tau(p_k)=p_k+\tau(1-p_k).
\]
An empty submission receives exposure $\tau$ at $a_0$. Thus $p_k=0$
means $m_k=a_0$, not a limiting tail at a support endpoint. With a vanishing
anchor, the population action is
\begin{equation}
 a^\tau(S_1,S_2)
 =\frac{r w_\tau(p_1)m_1+(1-r)w_\tau(p_2)m_2}
        {r w_\tau(p_1)+(1-r)w_\tau(p_2)}.
 \label{eq:partial-reallocation-action}
\end{equation}
At $\tau=0$, set the action at two empty submissions equal to $a_0$.
This convention has no effect on the equilibria below. The cases
$\tau=0$ and $\tau=1$ correspond to nontransferable and fully reallocated
exposure, respectively. The result concerns sufficiently small $\tau$;
it makes no claim about global comparative statics between the endpoints.

Let $\mathcal R$ be a compact interval contained in $(0,1)$, with $1/2$
in its interior. Define
\begin{equation}
 p_\theta=1-F_\theta(\theta),\qquad
 G_\theta=\E_\theta(X-\theta)^+
          =\E_\theta(\theta-X)^+,\qquad
 \kappa_\theta=4G_\theta,\qquad
 h_\theta=\frac{2p_\theta-1}{4p_\theta(1-p_\theta)}.
 \label{eq:partial-reallocation-coefficients}
\end{equation}
The coefficient $\kappa_\theta$ is the baseline response of the action to
the upward camp's share, evaluated at $r=1/2$.

\begin{proposition}[Partial reallocation near the nontransferable institution]
\label{prop:partial-reallocation}
There exists $\bar\tau>0$ such that, for every
$\tau\in[0,\bar\tau)$, $r\in\mathcal R$, and $\theta\in\Theta$,
the population game has a unique pure-strategy equilibrium up to
$F_\theta$-null sets. Its action $c^\tau(r,\theta)$ and two disclosure
cutoffs are twice continuously differentiable locally in $(r,\tau)$.
For $\tau>0$, the upward camp discloses above a cutoff $t_+$ and the
downward camp below a cutoff $t_-$ satisfying
\begin{equation}
 t_-<c^\tau(r,\theta)<t_+.
 \label{eq:partial-reallocation-gap}
\end{equation}
The unique state-contingent implementing share has the expansion
\begin{equation}
 r_\tau^*(\theta)=\frac12+\tau h_\theta+O(\tau^2),
 \qquad
 c^\tau(1/2,\theta)=\theta-\tau\kappa_\theta h_\theta+O(\tau^2).
 \label{eq:partial-reallocation-implementation}
\end{equation}
Let
\[
 \bar h=
 \frac{\sum_\theta\pi_\theta\kappa_\theta^2h_\theta}
      {\sum_\theta\pi_\theta\kappa_\theta^2}.
\]
The unique optimal committed share, for all sufficiently small $\tau$,
and its loss satisfy
\begin{align}
 r_\tau^{EA}&=\frac12+\tau\bar h+O(\tau^2),
 \label{eq:partial-reallocation-exante}\\
 \min_{r\in\mathcal R}\sum_\theta\pi_\theta
       [c^\tau(r,\theta)-\theta]^2
 &=\tau^2\sum_\theta\pi_\theta\kappa_\theta^2
          (h_\theta-\bar h)^2+o(\tau^2).
 \label{eq:partial-reallocation-loss}
\end{align}
The coefficient in \eqref{eq:partial-reallocation-loss} is strictly
positive if and only if $p_\theta$ varies across states.
\end{proposition}

Partial reallocation rewards selecting fewer, more extreme observations.
The two cutoffs separate even when the transfer fraction is small. The
institution can compensate in a given state by changing the camp shares,
but the required adjustment depends on the fractions of observations above
and below the state's mean. When these fractions differ across states, one
committed adjustment leaves a strictly positive second-order welfare cost.
The result does not say that loss increases monotonically over the entire
range $\tau\in[0,1]$, or that the cost must be strictly positive when the
coefficient in \eqref{eq:partial-reallocation-loss} vanishes.

The proof is in Appendix~\ref{app:partial-reallocation}. A two-state
example there gives opposite policy adjustments in the two states and a
strictly positive leading loss after the committed share is optimized.

%% file: precision-cost.tex
\begin{proposition}[The precision cost of selection]
\label{prop:varcost}
Maintain the uniform within-camp weights and anchor rate of this section.
Under Assumption~\ref{ass:evidence}, suppose both camps draw from a common
law $F_\theta$, their sizes are asymptotically equal
($\alpha_1=\alpha_2=1/2$), and $r=1/2$. Then $c^*=\theta$,
$D^*=1/2$, and
\begin{equation}
 V^C(1/2,\theta)-V^F(1/2,\theta)
 =\operatorname{Var}_\theta\big(|X-\theta|\big)>0,
 \label{eq:precision-premium}
\end{equation}
where $V^F(1/2,\theta)=\operatorname{Var}_\theta(X)$ is the
asymptotic variance under complete disclosure. The common law need not be
symmetric about its mean.
\end{proposition}

The proof is in Appendix~\ref{app:proof:varcost}. Equal attention removes
limiting bias but leaves a strictly positive precision cost under the
maintained positive-density assumption. Algebraically, the premium is
nonnegative and equals zero only when $|X-\theta|$ is constant; that
exception is excluded by Assumption~\ref{ass:evidence}. The result compares
equilibrium selection with complete disclosure at the same attention
share. It does not assert this ranking for heterogeneous camp laws;
Appendix Remark~\ref{rem:reversal} gives a reversal.
For $X\sim U[0,1]$ and $\theta=1/2$, the variances are
$V^C=5/48$ and $V^F=1/12$: selection raises the asymptotic variance by
$25\%$ even though the action converges to the truth.

%% file: new-finite-result.tex
The order-$1/n$ movement in the optimal share has a smaller welfare effect.
The uniform derivative expansion used to establish
Theorem~\ref{thm:exactpolicy} also identifies that effect for exact risk.

\begin{corollary}[The welfare gain from finite-pool adjustment]
\label{cor:finite-gain}
Under the conditions of Theorem~\ref{thm:exactpolicy}, let
$A=\mathcal L_\mu''(r_\mu)>0$ and
$B=\mathcal Q_{\mu,\nu}'(r_\mu)$. For every sequence of exact optimal shares,
\begin{equation}
 \mathcal W_{n,\mu}(r_\mu)
 -\mathcal W_{n,\mu}(\widehat r_{n,\mu})
 =\frac{B^2}{2An^2}+o(n^{-2}).
 \label{eq:exact-finite-gain}
\end{equation}
Moreover, the second-order policy satisfies
\begin{equation}
 \mathcal W_{n,\mu}(\widetilde r_{n,\mu})
 -\mathcal W_{n,\mu}(\widehat r_{n,\mu})=o(n^{-2}).
 \label{eq:approximate-policy-regret}
\end{equation}
\end{corollary}

\begin{proof}
Write $h_n=\widehat r_{n,\mu}-r_\mu=-B/(An)+o(n^{-1})$.
The uniform derivative expansion in
Proposition~\ref{prop:riskderivative} and Taylor's theorem give, uniformly
for $|u|\leq C/n$ for each fixed $C$,
\[
 \mathcal W_{n,\mu}'(r_\mu+u)=Au+B/n+o(n^{-1}).
\]
Integrating from $0$ to $h_n$ yields
\[
 \mathcal W_{n,\mu}(\widehat r_{n,\mu})
 -\mathcal W_{n,\mu}(r_\mu)
 =\tfrac12 Ah_n^2+(B/n)h_n+o(n^{-2})
 =-\frac{B^2}{2An^2}+o(n^{-2}).
\]
The exact and second-order policies differ by $o(n^{-1})$ and both lie
within $O(n^{-1})$ of $r_\mu$. The derivative is $O(n^{-1})$ on the
interval between them, so a second integration gives
\eqref{eq:approximate-policy-regret}.
\end{proof}

When $B\neq0$, the adjustment strictly lowers exact risk for all
sufficiently large pools. When $B=0$, the gain is only known to be
$o(n^{-2})$. Thus finite-pool policy adjustment is a refinement of the
large-pool design: its welfare gain is at most of order $n^{-2}$, whereas the
leading sampling risk at an implementing share is of order $n^{-1}$.

%% file: appendix-body.tex
\appendix

\section{Supplementary Results and Extensions}
\label{sec:robustness}

The main analysis assumes that both senders observe the realized pool, the
platform evaluates one receiver action, and unused exposure is neither
processed by the receiver nor reallocated. This appendix relaxes these
assumptions one at a time.
\subsection{Technical foundations}
\label{app:foundations}

\subsubsection{The decision maker's working model}
\label{sec:subjective-bayes}

The visible-average rule does not describe a receiver that ignores the
evidence it sees. It describes a receiver that processes displayed values
carefully but does not model how advocates selected them. The rule is the
exact Bayes action under a simple misspecified working model. Let $\vartheta$
denote the state in that subjective model, to distinguish it from the
objective state $\theta$ governing the evidence process.

\begin{proposition}[Selection-naive Bayesian representation]
\label{prop:selection-naive}
Fix $n$, a disclosure set $S$, and any common precision scale $\tau_n>0$.
Suppose the receiver's subjective prior and signal model are
\begin{equation}
 \vartheta\sim N\!\left(a_0,\frac{1}{\tau_n\eta_n}\right),
 \qquad
 X_i\mid\vartheta\sim N\!\left(\vartheta,
                         \frac{1}{\tau_n v_{ni}}\right)
 \quad\text{independently across }i,
 \label{eq:subjective-gaussian}
\end{equation}
and the receiver believes that $S$ is generated independently of
$\vartheta$ and of all evidence values. After observing $S$ and
$\{X_i:i\in S\}$, its Bayes action under quadratic loss is exactly
$a_n(S)$ in~\eqref{eq:action}, for every possible disclosure record.
\end{proposition}

The proof is in Appendix~\ref{app:proof:selection-naive}.

This is a representation of the receiver's subjective model, not a
restriction on the objective evidence distributions. The true laws may be
bounded, non-Gaussian, and camp specific as in
Assumption~\ref{ass:evidence}. The common scale $\tau_n$ cancels from the
posterior mean and may vary with $n$. Exposure weights are relative
subjective precisions, or repetitions that the receiver treats as
independent confirmations. The receiver uses the identities of displayed
items to attach these precisions but assigns no likelihood content to the
event that a slot is filled or empty. A correctly specified Bayesian would
instead use the equilibrium probability of disclosure and withholding.
Section~\ref{sec:missing-processing} considers the case in which empty slots
receive partial weight as benchmark pseudo-signals; Appendix
Section~\ref{sec:neutral-backfill} gives the full derivation.

\subsubsection{Monotone sender preferences}

\begin{proposition}[Ordinal robustness of the disclosure game]
\label{prop:ordinal}
Fix $n$, a realized evidence profile, and a state. Replace Sender 1's
payoff by $u_1(a_n,\theta)$ and Sender 2's payoff by $u_2(a_n,\theta)$,
where $u_1(\cdot,\theta)$ is strictly increasing and
$u_2(\cdot,\theta)$ is strictly decreasing on $\mathcal X$.
Then a pure or mixed strategy profile is a Nash equilibrium if and only if
its supports satisfy \eqref{eq:mixed-equilibrium-support}. Thus the
entire equilibrium characterization of Theorem~\ref{thm:finite} is unchanged:
every equilibrium produces $c_n$ with probability one and has the same
disclosure record away from cutoff ties.
\end{proposition}

The proof is in Appendix~\ref{app:proof:ordinal}. The result covers, for
example, quadratic ideal-point payoffs when Sender 1's ideal lies weakly
above $\overline x$ and Sender 2's ideal lies weakly below $\underline x$.
It also applies when the state is unobserved and each sender's finite
interim expected utility remains strictly monotone in the stated direction.
The baseline common-pool design and welfare results carry over because
the equilibrium action and record away from ties are unchanged. In the
neutral-imputation and reallocation extensions, deterministic best-response
comparisons preserve the stated pure-equilibrium conclusions; no additional
mixed-equilibrium claim for those institutions is needed here.

\subsection{Information content and further implications}
\label{app:further-implications}

\subsubsection{Accuracy without a complete record}

\begin{lemma}[The record is not sufficient]
\label{lem:insufficient}
Under Assumption~\ref{ass:evidence}, suppose that one camp contains at least
one item and, for all $\theta\ne\theta'$ in a neighborhood of the true state,
the likelihood ratio $f_{k,\theta}/f_{k,\theta'}$ for that camp is
nonconstant on every nonempty open interval in
$\operatorname{int}\mathcal X$. Then, for every positive attention share,
the equilibrium record
$R=(S_1^{*},S_2^{*},\{X_{i}\}_{i\in S_1^{*}\cup S_2^{*}})$ is not a
sufficient statistic for $\theta$ relative to the full evidence pool: on a
positive-probability set of records, the conditional distribution of the
withheld items given $R$ depends on $\theta$.
\end{lemma}

The proof in Appendix~\ref{app:insufficient} uses the factorization
criterion. The receiver action, and hence the equilibrium cutoff $c_n$, can
be computed from $R$. Holding $R$ fixed, a withheld upward item can be moved
below $c_n$ without changing the record. Sufficiency would therefore require
the relevant likelihood ratio to be constant on $(\underline x,c_n)$,
contrary to the condition in the lemma. It follows that
at $r=r^*(\theta)$ the two actions have the same probability limit even
though the equilibrium record remains statistically insufficient. The two
finite-pool statistics need not be equal.

\subsubsection{Expectiles, restoration, and mechanical censoring}

\begin{remark}[Expectile benchmark]
\label{rem:expectile}
Under the common distribution $F_{\theta}$, $c^*(r)$ is exactly the
$r$-expectile:
\[
 c^*(r)
 =\operatorname*{arg\,min}_{a\in\mathcal X}
 \E_\theta\!\left[
   r\big((X-a)^+\big)^2+(1-r)\big((a-X)^+\big)^2
 \right].
\]
The first-order condition is
$r\E_\theta(X-a)^+=(1-r)\E_\theta(a-X)^+$, the equilibrium balance
equation. Thus standard expectile properties provide a statistical
benchmark for the common-distribution case
\citep{neweypowell1987,ehm2016}. The game does not, however, posit an
asymmetric loss or ask an observer to estimate this functional: strategic
tail disclosure generates it, and platform attention selects its index.
Accordingly, the location and monotonicity of $c^*(r)$ under a common law
are not claimed as new properties of expectiles; the strategic content is
the equilibrium record and its response to the attention instrument.
When $F_{1,\theta}\neq F_{2,\theta}$, $c^*$ instead equates truncated
moments from two different laws and need not equal the $r$-expectile of
either camp's distribution.
\end{remark}

\begin{proposition}[Distributional restoration at balance]
\label{prop:restoration}
Let $r=1/2$ under a common evidence distribution, and draw an item $J$
with $\Prob(J=i)=v_{ni}$, independently of the evidence. Then, as
$n\to\infty$, the law of $X_{J}$ conditional on the event that item $J$ is
disclosed in equilibrium converges in total variation to the uncensored law
$F_{\theta}$.
\end{proposition}

The proof in Appendix~\ref{app:restoration} proceeds as follows. At equal
attention,
the equilibrium threshold converges to $\theta$, above which only upward
items (attention mass $1/2$) are shown and below which only downward items
(attention mass $1/2$) are shown. The two weighted tails therefore fit
together into the uncensored marginal distribution. The visible record is
still censored because each side shows only one tail, and by
Lemma~\ref{lem:insufficient} it remains statistically insufficient. The
result concerns the law of an attention-weighted visible item and the
receiver action, not recovery of the full record.

\begin{remark}[Microfoundation of mechanical censoring]
\label{rem:identity}
For a fixed exogenous cutoff $c$ used by both camps, define the mechanical
visible mean $m(c;r)$, i.e., the probability limit of the naive average
when upward items above $c$ and downward items below $c$ are shown. A
direct computation gives the identity
\[
  H(c;r,\theta)=D(c;r)\,\big[m(c;r)-c\big],
  \qquad
  D(c;r)=r\,q(c)+(1-r)(1-q(c)),
\]
so the equilibrium narrative is precisely the fixed point $m(c^{*})=c^{*}$
of the mechanical model's visible-mean map. The strategic game thus
microfounds the reduced-form censoring rule studied in nonstrategic models
of expression---with the cutoff pinned down endogenously as the public
narrative itself.
\end{remark}

\subsubsection{Amplification relative to a truth-centered standard}

\begin{proposition}[Strategic amplification]
\label{prop:amplification}
Let $m(c;r)$ be the mechanical visible mean of
Remark~\ref{rem:identity}. If $r>1/2$, then
\[
  c^{*}(r)\;>\;m(\theta;r)\;>\;\theta ,
\]
and both inequalities reverse for $r<1/2$: relative to a benchmark in
which both camps hold their editorial standard fixed at the truth,
equilibrium curation strictly amplifies the bias of imbalanced attention.
\end{proposition}

The proof is in Appendix~\ref{app:proof:amplification}.

\subsubsection{Vanishing opposition and minority leverage}

\begin{proposition}[Extremization]
\label{prop:extremization}
Under a common evidence distribution,
$\lim_{r\uparrow1}c^{*}(r)=\overline x$ and
$\lim_{r\downarrow0}c^{*}(r)=\underline x$.
\end{proposition}

The proof is in Appendix~\ref{app:proof:extremization}.

A one-sided information environment is thus driven to its own extreme even
though listeners are naive and no one ever lies: each ally whose evidence
sits below the going narrative would drag it down, and therefore stays
silent---sender-side unraveling, with competition among allies playing the
role that receiver skepticism plays in \citet{milgrom1981}. The rate at
which countervailing attention undoes this is the economically informative
object.

\begin{proposition}[Square-root minority leverage]
\label{thm:leverage}
Suppose the camps draw from a common evidence law $F_\theta$ whose density
$f_\theta$ is continuous and strictly positive at $\underline x$ and
$\overline x$. Then
\[
  \overline x-c^{*}(r)
  \;\sim\;
  \sqrt{\frac{2(\overline x-\theta)}{f_{\theta}(\overline x)}}\,
  \sqrt{1-r}
  \quad(r\uparrow1),
  \qquad
  c^{*}(r)-\underline x
  \;\sim\;
  \sqrt{\frac{2(\theta-\underline x)}{f_{\theta}(\underline x)}}\,
  \sqrt{r}
  \quad(r\downarrow0).
\]
\end{proposition}

The proof is in Appendix~\ref{app:leverage}. If a silenced opposition is
granted an attention share $\varepsilon$, it pulls the public narrative
away from the boundary by order $\sqrt{\varepsilon}$: the marginal
disciplinary value of countervailing attention is unbounded near a
one-sided environment.

\begin{examplebox}
For $X\sim U[0,1]$, solving $r(1-c)^{2}=(1-r)c^{2}$ gives
\[
  c^{*}(r)=\frac{\sqrt r}{\sqrt r+\sqrt{1-r}} .
\]
At $r=3/4$: $c^{*}\approx0.634$, against a mechanical truth-centered
visible mean of $0.625$ (amplification), and
$1-c^{*}(1-\varepsilon)\approx\sqrt{\varepsilon}$ near $r=1$ (minority
leverage).
\end{examplebox}

\begin{corollary}[Uniform implementation]
\label{cor:uniformimpl}
The minimum of $\mathcal L$ is zero for every prior supported on a set
$\Theta_{0}\subseteq\Theta$ if and only if the moment ratio
$\theta\mapsto \Gamma_2(\theta)/\Gamma_1(\theta)$ is constant on $\Theta_{0}$.
\end{corollary}

The proof is in Appendix~\ref{app:proof:uniformimpl}.

\subsection{Finite-sample welfare}
\label{sec:finitesample-appendix}

Moment balance eliminates limiting bias; it does not make a finite selected
record equivalent to complete disclosure. Three questions remain. How noisy
is the equilibrium action? How should a platform trade that noise against
small-sample bias? And does the resulting approximation describe the policy
that minimizes exact finite-pool risk? We answer these questions in turn.

Maintain the uniform within-camp schedule, anchor rate, and notation defined
in Section~\ref{sec:finitesample}.

\paragraph{Primitive bias and variance coefficients.}
For reference and proof use, we record the primitive forms of the
coefficients introduced in the main text. Fix $(r,\theta)$, abbreviate
$c^*=c^*(r,\theta)$, and let
\begin{align*}
 U(c^*)&=\E_\theta(X_1-c^*)^+,
 &L(c^*)&=\E_\theta(c^*-X_2)^+,\\
 p_1&=\Prob_\theta(X_1>c^*),
 &p_2&=\Prob_\theta(X_2<c^*).
\end{align*}
The visible attention mass and the variance of the empirical balance
equation are
\begin{align*}
 D^*(r,\theta)&=rp_1+(1-r)p_2,\\
 \Omega(r,\theta)
 &=\frac{r^2}{\alpha_1}\operatorname{Var}_\theta((X_1-c^*)^+)
   +\frac{(1-r)^2}{\alpha_2}\operatorname{Var}_\theta((c^*-X_2)^+).
\end{align*}
Accordingly,
\begin{align*}
 V^C(r,\theta)&=\frac{\Omega(r,\theta)}{D^*(r,\theta)^2},\\
 V^F(r,\theta)
 &=\frac{r^2}{\alpha_1}\operatorname{Var}_\theta(X_1)
  +\frac{(1-r)^2}{\alpha_2}\operatorname{Var}_\theta(X_2).
\end{align*}
Define
\[
 C(r,\theta)
 =-\frac{r^2}{\alpha_1}U(c^*)(1-p_1)
  +\frac{(1-r)^2}{\alpha_2}L(c^*)(1-p_2).
\]
Then the order-$1/n$ cutoff-bias coefficient is
\begin{align}
 B^{C,\nu}(r,\theta)
 &=\underbrace{\frac{C(r,\theta)}{D^*(r,\theta)^2}
   +\frac{[r f_{1,\theta}(c^*)-(1-r)f_{2,\theta}(c^*)]
          \Omega(r,\theta)}{2D^*(r,\theta)^3}}
   _{B^{C,0}(r,\theta)}
   +\frac{\nu(a_0-c^*)}{D^*(r,\theta)}.
 \label{eq:secondorderbias}
\end{align}
The two components inside $B^{C,0}$ respectively capture covariance between
the empirical balance and its slope, and curvature of the implicit cutoff
map. The final term is the receiver-benchmark component isolated in
Theorem~\ref{thm:exactpolicy}.

\paragraph{Detailed finite-pool regularity.}
The exact uniform formulation of Assumption~\ref{ass:finitepool}(i) is that
there are constants $\varepsilon,d_0,M>0$ such that, for every
$(r,\theta)\in\mathcal K$, the interval
$[c^*(r,\theta)-\varepsilon,c^*(r,\theta)+\varepsilon]$ lies in the interior
of $\mathcal X$, $D^*(r,\theta)\ge d_0$, and each density $f_{k,\theta}$ is
continuously differentiable on that interval, with the density and its first
derivative bounded by $M$ uniformly over $\mathcal K$. For part (ii), there
is a compact interval $\mathcal N_\mu$ with
$r_\mu\in\operatorname{int}\mathcal N_\mu\subset
\operatorname{int}\mathcal R_0$ such that, on the corresponding cutoff
neighborhoods for $\mathcal N_\mu\times\operatorname{supp}\mu$, each density
has three derivatives that are jointly continuous in $(x,\theta)$ and
uniformly bounded; moreover, $\eta_n^{-1}=O(n^\rho)$ for some finite $\rho$.

We first establish the distribution and risk expansion used by the
main-text policy theorem.

\begin{proposition}[Finite-pool distribution and risk]
\label{thm:finitewelfare}
Under Assumption~\ref{ass:evidence}, conditional on $\theta$ and for fixed
$r\in(0,1)$:
\begin{align}
  \sqrt n\,[c_{n}-c^{*}(r,\theta)]
  &\ \Rightarrow\ N\big(0,V^{C}(r,\theta)\big),
  \label{eq:cltcurated}\\
  \sqrt n\,[a_n(I_1\cup I_2)-\theta]
  &\ \Rightarrow\ N\big(0,V^{F}(r,\theta)\big).
  \label{eq:cltfull}
\end{align}
Moreover, under Assumption~\ref{ass:finitepool}(i),
\begin{equation}
  \E_{\theta}[c_n-c^{*}(r,\theta)]
  =\frac{B^{C,\nu}(r,\theta)}{n}+o(n^{-1}),
  \label{eq:biasexpansion}
\end{equation}
and
\begin{equation}
  \E_{\theta}\big[(c_{n}-\theta)^{2}\big]
  =\big[c^{*}(r,\theta)-\theta\big]^{2}
   +\frac{Q_{\nu}(r,\theta)}{n}+o(n^{-1}),
  \qquad
  Q_{\nu}(r,\theta)
  =V^{C}(r,\theta)
   +2[c^{*}(r,\theta)-\theta]B^{C,\nu}(r,\theta),
  \label{eq:riskexpansion}
\end{equation}
Both second-order expansions are uniform on $\mathcal K$: the suprema over
$\mathcal K$ of the absolute remainders, after multiplication by $n$,
converge to zero. In particular, at an implementing pair
$c^{*}(r,\theta)=\theta$, the leading order-$1/n$ risk term is
$V^{C}(r,\theta)/n$.
\end{proposition}

The proof is in Appendix~\ref{app:finitewelfare}. The equilibrium action is
an implicit two-sample estimator. Its numerator uses one-sided deviations,
while its effective slope is the visible attention mass $D^{*}$; selective
curation can therefore either reduce raw sampling variation or magnify it
through a thin visible record. Proposition~\ref{prop:varcost} settles the
comparison under a common evidence law and equal camp sizes and attention;
Remark~\ref{rem:reversal} shows the
ranking can reverse outside it.

We next use the risk objects defined in Section~\ref{sec:finitesample} to
compare the exact and second-order attention policies. Write
$\widetilde{\mathcal W}_{n,\mu}
=\mathcal L_\mu+\mathcal Q_{\mu,\nu}/n$ for the second-order criterion.

\begin{proposition}[Second-order attention design]
\label{prop:finitepolicy}
Suppose Assumption~\ref{ass:finitepool}(i) holds. Suppose also that
$\mathcal L_\mu$ has a unique minimizer
$r_\mu\in\operatorname{int}\mathcal R_0$, that
$\mathcal L_\mu$ and $\mathcal Q_{\mu,\nu}$ are continuous on $\mathcal R_0$ and
twice continuously differentiable near $r_\mu$, and that
$\mathcal L_\mu''(r_\mu)>0$. Then:

\smallskip
\noindent\textbf{(i) Exact policy consistency.}
The exact finite-pool risk $\mathcal W_{n,\mu}$ has a minimizer on
$\mathcal R_0$, and every sequence
$\widehat r_{n,\mu}\in\arg\min_{r\in\mathcal R_0}\mathcal W_{n,\mu}(r)$
satisfies $\widehat r_{n,\mu}\to r_\mu$.

\smallskip
\noindent\textbf{(ii) The second-order policy correction.}
For all sufficiently large $n$, the second-order criterion
$\widetilde{\mathcal W}_{n,\mu}$ has a unique minimizer
$\widetilde r_{n,\mu}$ on $\mathcal R_0$, and
\begin{equation}
 n(\widetilde r_{n,\mu}-r_\mu)
 \longrightarrow
 -\frac{\mathcal Q_{\mu,\nu}'(r_\mu)}{\mathcal L_\mu''(r_\mu)}.
 \label{eq:secondorderpolicy}
\end{equation}
Moreover,
\begin{equation}
 \widetilde{\mathcal W}_{n,\mu}(\widetilde r_{n,\mu})
 =\mathcal L_\mu(r_\mu)+\frac{\mathcal Q_{\mu,\nu}(r_\mu)}{n}
  -\frac{[\mathcal Q_{\mu,\nu}'(r_\mu)]^2}
         {2\mathcal L_\mu''(r_\mu)n^2}
  +o(n^{-2}).
 \label{eq:secondordervalue}
\end{equation}

\smallskip
\noindent\textbf{(iii) Conditional bias--variance decomposition.}
If $\mu$ is a point mass at $\theta$, then $r_\mu=r^*(\theta)$. Write
\[
 \kappa(\theta)
 =\left.\frac{\partial c^*(r,\theta)}{\partial r}
  \right|_{r=r^*(\theta)}>0.
\]
The second-order conditional policy satisfies
\begin{equation}
 n[\widetilde r_n(\theta)-r^*(\theta)]
 \longrightarrow
  -\frac{B^{C,\nu}(r^*,\theta)}{\kappa(\theta)}
  -\frac{\partial_r V^C(r^*,\theta)}{2\kappa(\theta)^2}.
 \label{eq:biasvarianceshift}
\end{equation}
Thus finite-pool attention moves against endogenous estimator bias and the
anchor's pull, and toward lower sampling variance. Since
\[
 B^{C,\nu}(r^*,\theta)
 =B^{C,0}(r^*,\theta)+\frac{\nu(a_0-\theta)}{D^*(r^*,\theta)},
\]
the anchor contributes the transparent policy shift
$-\nu(a_0-\theta)/[\kappa(\theta)D^*(r^*,\theta)]$: an anchor below the state
calls for more upward attention, and an anchor above it for less. For a
nondegenerate prior, the same
formula~\eqref{eq:secondorderpolicy} applies at the ex ante optimum, with
$\mathcal Q_{\mu,\nu}$ averaging the state-specific endogenous-bias, anchor,
and variance terms.
\end{proposition}

The proof is in Appendix~\ref{app:proof:finitepolicy}.

Parts (i) and (ii) make different claims. The uniform $o(n^{-1})$ remainder
identifies the limit of the exact optimizer but not its order-$1/n$
displacement, because policies separated by $O(n^{-1})$ differ in welfare
only at order $n^{-2}$. Thus \eqref{eq:secondorderpolicy} follows for the
second-order criterion, but it cannot be obtained for the exact optimizer by
differentiating \eqref{eq:integratedrisk}. The stronger smoothness condition
in Assumption~\ref{ass:finitepool}(ii) is used below to control observations
that cross the cutoff when attention changes.

\begin{proposition}[Uniform derivative expansion]
\label{prop:riskderivative}
Maintain the conditions of Proposition~\ref{prop:finitepolicy}. For every
compact interval $\mathcal N_\mu$ with
$r_\mu\in\operatorname{int}\mathcal N_\mu\subset
\operatorname{int}\mathcal R_0$ on which $\mathcal L_\mu$ and
$\mathcal Q_{\mu,\nu}$ are twice continuously differentiable and
Assumption~\ref{ass:finitepool}(ii) holds, the exact finite-pool risk is
continuously differentiable on $\operatorname{int}\mathcal N_\mu$ and
\begin{equation}
 \sup_{r\in\mathcal N_\mu}
 n\left|
 \mathcal W_{n,\mu}'(r)
 -\mathcal L_\mu'(r)
 -\frac{\mathcal Q_{\mu,\nu}'(r)}{n}
 \right|\longrightarrow0.
 \label{eq:riskderivativeexpansion}
\end{equation}
Consequently, every exact minimizer has the order-$1/n$ displacement stated
in Theorem~\ref{thm:exactpolicy}.
\end{proposition}

The proof is combined with the proof of Theorem~\ref{thm:exactpolicy} in
Appendix~\ref{app:exactpolicy}.

\subsubsection{Exact two-point benchmark}
\label{app:exact-two-point-benchmark}

Here we give the exact calculation behind
Figure~\ref{fig:finite-policy-validation}. Let $n$ be even, give each camp
$m=n/2$ draws, and set $\eta_n=n^{-2}$ and $a_0=\theta=5$. Let
$K_1\sim\operatorname{Bin}(m,1/5)$ count value-$9$ draws in Sender 1's pool,
and let $K_2\sim\operatorname{Bin}(m,4/5)$ count value-$4$ draws in Sender
2's pool. Conditional on these counts, the unique finite-pool cutoff is
\[
 c_n(r;k_1,k_2)
 =\frac{5\eta_n+9r(k_1/m)+4(1-r)(k_2/m)}
        {\eta_n+r(k_1/m)+(1-r)(k_2/m)}.
\]
Writing $p_m(k;q)=\binom{m}{k}q^k(1-q)^{m-k}$, the exact finite-pool risk is
the finite sum
\[
 W_n(r)=\sum_{k_1=0}^{m}\sum_{k_2=0}^{m}
 p_m(k_1;1/5)p_m(k_2;4/5)
 \bigl[c_n(r;k_1,k_2)-5\bigr]^2.
\]
The exact policy plotted in Figure~\ref{fig:finite-policy-validation} is the
numerical minimizer of this expression over $[0,1]$. Because the calculation
sums over all pairs $(k_1,k_2)$, the figure contains no simulation error.

\subsubsection{Smooth validation and the information cost of curation}

The exact calculation in the main text uses a discrete distribution. We now
repeat the exercise under a smooth evidence law satisfying
Assumptions~\ref{ass:evidence} and~\ref{ass:finitepool}. Let $X=4+5Z$,
where $Z$ has density
\begin{equation}
 g(z)
 =\frac{12}{65}\,40z^{39}
  +\frac{53}{65}\,40(1-z)^{39},
 \qquad 0\leq z\leq1.
 \label{eq:smooth-validation-density}
\end{equation}
This mixture density is smooth and strictly positive on the interior, and
$\E[Z]=1/5$, so $\E[X]=\theta=5$. Both camps draw independently from this
common law, camp sizes are equal, $a_0=5$, and $\eta_n=n^{-2}$. Thus
$r^*=1/2$ and $\nu=0$. Exact polynomial integration of the population objects in
Proposition~\ref{prop:finitepolicy} and Theorem~\ref{thm:exactpolicy} gives,
rounded to six decimal places,
\[
 \kappa=2.863832,
 \qquad B^{C,0}=-0.902898,
 \qquad \partial_rV^C=24.016585,
\]
Let $\gamma=-\mathcal Q_{\mu,0}'(1/2)/\mathcal L_\mu''(1/2)
\approx-1.148876$, where $\mu$ is concentrated at $\theta=5$.
Define the explicit linear correction $r_n^{\mathrm{lin}}=1/2+\gamma/n$.
Then
\begin{equation}
 n(\widehat r_n-\tfrac12)\longrightarrow\gamma,
 \qquad n(\widehat r_n-\widetilde r_n)\longrightarrow0,
 \qquad
 \widetilde r_n=r_n^{\mathrm{lin}}+o(n^{-1}).
 \label{eq:smooth-validation-policy}
\end{equation}

Figure~\ref{fig:smooth-finite-policy-validation} compares $r_n^{\mathrm{lin}}$
with the minimizer of simulated finite-pool risk. For each $n$, common random
numbers are used to solve the simulated risk first-order condition; the
vertical bars are pointwise $95\%$ Monte Carlo intervals obtained from the
score variance and local risk curvature. At $n=80$, the simulated and
linear-correction shares are $0.48434$ and $0.48564$; at $n=160$, they are
$0.49257$ and $0.49282$; and at $n=320$, they are $0.49619$ and $0.49641$.
The scaled intervals at $n=160$ and $n=320$ contain the theoretical limit
$\gamma\approx-1.148876$. Thus the policy correction is not an artifact of the
off-assumption two-point calculation.

\begin{figure}[t]
 \centering
 \includegraphics[width=\textwidth]{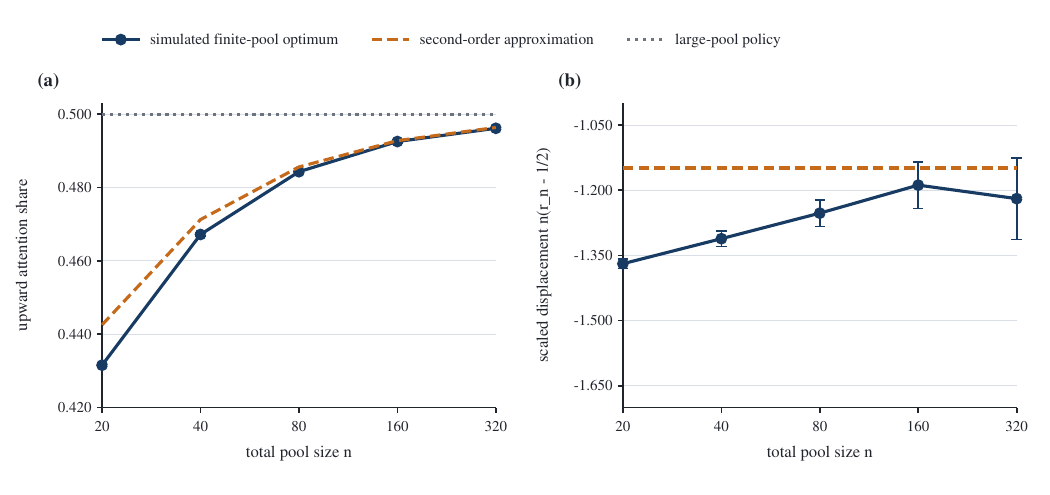}
 \caption{The smooth finite-pool risk-minimizing share and its second-order approximation.
 The evidence density is~\eqref{eq:smooth-validation-density}. Panel (a)
 compares the simulated finite-pool optimum, the linear correction in
 \eqref{eq:smooth-validation-policy}, and the large-pool policy. Panel (b)
 rescales the simulated displacement and compares it with its theoretical
 limit. The calculations use respectively $300{,}000$, $250{,}000$,
 $180{,}000$, $120{,}000$, and $80{,}000$ independent pool pairs at
 $n=20,40,80,160,320$; vertical bars are pointwise $95\%$ Monte Carlo
 intervals and may be hidden by the markers.}
 \label{fig:smooth-finite-policy-validation}
\end{figure}

Proposition~\ref{prop:varcost} in the main text gives the exact variance
premium under a common evidence law and equal camp sizes and attention.
No symmetry of the common law is required. The next example shows why
that ranking need not extend to heterogeneous evidence laws.

\begin{remark}[Reversal outside the benchmark]
\label{rem:reversal}
With asymmetric camps the ranking can reverse, because equilibrium
suppression deletes each camp's \emph{unfavorable} tail from the record:
when a camp's sampling risk is concentrated in that tail, curation filters
out its own worst noise. For instance, let Sender 2 draw
$\theta+U[-1,1]$ and Sender 1 draw $\theta+u$, where $u$ equals
$-10.5$ with probability $0.1$ and $+7/6$ with probability $0.9$ (mean
zero; a smoothed version satisfies Assumption~\ref{ass:evidence} and
changes nothing by continuity). At the implementing share
$r^{*}\approx0.19$ and equal camp sizes,
$V^{C}\approx0.44$ while $V^{F}\approx1.34$: the curated record is three
times \emph{less} noisy than complete disclosure. Accuracy aside,
suppression of a camp's own bad news can be variance-reducing for the
receiver.
\end{remark}

\subsection{Private evidence}
\label{sec:private}

We now suppose that each sender observes only its own evidence. The state
remains common knowledge between the senders, so the result does not cover
private learning about the state. The support argument in Proposition~\ref{prop:ordinal} fixes the entire
realized pool and does not establish ordinal invariance under private
evidence. We retain the baseline linear objectives and their positive
affine transformations in this extension. Write $n_k=|I_k|$ for
$k\in\{1,2\}$ and suppose $n_k\to\infty$ for both senders. Weights are uniform within each camp,
$v_{ni}=r_n/n_1$ on $I_1$ and $v_{ni}=(1-r_n)/n_2$ on $I_2$, and
$r_n\to r\in(0,1)$. For any deterministic sequence $\eta_n>0$ with
$\eta_n\to\bar\eta\in[0,\infty)$, let $c_n^{\circ}$ be the unique zero of
\begin{equation}
  \eta_n(a_0-c)+H(c;r_n,\theta)=0.
  \label{eq:private-target}
\end{equation}

We use behavioral Bayesian Nash equilibrium. A sender's strategy is a
measurable probability kernel from its own realized pool to its finite set
of disclosure subsets; private randomizations are independent conditional
on the realized pools. This is equivalent here to the distributional
strategy representation of \citet{milgromweber1985}. To make a finite-pool
tail cutoff unique, represent disclosure of the top $k$ observations by the
midpoint between the largest excluded and smallest disclosed observations,
with cutoffs $\overline x$ and $\underline x$ for the empty and full sets;
use the symmetric convention for a bottom-$k$ set. For a behavioral-strategy
profile, write $a_n(\sigma_1,\sigma_2)$ for its induced random receiver
action.

\begin{proposition}[Private evidence]
\label{prop:private}
Suppose each sender observes only its own pool and chooses its disclosure
set as a measurable, possibly randomized function of its own realization.
For every $n$, a behavioral Bayesian Nash equilibrium exists, and in every
such equilibrium each pure disclosure action used by either sender is
almost surely a tail policy, described under the convention above by a
threshold $T_1^n$ or $T_2^n$. Moreover:
\smallskip

\noindent\textbf{(a) Uniform security.}
Let $\tau_{1,n}^{\circ}=\{i:X_i>c_n^{\circ}\}$ and
$\tau_{2,n}^{\circ}=\{j:X_j<c_n^{\circ}\}$. Against arbitrary measurable,
possibly randomized opponent disclosure rules $\sigma_2$ and $\sigma_1$,
respectively, for every $\varepsilon>0$,
\begin{align*}
 \sup_{\sigma_2}\Pr\!\left(
 a_n(\tau_{1,n}^{\circ},\sigma_2)<c_n^{\circ}-\varepsilon\right)&\longrightarrow0,\\
 \sup_{\sigma_1}\Pr\!\left(
 a_n(\sigma_1,\tau_{2,n}^{\circ})>c_n^{\circ}+\varepsilon\right)&\longrightarrow0.
\end{align*}
Thus the two constant-cutoff rules are asymptotic maxmin strategies.

\noindent\textbf{(b) All-equilibrium concentration.}
Along every sequence of equilibria,
\[
  T_1^n-c_n^{\circ}\xrightarrow{p}0,\qquad
  T_2^n-c_n^{\circ}\xrightarrow{p}0,\qquad
  a_n-c_n^{\circ}\xrightarrow{p}0.
\]

\noindent\textbf{(c) Limit.}
$c_n^{\circ}\to c^*_{\bar\eta}$, the unique zero of
$\bar\eta(a_0-c)+H(c;r,\theta)$. In particular, for every rate
$\eta_n\downarrow0$, both private-evidence thresholds and the receiver
action converge to the common-knowledge narrative $c^*(r,\theta)$.
\end{proposition}

The proof (Appendix~\ref{app:private}) uses security levels rather than an
iterated-limit argument. If Sender 1 discloses exactly its items
above $c_n^{\circ}$, the action relative to $c_n^{\circ}$ is bounded below
by the empirical counterpart of \eqref{eq:private-target}, no matter which
items Sender 2 discloses; the symmetric cutoff gives
Sender 2 the matching upper bound. These two pathwise inequalities
pin down the value. Strict separation of the corresponding population
payoffs away from $c_n^{\circ}$ then forces both equilibrium thresholds to
concentrate there. Thus the result neither presumes threshold play by the
opponent when establishing the guarantee nor restricts the relative rates
at which the evidence pool grows and the anchor vanishes.

\input{private-learning}

\subsection{Personalized attention and heterogeneous receivers}
\label{sec:personalization}

The receiver so far is a single feed. Now let $J_n$ be a nonempty finite,
possibly growing set of listeners. Listener $j\in J_n$ aggregates the
\emph{same} equilibrium record with a deterministic attention row
$w_{nj}=(w_{nji})_{i=1}^n$, where $w_{nji}\geq0$ and
$\sum_iw_{nji}=1$. Define
\[
  \rho_{nj}=\sum_{i\in I_1}w_{nji},
  \qquad
  \omega_n=\max_{j\in J_n}\max_{1\leq i\leq n}w_{nji},
\]
and assume the uniform diffuseness condition
$\omega_n\log(2|J_n|)\to0$. The senders target the platform's main feed
(weights $v$, share $r$), not a weighted average of listener actions, so
their equilibrium cutoff still satisfies $c_n\xrightarrow{p}
c^*(r,\theta)$. From its disclosed record $S_n^*$, listener $j$ takes
\[
  a_{nj}=\frac{\sum_{i\in S_n^*}w_{nji}X_i}
                  {\sum_{i\in S_n^*}w_{nji}};
\]
set $a_{nj}=a_0$ if the denominator is zero, an event that the result below
makes asymptotically irrelevant. Define
\[
  y(\rho;c^{*})
  =
  \frac{\rho\,\E_{\theta}[X_1;X_1>c^{*}]
        +(1-\rho)\,\E_{\theta}[X_2;X_2<c^{*}]}
       {\rho\,\Prob_{\theta}(X_1>c^{*})
        +(1-\rho)\,\Prob_{\theta}(X_2<c^{*})}
  .
\]

\begin{proposition}[Uniform personalized accuracy]
\label{prop:personalized}
Under the uniform diffuseness condition above,
\begin{equation}
  \max_{j\in J_n}
  \big|a_{nj}-y(\rho_{nj};c^*(r,\theta))\big|
  \xrightarrow{p}0,
  \label{eq:uniform-personalization}
\end{equation}
and $y(\cdot\,;c^*)$ is strictly increasing. Consequently:
(a) for any main-feed design $r$, at most one local camp share $\rho$
can make a listener accurate;
(b) if the main feed uses the conditional implementing share
$r=r^*(\theta)$, then the entire growing receiver is uniformly accurate if
and only if its rows are uniformly balanced:
\[
  \max_{j\in J_n}|a_{nj}-\theta|\xrightarrow{p}0
  \quad\Longleftrightarrow\quad
  \max_{j\in J_n}|\rho_{nj}-r^*(\theta)|\longrightarrow0.
\]
\end{proposition}

The proof is in Appendix~\ref{app:personalized}.

\begin{corollary}[Two-community segregation under conditional balance]
\label{cor:homophily}
Suppose the main feed uses $r=r^*(\theta)$. Let $u_n^1,u_n^2$ be two
deterministic normalized attention rows representing source neighborhoods,
with
\[
 \max_{g\in\{1,2\},i}u_{ni}^g\log(2|J_n|)\to0,
 \qquad
 \sum_{i\in I_1}u_{ni}^1\to r^*(\theta)+\delta,
 \qquad
 \sum_{i\in I_1}u_{ni}^2\to r^*(\theta)-\delta,
\]
where $0<\delta<\min\{r^*(\theta),1-r^*(\theta)\}$. A community-1 listener
uses the row $h u_n^1+(1-h)u_n^2$, and a community-2 listener uses
$h u_n^2+(1-h)u_n^1$, where $h\in[1/2,1]$. Their local upward shares
converge to
\[
 \rho_1(h)=r^*(\theta)+(2h-1)\delta,
 \qquad
 \rho_2(h)=r^*(\theta)-(2h-1)\delta.
\]
At $h=1/2$ both communities are asymptotically accurate. For every
$h>1/2$, community 1's action converges above $\theta$ and community 2's
below it, and their limiting disagreement is strictly increasing in $h$.
\end{corollary}

\begin{proof}
The displayed shares follow from mixing the two neighborhood compositions.
Proposition~\ref{prop:personalized} gives the action limits, and
$y(\cdot;\theta)$ is strictly increasing with
$y(r^*(\theta);\theta)=\theta$. Finally,
\[
 \frac{d}{dh}\{y(\rho_1(h);\theta)-y(\rho_2(h);\theta)\}
 =2\delta\{y'(\rho_1(h);\theta)+y'(\rho_2(h);\theta)\}>0.
\]
\end{proof}

The condition $r=r^*(\theta)$ in part (b) and the corollary is a conditional
implementation benchmark. It is an ex ante feasible policy only when the
moment-balance share is state invariant or the platform observes enough
information before assigning attention; otherwise the second-best policy
of Section~\ref{sec:exante} governs the main feed, and exact receiver-wide
accuracy is not generally attainable. Conditional on a truthful aggregate
narrative, two groups of listeners can nevertheless err in opposite
directions. Segregated attention harms not by spreading falsehoods (all
disclosed evidence is true) but by preventing the two truncations from
meeting inside the same average. Proposition~\ref{prop:personalized} is a
static personalization theorem for a growing collection of deterministic
diffuse rows, not a model of repeated network learning, endogenous link
formation, or social influence. Its claim is deliberately row-level: a
platform with one aggregate instrument cannot repair heterogeneous
downstream access.

\subsection{Neutral backfilling and missing-slot imputation}
\label{sec:neutral-backfill}

Section~\ref{sec:missing-processing} introduced this relaxation and its
sharp robustness boundary. Here we give the full finite-pool equilibrium,
behavioral foundation, and implementation analysis. Under an operational
\emph{neutral-backfill} interpretation, a fraction $\lambda\in[0,1]$ of
each unused slot is filled by outside content with value $a_0$, while the
remainder stays blank. Under a behavioral interpretation, all slots are
publicly known and the receiver gives an empty slot relative weight
$\lambda$, imputing its value at $a_0$. In either case the induced action is
the minimizer of
\[
 \mathcal Q_{n,\lambda}(y;S)
 =\eta_n(y-a_0)^2
  +\sum_{i\in S}v_{ni}(y-X_i)^2
  +\lambda\sum_{i\notin S}v_{ni}(y-a_0)^2.
\]
A disclosed item is treated as a full-weight observation, while a missing
slot is treated as a benchmark pseudo-observation at $a_0$ with confidence
$\lambda$. Because $\eta_n>0$, the criterion is strictly convex, and its
unique minimizer is $a_{n,\lambda}(S)$ in~\eqref{eq:partial-action}.
Thus $\lambda=0$ is the baseline average, while $\lambda=1$ fills or imputes
every missing slot at $a_0$. The senders continue to maximize and minimize
the induced action. Unlike the vanishing anchor in the baseline, $a_0$ becomes
an economically substantive primitive whenever $\lambda>0$: operationally
it is the value of neutral replacement content; behaviorally it is a public
prior default or institutional reference point. It is fixed before
disclosure, not chosen as an additional policy instrument here. The
criterion $\mathcal Q_{n,\lambda}$ represents the receiver's decision
procedure under the behavioral interpretation, not the paper's welfare
benchmark, which remains squared error relative to the true state. More
precisely, under the subjective Gaussian model of
Proposition~\ref{prop:selection-naive}, let every missing slot $i$ contribute
a benchmark pseudo-signal equal to $a_0$ with subjective precision
$\lambda\tau_n v_{ni}$. Up to an additive constant and a positive scale,
$\mathcal Q_{n,\lambda}$ is the resulting negative log posterior, and
$a_{n,\lambda}$ is its posterior mean. Thus the receiver remains Bayesian
within the same working model while giving observable missingness partial
weight.

The behavioral interpretation also has an exact attention-to-absence
representation.
Suppose the receiver processes every displayed slot but processes each
publicly known empty slot with probability $\lambda$; whenever it processes
an empty slot, it codes the missing value at the default $a_0$. If it chooses
$y$ to minimize expected slot-weighted quadratic loss, its criterion is
exactly $\mathcal Q_{n,\lambda}$. Equivalently, $\lambda$ is the relative
cognitive weight placed on missing and displayed exposure: a receiver with
$\lambda=0.4$ gives each empty slot $40\%$ of the weight of an equally
exposed displayed item, always at the fixed imputation $a_0$. Even
$\lambda=1$ means full attention to the existence of missing slots, not
correct Bayesian inference about why those particular values were
suppressed.

The parameter also has direct behavioral content outside the equilibrium
feedback. Let
\[
 m_n(S)=\sum_{i\notin S}v_{ni},\qquad
 A_n(S)=\eta_n a_0+\sum_{i\in S}v_{ni}X_i,\qquad
 B_n(S)=\eta_n+\sum_{i\in S}v_{ni}.
\]
Then
$a_{n,\lambda}(S)=[A_n(S)+\lambda a_0 m_n(S)]/
[B_n(S)+\lambda m_n(S)]$. Holding the displayed record fixed,
\[
 \frac{\partial a_{n,\lambda}(S)}{\partial\lambda}
 =\frac{m_n(S)[a_0-a_{n,\lambda}(S)]}
        {B_n(S)+\lambda m_n(S)}.
\]
Greater weight on missingness therefore moves the action monotonically toward
$a_0$, with a response that scales with the publicly known missing attention
mass. An experiment that varies the salience of known missing slots while
holding displayed values fixed can use this response to discipline $\lambda$
before studying the senders' equilibrium adjustment. More directly, when
$m_n(S)>0$ and $a_{n,\lambda}(S)\neq a_0$, the observed action and displayed
record identify the processing weight through
\begin{equation}
 \lambda
 =\frac{a_{n,\lambda}(S)B_n(S)-A_n(S)}
        {m_n(S)[a_0-a_{n,\lambda}(S)]}.
 \label{eq:lambda-identification}
\end{equation}
If the displayed-record average already equals $a_0$, changing the weight on
missing slots has no behavioral effect and $\lambda$ is correspondingly not
identified from that choice. This is an off-equilibrium measurement
statement: equilibrium feedback is studied only after the receiver rule is
disciplined.

For $\lambda<1$, define
\begin{align}
 K_{n,\lambda}(t)
 &=\frac{\eta_n+\lambda}{1-\lambda}(a_0-t)
   +\sum_{i\in I_1}v_{ni}(X_i-t)^+
   -\sum_{i\in I_2}v_{ni}(t-X_i)^+,
 \label{eq:partial-Kn}\\
 K_\lambda(t;r,\theta)
 &=\frac{\lambda}{1-\lambda}(a_0-t)
   +r\E_\theta(X_1-t)^+
   -(1-r)\E_\theta(t-X_2)^+.
 \label{eq:partial-K}
 \\
 D_\lambda(t;r,\theta)
 &=r\Prob_\theta(X_1>t)+(1-r)\Prob_\theta(X_2<t).
 \label{eq:partial-D}
\end{align}

\begin{proposition}[Neutral backfilling and missing-slot imputation]
\label{prop:neutral-backfill}
Fix $\lambda\in[0,1]$.

\smallskip
\noindent\textbf{(i) Unique finite equilibrium outcome.}
If $\lambda<1$, $K_{n,\lambda}$ has a unique zero
$t_{n,\lambda}$ in
\[
 \big(\lambda a_0+(1-\lambda)\underline x,
      \lambda a_0+(1-\lambda)\overline x\big).
\]
Every pure-strategy Nash equilibrium induces the same action
\begin{equation}
 y_{n,\lambda}
 =\frac{t_{n,\lambda}-\lambda a_0}{1-\lambda},
 \label{eq:partial-y}
\end{equation}
and, up to cutoff ties, the same tail disclosure profile
\[
 S_{1,\lambda}^*=\{i\in I_1:X_i>t_{n,\lambda}\},
 \qquad
 S_{2,\lambda}^*=\{i\in I_2:X_i<t_{n,\lambda}\}.
\]
If $\lambda=1$, the unique cutoff is $a_0$, the same two tail rules apply,
and the unique action is
\begin{equation}
 y_{n,1}
 =a_0+\frac{\sum_{i\in I_1}v_{ni}(X_i-a_0)^+
             -\sum_{i\in I_2}v_{ni}(a_0-X_i)^+}
            {1+\eta_n}.
 \label{eq:full-imputation-action}
\end{equation}

\smallskip
\noindent\textbf{(ii) Large pools and the representation paradox.}
Under the conditions of Proposition~\ref{thm:largepool}, if $\lambda<1$ then
$t_{n,\lambda}\xrightarrow{p}t_\lambda^*(r,\theta)$, the unique zero of
$K_\lambda$, and
\[
 y_{n,\lambda}\xrightarrow{p}
 y_\lambda^*(r,\theta)
 =\frac{t_\lambda^*(r,\theta)-\lambda a_0}{1-\lambda}.
\]
The attention derivatives satisfy
\begin{align}
 \frac{\partial t_\lambda^*}{\partial r}
 &=\frac{\E_\theta(X_1-t_\lambda^*)^+
          +\E_\theta(t_\lambda^*-X_2)^+}
         {\lambda/(1-\lambda)+D_\lambda(t_\lambda^*;r,\theta)}>0,
 \label{eq:partial-cutoff-cs}\\
 \frac{\partial y_\lambda^*}{\partial r}
 &=\frac{\E_\theta(X_1-t_\lambda^*)^+
          +\E_\theta(t_\lambda^*-X_2)^+}
         {\lambda+(1-\lambda)D_\lambda(t_\lambda^*;r,\theta)}>0.
 \label{eq:partial-action-cs}
\end{align}
Hence for every $\lambda<1$, more upward attention raises the action and
the common cutoff: Sender 1 discloses less often and more extremely,
while Sender 2 discloses more often. As $\lambda\uparrow1$,
$t_\lambda^*\to a_0$ and the cutoff response vanishes. At $\lambda=1$,
\[
 y_1^*(r,\theta)
 =a_0+r\E_\theta(X_1-a_0)^+-(1-r)\E_\theta(a_0-X_2)^+;
\]
the action remains strictly increasing in $r$, but disclosure frequencies
no longer respond to attention.

\smallskip
\noindent\textbf{(iii) Conditional implementation.}
Let
\[
 \bar t_\lambda(\theta)=\lambda a_0+(1-\lambda)\theta,
 \quad
 \Gamma_{1,\lambda}(\theta)=\E_\theta(X_1-\bar t_\lambda(\theta))^+,
 \quad
 \Gamma_{2,\lambda}(\theta)=\E_\theta(\bar t_\lambda(\theta)-X_2)^+.
\]
An interior attention share implements
$y_\lambda^*(r,\theta)=\theta$ if and only if
\begin{equation}
 r=r_\lambda^*(\theta)
 =\frac{\Gamma_{2,\lambda}(\theta)+\lambda(\theta-a_0)}
        {\Gamma_{1,\lambda}(\theta)+\Gamma_{2,\lambda}(\theta)},
 \label{eq:partial-implementing-share}
\end{equation}
which is feasible if and only if
\begin{equation}
 -\Gamma_{2,\lambda}(\theta)
 <\lambda(\theta-a_0)
 <\Gamma_{1,\lambda}(\theta).
 \label{eq:partial-feasibility}
\end{equation}
Whenever feasible, the implementing share is unique. At $\lambda=0$ this
is exactly $r^*(\theta)$ in Theorem~\ref{thm:design}. If $a_0=\theta$ and the
camps have a common evidence law, equal attention implements for every
$\lambda$.
\end{proposition}

The proof is in Appendix~\ref{app:neutral-backfill}. Neutral backfilling or
benchmark imputation acts like a persistent anchor at $a_0$: it preserves the
tail game but dampens the response of editorial standards to attention. It
also creates a new constraint on policy. When the benchmark is too far from the state,
condition~\eqref{eq:partial-feasibility} fails and no interior allocation of
attention can offset the receiver's imputation error. Under a prior, a fixed
share implements almost every state if and only if
$r_\lambda^*(\theta)$ is feasible and constant almost surely; otherwise the
platform again faces an ex ante compromise, now over the distortions
$y_\lambda^*(r,\theta)-\theta$.

The operational interpretation requires no claim about receiver
sophistication. Under the behavioral interpretation, the receiver is
Bayesian within the misspecified working model of
Proposition~\ref{prop:selection-naive}, but it does not form the correctly
specified posterior induced by strategic disclosure. If the receiver uses
that equilibrium posterior, replacing a disclosed item by a higher one also
changes the selection likelihood, with no generally signed effect. The
exchange argument need not yield tail best responses. A correctly specified
Bayesian extension therefore requires additional restrictions on what
missing-data statistics are observed and how posterior actions respond to
them. The imputation rule instead isolates the effect of giving observable
empty slots mechanical or cognitive weight; it is not an intermediate
equilibrium-inference concept. Nor is $\lambda$
the $\chi$ parameter of cursed equilibrium \citep{eysterrabin2005}: it scales
the weight on observable missing slots rather than the perceived correlation
between another player's information and action.

\subsection{Fully fungible within-camp exposure: a boundary}
\label{sec:fungible}

Neutral backfilling leaves an item's promised exposure nontransferable to
the sender's other evidence. We now study the opposite boundary. Suppose
the platform still commits camp shares $r_n$ and $1-r_n$, but after
disclosure reallocates each camp's entire share across the items that camp
submits. For $k\in\{1,2\}$, define its submitted-record average by
\[
 \overline X_{k,n}(S_k)
 =\begin{cases}
   \displaystyle
   \frac{\sum_{i\in S_k}v_{ni}X_i}
        {\sum_{i\in S_k}v_{ni}},&S_k\neq\varnothing,\\[8pt]
   a_0,&S_k=\varnothing,
  \end{cases}
\]
where the empty submission is fully backfilled at $a_0$. The receiver action
under full within-camp reallocation is
\begin{equation}
 a_n^{R}(S_1,S_2)
 =\frac{\eta_n a_0
       +r_n\overline X_{1,n}(S_1)
       +(1-r_n)\overline X_{2,n}(S_2)}{1+\eta_n}.
 \label{eq:reallocated-action}
\end{equation}
Unlike \eqref{eq:action}, the denominator and each camp's realized exposure
are now independent of how much evidence that camp withholds.

\begin{proposition}[Fungible-exposure boundary]
\label{prop:fungible}
Under the baseline evidence assumptions, the disclosure game induced by
\eqref{eq:reallocated-action} has a Nash equilibrium for every finite pool,
and all pure-strategy equilibria induce the same action
\begin{equation}
 c_n^R
 =\frac{\eta_n a_0+r_n M_{1,n}+(1-r_n)m_{2,n}}{1+\eta_n},
 \quad
 M_{1,n}=\max\bigl\{a_0,\max_{i\in I_1}X_i\bigr\},
 \quad
 m_{2,n}=\min\bigl\{a_0,\min_{i\in I_2}X_i\bigr\}.
 \label{eq:reallocated-finite-outcome}
\end{equation}
Except for ties and the option of an empty submission when $a_0$ is best, the
Sender 1 submits only its sample maximum and Sender 2 only
its sample minimum.

If $\eta_n\to0$, $r_n\to r\in(0,1)$,
$\max_i v_{ni}\to0$, and the two camp sizes diverge, then
\begin{equation}
 c_n^R\xrightarrow{p}c^R(r)
 =r\overline x+(1-r)\underline x,
 \label{eq:reallocated-limit}
\end{equation}
whereas the complete-disclosure action converges in probability to $\theta$.
Consequently the unique state-contingent implementing share is
\begin{equation}
 r_R^*(\theta)
 =\frac{\theta-\underline x}{\overline x-\underline x}.
 \label{eq:reallocated-share}
\end{equation}
If the state can take two distinct values in the common support, no fixed
share implements both. Attention changes the mechanical mixture of the two
extremes but not either camp's disclosure rule, so the representation
paradox of Proposition~\ref{prop:paradox} does not survive full within-camp
reallocation.
\end{proposition}

\begin{proof}
The denominator in \eqref{eq:reallocated-action} is fixed and the two
senders' choices are additively separable. Among nonempty subsets, a
weighted average is no larger than the sample maximum and no smaller than
the sample minimum; either bound is attained by submitting a corresponding
singleton. The empty set yields $a_0$. Thus every upward best response has
average $M_{1,n}$, every downward best response has average $m_{2,n}$, and
\eqref{eq:reallocated-finite-outcome} follows. Any pair of such best
responses is a Nash equilibrium.

Positive density on the interior of $\mathcal X$ and divergent camp sizes
give
$M_{1,n}\xrightarrow{p}\overline x$ and
$m_{2,n}\xrightarrow{p}\underline x$, proving
\eqref{eq:reallocated-limit}. Under complete disclosure, diffuseness of the
within-camp weights and unbiasedness give
$\overline X_{1,n}(I_1)\xrightarrow{p}\theta$ and
$\overline X_{2,n}(I_2)\xrightarrow{p}\theta$. Finally,
$c^R(r)=\theta$ if and only if \eqref{eq:reallocated-share} holds; strict
monotonicity in $r$ gives uniqueness. Because this share is strictly
increasing in $\theta$, it cannot be constant at two distinct states.
\end{proof}

This result shows why item-level commitment matters. When exposure cannot be
transferred across items, a sender trades off frequency against extremity,
which generates tail disclosure and moment balance. With fully fungible
capacity, each camp instead concentrates on an extreme, and accuracy depends
on the endpoints of the support. Neutral backfilling in
Proposition~\ref{prop:neutral-backfill} leaves item-level exposure
nontransferable and therefore preserves the tail-disclosure mechanism.

\input{partial-reallocation-supplement}

\section{Proofs}
\label{app:proofs}

\subsection{Proof of Proposition~\ref{prop:selection-naive}}
\label{app:proof:selection-naive}

\begin{proof}
Under the receiver's independence belief, the event $S$ contributes no
state-dependent likelihood term. Gaussian updating gives posterior precision
$\tau_n(\eta_n+\sum_{i\in S}v_{ni})$ and precision-weighted posterior numerator
$\tau_n(\eta_na_0+\sum_{i\in S}v_{ni}X_i)$. Their ratio is $a_n(S)$, which is
the posterior mean and hence the unique quadratic-loss Bayes action.
\end{proof}

\subsection{Proof of Theorem~\ref{thm:finite}}
\label{app:proof:finite}

\begin{proof}
\emph{Uniqueness of the zero.} $H_{n}$ is continuous, and strictly
decreasing because the anchor term is strictly decreasing while each
$(X_{i}-c)^{+}$ is nonincreasing and each $-(c-X_{i})^{+}$ is
nonincreasing in $c$. At the endpoints,
$H_{n}(\underline x)\geq\eta_{n}(a_0-\underline x)>0$ and
$H_{n}(\overline x)\leq\eta_{n}(a_0-\overline x)<0$.

\emph{The tail profile is an equilibrium.} Set
$S_1^{*}=\{i\in I_1:X_i>c_n\}$ and
$S_2^{*}=\{j\in I_2:X_j<c_n\}$. Fix $S_2^{*}$ and consider
Sender 1's problem: choose $S_1\subseteq I_1$ to maximize
\[
  \frac{A+\sum_{i\in S_1}v_{ni}X_{i}}{B+\sum_{i\in S_1}v_{ni}},
  \qquad
  A=\eta_{n}a_0+\sum_{j\in S_2^{*}}v_{nj}X_{j},
  \quad
  B=\eta_{n}+\sum_{j\in S_2^{*}}v_{nj}.
\]
For any candidate value $y$ of this ratio, including an item with
$X_{i}>y$ strictly raises the ratio and excluding an item with $X_{i}<y$
strictly raises it. Hence at an optimum with value $y_{1}$, the chosen set
is $\{i:X_{i}>y_{1}\}$ up to ties, and $y_{1}$ satisfies
$G_{1}(y_{1})=0$, where
\[
 G_{1}(y)=A-By+\sum_{i\in I_1}v_{ni}(X_{i}-y)^{+}.
\]
The function $G_{1}$ is strictly decreasing because $B>0$. Evaluating it
at the proposed equilibrium value, rather than identifying it with
$H_n(y)$ away from that value, gives
\begin{align*}
 G_{1}(c_n)
 &=\eta_n(a_0-c_n)
   +\sum_{j\in S_2^*}v_{nj}(X_j-c_n)
   +\sum_{i\in I_1}v_{ni}(X_i-c_n)^+\\
 &=H_n(c_n)=0,
\end{align*}
where the second equality uses
$S_2^*=\{j\in I_2:X_j<c_n\}$. Thus Sender 1's unique best-response value is
$y_1=c_n$, and its optimal set is $S_1^*$ up to ties.

For Sender 2, fix $S_1^*$ and define
\[
 \widetilde A=\eta_n a_0+\sum_{i\in S_1^*}v_{ni}X_i,
 \qquad
 \widetilde B=\eta_n+\sum_{i\in S_1^*}v_{ni}.
\]
Sender 2 minimizes the corresponding ratio. Its optimal set is
$\{j\in I_2:X_j<y_2\}$ up to ties, and its value satisfies
\[
 G_2(y_2)=0,
 \qquad
 G_2(y)=\widetilde A-\widetilde B y
          -\sum_{j\in I_2}v_{nj}(y-X_j)^+.
\]
Again $G_2$ is strictly decreasing, and
\[
 G_2(c_n)
 =\eta_n(a_0-c_n)
  +\sum_{i\in I_1}v_{ni}(X_i-c_n)^+
  -\sum_{j\in I_2}v_{nj}(c_n-X_j)^+
 =H_n(c_n)=0.
\]
Therefore $y_2=c_n$, and Sender 2's optimal sets coincide with $S_2^*$ up to
ties. Neither sender can profit by deviating, so
$(S_1^*,S_2^*)$ is a Nash equilibrium with outcome $c_n$.

\emph{All pure-strategy equilibria.} Let $(S_1,S_2)$ be any
pure-strategy Nash equilibrium with
induced action $y$. If some $i\in I_1$ with $X_{i}>y$ were excluded,
adding it would strictly raise the action; if some $i\in I_1$ with
$X_{i}<y$ were included, removing it would strictly raise the action.
Hence $S_1=\{i\in I_1:X_{i}>y\}$ up to ties, and symmetrically
$S_2=\{i\in I_2:X_{i}<y\}$. The induced action then satisfies
$H_{n}(y)=0$, so $y=c_{n}$.

\emph{All mixed-strategy equilibria.} For disclosure sets $S_1$ and $S_2$,
write
\[
 D_n(S_1,S_2)=\eta_n+\sum_{i\in S_1}v_{ni}
                         +\sum_{j\in S_2}v_{nj}>0.
\]
If $S_1\in\mathcal T_1(c_n)$, then, for every $S_2\subseteq I_2$,
\begin{align*}
 D_n(S_1,S_2)\bigl[a_n(S_1\cup S_2)-c_n\bigr]
 &=\eta_n(a_0-c_n)
   +\sum_{i\in S_1}v_{ni}(X_i-c_n)
   +\sum_{j\in S_2}v_{nj}(X_j-c_n)\\
 &\geq \eta_n(a_0-c_n)
   +\sum_{i\in I_1}v_{ni}(X_i-c_n)^+
   -\sum_{j\in I_2}v_{nj}(c_n-X_j)^+\\
 &=H_n(c_n)=0.
\end{align*}
The inequality follows because Sender 1 includes every strictly positive
term, while the smallest contribution from Sender 2 is obtained by including
every strictly negative term and no strictly positive term. Thus every set in
$\mathcal T_1(c_n)$ guarantees an action weakly above $c_n$. Symmetrically,
every set in $\mathcal T_2(c_n)$ guarantees an action weakly below $c_n$,
whatever Sender 1 discloses. When both sets belong to their respective tail
families, the action is therefore exactly $c_n$. It follows that any pair of
mixed strategies supported on $\mathcal T_1(c_n)$ and
$\mathcal T_2(c_n)$ is a Nash equilibrium.

Conversely, let $(\sigma_1,\sigma_2)$ be a mixed-strategy Nash equilibrium
and let $V$ be its expected receiver action. A pure set in
$\mathcal T_1(c_n)$ guarantees at least $c_n$, whereas a pure set in
$\mathcal T_2(c_n)$ guarantees at most $c_n$. The two equilibrium deviation
inequalities therefore imply $V\geq c_n$ and $V\leq c_n$, so $V=c_n$.
Fix any $\bar S_2\in\mathcal T_2(c_n)$. Sender 2's equilibrium inequality and
the preceding upper bound give
\[
 c_n=V\leq
 \mathbb E_{S_1\sim\sigma_1}a_n(S_1\cup\bar S_2)
 \leq c_n.
\]
The last inequality is strict for every $S_1\notin\mathcal T_1(c_n)$:
relative to the upper tail, such a set omits a strictly positive term or
includes a strictly negative one. Hence
$\sigma_1(\mathcal T_1(c_n))=1$. Fixing any
$\bar S_1\in\mathcal T_1(c_n)$ and applying the symmetric argument yields
$\sigma_2(\mathcal T_2(c_n))=1$. This proves both necessity and sufficiency
of~\eqref{eq:mixed-equilibrium-support}; because every supported pure profile
produces $c_n$, the action equals $c_n$ with probability one.
\end{proof}

\subsection{Proof of Proposition~\ref{prop:ordinal}}
\label{app:proof:ordinal}

\begin{proof}
Fix any mixed-strategy Nash equilibrium, with nonempty finite supports
$\mathcal A$ and $\mathcal B$. For $S\in\mathcal A$ and
$T\in\mathcal B$, let
\[
 P(S)=\sum_{i\in S}v_{ni}X_i,\quad W(S)=\sum_{i\in S}v_{ni},\qquad
 Q(T)=\sum_{j\in T}v_{nj}X_j,\quad V(T)=\sum_{j\in T}v_{nj}.
\]
Define
\[
 U_S(z)=\eta_n(a_0-z)+P(S)-zW(S),\qquad
 L_T(z)=Q(T)-zV(T).
\]
The function
\[
 F(z)=\max_{S\in\mathcal A}U_S(z)+\min_{T\in\mathcal B}L_T(z)
\]
is continuous and strictly decreasing: each $U_S$ decreases at rate at
least $\eta_n>0$, whereas each $L_T$ is nonincreasing. Moreover,
$F(\underline x)>0>F(\overline x)$. Let $z$ be its unique zero.
At this $z$, select $S^+$ and $S^-$ maximizing and minimizing $U_S(z)$
over $\mathcal A$, and $T^+$ and $T^-$ maximizing and minimizing $L_T(z)$
over $\mathcal B$. The identity
\[
 [\eta_n+W(S)+V(T)]\,[a_n(S\cup T)-z]=U_S(z)+L_T(z)
\]
implies that row $S^+$ produces actions weakly above $z$, whereas column
$T^-$ produces actions weakly below $z$.

If $a_n(S^-\cup T^+)\leq z$, row $S^-$ produces actions weakly below
$z$ against every supported opponent set. Rows $S^-$ and $S^+$ are both
best responses in the mixed equilibrium. Strict increase of $u_1$ and
positive probability of every $T\in\mathcal B$ therefore force
\[
 a_n(S^-\cup T)=a_n(S^+\cup T)=z
 \qquad(T\in\mathcal B).
\]
The residual identity then yields $U_{S^-}(z)=U_{S^+}(z)$ and a common
value of $L_T(z)$ on $\mathcal B$. Since $S^-$ and $S^+$ attain the
extreme values of $U$, every supported pure profile produces $z$.

If instead $a_n(S^-\cup T^+)\geq z$, column $T^+$ produces actions
weakly above $z$ against every supported row. Columns $T^-$ and $T^+$
are both best responses. Strict decrease of $u_2$ and support indifference
force both columns to produce $z$ in every supported row. The residual
identity again implies that every supported pure profile produces $z$.
The equality case is covered by either argument.

Suppose a supported upward set omits an item with $X_i>z$. Adding the
item raises the action strictly above $z$ against every supported opponent
set, contradicting optimality. Removing an included item with $X_i<z$
gives the same contradiction. The analogous deviations for Sender~2 show
that every supported set belongs to its favorable-tail family at $z$.
Substitution into the action formula gives $H_n(z)=0$, hence $z=c_n$.

Conversely, against any downward tail at $c_n$, every upward disclosure
set produces an action at most $c_n$, and an upward tail achieves $c_n$.
Against any upward tail at $c_n$, every downward disclosure set produces
an action at least $c_n$, and a downward tail achieves $c_n$.
These inequalities follow by including all positive centered terms for
Sender~1 and all negative centered terms for Sender~2, as in the proof of
Theorem~\ref{thm:finite}. Strictly monotone preferences therefore make any
pair of mixtures supported on the two tail families mutual best responses.
In particular, the strict-tail pure profile establishes existence.

For a fixed deterministic opponent set, each sender's complete ranking
of disclosure sets also coincides with its ranking under linear payoffs.
This last observation preserves the stated pure-equilibrium conclusions
in the neutral-imputation and reallocation extensions; it does not assert
invariance of best responses to arbitrary lotteries or of private-evidence
Bayesian equilibria.
\end{proof}

\subsection{Proof of Proposition~\ref{thm:largepool}}
\label{app:proof:largepool}

\begin{proof}
Strict monotonicity: for $c<c'$,
\[
  H(c;r,\theta)-H(c';r,\theta)
  =\int_{c}^{c'}
  \big[r\,\Prob_{\theta}(X_1>s)+(1-r)\,\Prob_{\theta}(X_2<s)\big]\,ds
  >0,
\]
since the integrand is strictly positive on the interior by
Assumption~\ref{ass:evidence}. At the endpoints,
$H(\underline x)=r(\theta-\underline x)>0$ and
$H(\overline x)=-(1-r)(\overline x-\theta)<0$, so the zero exists, is
unique, and is interior. Lemma~\ref{lem:uniform} in the Appendix shows
$\sup_{c}|H_{n}(c)-H(c;r,\theta)|\xrightarrow{p}0$; since $H$ is strictly
decreasing and continuous, $H(c^{*}\pm\varepsilon)$ straddle zero, so with
probability tending to one $H_{n}(c^{*}-\varepsilon)>0>H_{n}(c^{*}+\varepsilon)$
and hence $|c_{n}-c^{*}|<\varepsilon$. Convergence of the disclosed
fractions follows from $c_{n}\xrightarrow{p}c^{*}$, the continuity of
$F_{1,\theta}$ and $F_{2,\theta}$, and the law of large numbers applied to
the two tail indicators at $c=c^{*}\pm\varepsilon$. Differentiability: $H$ is
$C^{1}$ in $(c,r)$ with
$\partial H/\partial c=-[r\Prob_{\theta}(X_1>c)+(1-r)\Prob_{\theta}(X_2<c)]<0$
and $\partial H/\partial r=\E_{\theta}(X_1-c)^{+}+\E_{\theta}(c-X_2)^{+}>0$;
the implicit function theorem gives \eqref{eq:cs}.
\end{proof}

\subsection{Proof of Corollary~\ref{thm:countervailing}}
\label{app:proof:countervailing}

\begin{proof}
$H(\theta;r,\theta)=r\,\E_{\theta}(X-\theta)^{+}-(1-r)\,\E_{\theta}(\theta-X)^{+}
=(2r-1)M_{0}(\theta)$. Since $H$ is strictly decreasing in $c$
(Proposition~\ref{thm:largepool}), its zero lies at $\theta$ when $r=1/2$,
strictly above when $r>1/2$, and strictly below when $r<1/2$.
\end{proof}

\subsection{Proof of Proposition~\ref{prop:paradox}}
\label{app:proof:paradox}

\begin{proof}
Fix $\theta$ and write $c=c^*(r,\theta)$. Equation~\eqref{eq:cs}
gives $c_r>0$. Since both densities are positive at the interior cutoff,
the chain rule gives
\[
 p_1'=-f_{1,\theta}(c)c_r<0,
 \qquad
 p_2'=f_{2,\theta}(c)c_r>0.
\]
For each camp, differentiating its truncated conditional mean yields
\begin{align*}
 \frac{d}{dc}\E_\theta[X_1\mid X_1>c]
 &=\frac{f_{1,\theta}(c)}{1-F_{1,\theta}(c)}
   \big(\E_\theta[X_1\mid X_1>c]-c\big)>0,\\
 \frac{d}{dc}\E_\theta[X_2\mid X_2<c]
 &=\frac{f_{2,\theta}(c)}{F_{2,\theta}(c)}
   \big(c-\E_\theta[X_2\mid X_2<c]\big)>0.
\end{align*}
Both denominators are positive, and the strict inequalities follow from
positive densities on either side of $c$. Multiplying by $c_r$ proves
the two formulas for $\mu_1'$ and $\mu_2'$. These statements concern
disclosure fractions within each eligible pool, not the derivatives of
the exposure masses $rp_1$ and $(1-r)p_2$.
\end{proof}

\subsection{Proof of Proposition~\ref{prop:amplification}}
\label{app:proof:amplification}

\begin{proof}
Differentiate $H(c;r)=D(c;r)[m(c;r)-c]$ from
Remark~\ref{rem:identity}. Since $H_c=-D$ and
$D_c=(1-2r)f_\theta(c)$,
\[
 -D=D_c(m-c)+D(m_c-1).
\]
Rearranging gives
\[
  \frac{\partial m(c;r)}{\partial c}
  =
  (2r-1)\,\frac{f_{\theta}(c)}{D(c;r)}\,\big[m(c;r)-c\big].
\]
For $r>1/2$: $m(\theta;r)-\theta=(2r-1)M_{0}/D>0$ and
$c^{*}(r)>\theta$ by Corollary~\ref{thm:countervailing}. On
$(\theta,c^{*}(r))$ we have $H>0$, hence $m(c;r)>c$, hence
$\partial m/\partial c>0$: $m(\cdot;r)$ is strictly increasing between
$\theta$ and $c^{*}$. Therefore
$c^{*}(r)=m(c^{*}(r);r)>m(\theta;r)>\theta$.
\end{proof}

\subsection{Proof of Proposition~\ref{prop:extremization}}
\label{app:proof:extremization}

\begin{proof}
$c^{*}$ is increasing, so the limit as $r\uparrow1$ exists; call it
$\bar c$. If $\bar c<\overline x$, continuity of
$(c,r)\mapsto H(c;r,\theta)$ would give
$0=\lim H(c^{*}(r);r,\theta)=\E_{\theta}(X-\bar c)^{+}>0$, a
contradiction. The other limit is symmetric.
\end{proof}

\subsection{Proof of Theorem~\ref{thm:design}}
\label{app:proof:design}

\begin{proof}
$H(\theta;r,\theta)=r\,\Gamma_1(\theta)-(1-r)\,\Gamma_2(\theta)$. Since
$H(\cdot;r,\theta)$ is strictly decreasing
(Proposition~\ref{thm:largepool}), $c^{*}(r,\theta)=\theta$ if and only if
this expression vanishes, which is exactly $r=r^{*}(\theta)$. Strict
positivity of the loss elsewhere follows from strict monotonicity of
$c^{*}$ in $r$ \eqref{eq:cs}. If $\Gamma_2(\theta)/\Gamma_1(\theta)$ is state independent, then so
is the displayed implementing share, proving the final claim.
\end{proof}

\subsection{Proof of Theorem~\ref{thm:exante}}
\label{app:proof:exante}

\begin{proof}
Joint continuity and compactness imply continuity of $\mathcal L$ and
existence of a minimizer. The integrand is nonnegative. It vanishes almost
surely at a constant $r$ if and only if
$c^{*}(r,\theta)=\theta$ almost surely, which by
Theorem~\ref{thm:design} is equivalent to
$r=r^{*}(\theta)$ almost surely. This proves (i).

For (ii), if $r<r_{L}$, then $c^{*}(r,\theta)<\theta$ almost surely by
Theorem~\ref{thm:design}; since $\partial c^{*}/\partial r>0$, increasing
$r$ strictly lowers the loss until $r_{L}$ is reached. The symmetric
argument applies above $r_{U}$. If $r^*$ is not almost surely constant,
then at $r_L$ the distortion is weakly negative almost surely and strictly
negative on a set of positive probability, so the derivative of
$\mathcal L$ is strictly negative; symmetrically it is strictly positive at
$r_U$. Hence neither endpoint is optimal. For (iii), differentiate under the integral,
which is justified by Assumption~\ref{ass:design} and compactness; the
factor two cancels from the first-order condition. Strict convexity gives
uniqueness. For (iv), posterior minimization follows from iterated
expectations. A more informative signal expands the set of feasible
state-contingent decision rules (the platform can ignore the extra
information), so attainable loss cannot increase.
\end{proof}

\subsection{Proof of Corollary~\ref{cor:uniformimpl}}
\label{app:proof:uniformimpl}

\begin{proof}
$r^{*}(\theta)=[1+\Gamma_1(\theta)/\Gamma_2(\theta)]^{-1}$ is constant on
$\Theta_{0}$ if and only if the ratio is; apply part (i).
\end{proof}

\subsection{Proof of Proposition~\ref{prop:primitiveexante}}
\label{app:proof:primitiveexante}

\begin{proof}
For $X=\theta+\sigma Z$ with $Z\sim U[-1,1]$ and
$d=c-\theta\in[-\sigma,\sigma]$,
\[
 \E(X-c)^+=\frac{(\sigma-d)^2}{4\sigma},
 \qquad
 \E(c-X)^+=\frac{(\sigma+d)^2}{4\sigma}.
\]
The same monotonicity argument as in Theorem~\ref{thm:exante}(ii) places
every optimum between the two state-contingent shares: below both shares an
increase in $r$ reduces both state losses, and above both shares a decrease
does so. On that interval $1/k\leq t\leq k$; the restriction $k<\sqrt3$
keeps both state-specific roots inside the overlapping supports.
Taking the positive square root of the equilibrium balance equation and
solving for $d$ gives \eqref{eq:primitiveDL} and
\eqref{eq:primitiveDH}. Their zeros give the displayed state-contingent
shares.

The unknown-state loss in the odds coordinate is
\[
 \mathcal L_p(t)=(1-p)d_L(t)^2+p\,d_H(t)^2.
\]
Direct differentiation yields
\begin{equation}
 \frac{\mathcal L_p'(t)}{2s^2k^2(k^2+1)}
 =(1-p)\frac{kt-1}{(t+k)^3}
   +p\frac{t-k}{(kt+1)^3}.
 \label{eq:primitiveLossDerivative}
\end{equation}
The derivative is negative at $1/k$ and positive at $k$. On $(1/k,k)$ its
zero is equivalent to \eqref{eq:primitiveFOC}. Moreover,
\[
 \frac{d}{dt}\log R_k(t)
 =\frac{k}{kt-1}+\frac{3k}{kt+1}
  +\frac{1}{k-t}-\frac{3}{t+k}>0,
\]
because
$3k/(kt+1)-3/(t+k)=3(k^2-1)/[(kt+1)(t+k)]>0$.
Thus the zero is unique, and equating $\log R_k(t_p)$ to the prior log odds
proves the comparative static in $p$. Symmetry gives $R_k(1)=1$ and hence
$r^{EA}_{1/2}=1/2$. The zero is also simple. Indeed, write the right-hand
side of \eqref{eq:primitiveLossDerivative} as $Q(t)-P(t)$, where
\[
 Q(t)=(1-p)\frac{kt-1}{(t+k)^3},
 \qquad
 P(t)=p\frac{k-t}{(kt+1)^3}.
\]
At $t_p$, $Q(t_p)=P(t_p)>0$ and
\[
 \frac{Q(t)}{P(t)}=\frac{1-p}{p}R_k(t).
\]
Consequently
\begin{equation}
 \frac{\mathcal L_p''(t_p)}{2s^2k^2(k^2+1)}
 =Q'(t_p)-P'(t_p)
 =Q(t_p)\left.\frac{d}{dt}\log R_k(t)\right|_{t=t_p}>0.
 \label{eq:primitiveStrictCurvature}
\end{equation}

State observation permits the platform to choose the two zeros and attain
zero loss. At $t=1$ the distortions are equal and opposite, with magnitude
$sk(k-1)/(k+1)$, which proves \eqref{eq:primitiveVOI}; its derivative with
respect to $k$ is positive for $k>1$.

For the final claim, approximate each centered uniform density by a
symmetric $C^\infty$ mollification with the same mean, and mix it with a
vanishing weight on a smooth density that is strictly positive on
$\mathcal X$ and has mean $\theta_j$ (an exponential tilt on $\mathcal X$
provides one). Choose the approximations in mirror-image pairs across the
two states. The resulting densities are atomless, unbiased, strictly
positive on the interior of $\mathcal X$, and hence satisfy
Assumption~\ref{ass:evidence}.

Fix a compact parameter set $K$ as in the statement. The strict parameter
inequalities give a common positive distance between the relevant
equilibrium paths and the moving support edges. We may therefore choose the
mollifications so that, on a common neighborhood of those paths, the
truncated-moment maps and their derivatives through order two converge
uniformly. Their derivative with respect to the action is bounded away from
zero there. The implicit-function theorem with parameters then implies
$C^2$ convergence of the state-specific equilibrium-action paths, and thus
uniform $C^2$ convergence of the perturbed ex ante losses to
$\mathcal L_p$ in the odds coordinate.

It remains to verify that this convergence preserves global uniqueness.
By \eqref{eq:primitiveStrictCurvature} and compactness of $K$, there are a
common neighborhood of the minimizer path $t_p$ and a constant $m>0$ on
which $\mathcal L_p''\geq m$. Outside that neighborhood, but between the two
state-contingent policies, $\mathcal L_p'$ is bounded away from zero and has
the sign pointing toward $t_p$; this follows from the strict increase of
$R_k$. For every sufficiently close approximation, its second derivative
is therefore positive on the neighborhood and its first derivative retains
the two exterior signs. It has exactly one critical point in the relevant
interval, and that point is its unique global minimizer. Theorem~\ref{thm:exante}(ii)
rules out an optimum outside the interval between the perturbed
state-contingent policies. The minimizers and attained losses consequently
converge uniformly on $K$. The same strict inequalities remain valid under
small $C^2$ perturbations, which proves the final open-set claim.
\end{proof}

\subsection{Proof of Proposition~\ref{prop:finitepolicy}}
\label{app:proof:finitepolicy}

\begin{proof}
Proposition~\ref{thm:finitewelfare} implies, uniformly on $\mathcal R_0$,
\begin{equation}
 \mathcal W_{n,\mu}(r)
 =\mathcal L_\mu(r)+\frac{\mathcal Q_{\mu,\nu}(r)}{n}+o(n^{-1}).
 \label{eq:integratedrisk}
\end{equation}
For every realization, $c_n(r)$ is continuous in $r$: it is the unique zero
of a function continuous in $(c,r)$. Bounded support and dominated
convergence therefore make $\mathcal W_{n,\mu}$ continuous, so an exact
minimizer exists. Uniform convergence in \eqref{eq:integratedrisk} and the
unique minimizer of $\mathcal L_\mu$ give part (i) by the standard compact
argmin argument.

For part (ii), outside any neighborhood of $r_\mu$ the limiting loss is
bounded strictly above its minimum, whereas the $\mathcal Q_{\mu,\nu}/n$ term is
uniformly vanishing. Hence every minimizer of the second-order criterion is
eventually local. On a sufficiently small neighborhood,
$\mathcal L_\mu''$ is bounded away from zero and
$\mathcal Q_{\mu,\nu}''/n$ is uniformly negligible, so the criterion is strictly
convex there and its minimizer is unique. Its first-order condition is
\[
 0=\mathcal L_\mu'(\widetilde r_{n,\mu})
   +\frac{\mathcal Q_{\mu,\nu}'(\widetilde r_{n,\mu})}{n}.
\]
Taylor expansion first gives
$\widetilde r_{n,\mu}-r_\mu=O(n^{-1})$ and then yields
\eqref{eq:secondorderpolicy}. Expanding the criterion once more at $r_\mu$
and substituting that limit gives \eqref{eq:secondordervalue}.

For a point mass at $\theta$, Theorem~\ref{thm:design} gives
$r_\mu=r^*(\theta)$, and
\[
 \mathcal L_\mu''(r^*)=2\kappa(\theta)^2,
 \qquad
 \partial_r Q_{\nu}(r^*,\theta)
 =\partial_r V^C(r^*,\theta)+2\kappa(\theta)B^{C,\nu}(r^*,\theta).
\]
Substitution into \eqref{eq:secondorderpolicy} proves part (iii).
\end{proof}

\subsection{Proof of Proposition~\ref{prop:varcost}}
\label{app:proof:varcost}

\begin{proof}
$c^{*}=\theta$ by Corollary~\ref{thm:countervailing}, and
$D^{*}=\tfrac12\Prob_{\theta}(X>\theta)+\tfrac12\Prob_{\theta}(X<\theta)
=\tfrac12$ because $F_{\theta}$ is atomless. Write
$g_1=(X-\theta)^{+}$ and $g_2=(\theta-X)^{+}$, so that
$g_1^{2}+g_2^{2}=(X-\theta)^{2}$ and, by unbiasedness,
$\E g_1=\E g_2=\tfrac12\E|X-\theta|$. Hence
\[
  \operatorname{Var}(g_1)+\operatorname{Var}(g_2)
  =\E(X-\theta)^{2}-\tfrac12\big(\E|X-\theta|\big)^{2},
\]
so $\Omega=\tfrac12[\operatorname{Var}(g_1)+\operatorname{Var}(g_2)]$
and
$V^{C}=\Omega/D^{*2}=4\Omega
=2\E(X-\theta)^{2}-(\E|X-\theta|)^{2}$, while
$V^{F}=\operatorname{Var}_{\theta}(X)=\E(X-\theta)^{2}$. Subtracting,
$V^{C}-V^{F}=\E|X-\theta|^{2}-(\E|X-\theta|)^{2}
=\operatorname{Var}(|X-\theta|)$, which is nonnegative and zero only in
the degenerate case.
\end{proof}

\subsection{Uniform convergence}
\label{app:uniform}

\begin{lemma}
\label{lem:uniform}
Under the conditions of Proposition~\ref{thm:largepool},
$\sup_{c\in\mathcal X}\big|H_{n}(c)-H(c;r,\theta)\big|
\xrightarrow{p}0$.
\end{lemma}

\begin{proof}
Write $\bar H_{n}(c)=\E_{\theta}[H_{n}(c)]
=\eta_{n}(a_0-c)+r_{n}\E_{\theta}(X_1-c)^{+}
-(1-r_{n})\E_{\theta}(c-X_2)^{+}$. First,
\[
  \sup_{c}\big|\bar H_{n}(c)-H(c;r,\theta)\big|
  \leq
  \eta_{n}\,\mathrm{diam}(\mathcal X)
  +2\,|r_{n}-r|\,\mathrm{diam}(\mathcal X)
  \longrightarrow0 .
\]
Second, for fixed $c$, $H_{n}(c)-\bar H_{n}(c)=\sum_{i}v_{ni}\xi_{i}(c)$
with $\xi_{i}(c)$ independent, mean zero, and
$|\xi_{i}(c)|\leq\mathrm{diam}(\mathcal X)$, so
$\mathrm{Var}(H_{n}(c))\leq\mathrm{diam}(\mathcal X)^{2}\sum_{i}v_{ni}^{2}
\leq\mathrm{diam}(\mathcal X)^{2}\max_{i}v_{ni}\to0$: pointwise
convergence in probability. Third, every function involved is Lipschitz in
$c$ with constant at most $1+\eta_{n}$, so convergence on a finite
$\delta$-grid extends to uniform convergence with error
$2(1+\eta_{n})\delta$; letting $\delta\to0$ along a diagonal completes the
proof.
\end{proof}

\subsection{Proof of Proposition~\ref{thm:finitewelfare}}
\label{app:finitewelfare}

The only nonstandard step is control of the random kink at the equilibrium
cutoff. Second-order bias and mean-squared-error expansions for ordinary
asymmetric least squares are developed by \citet{leeullahwang2019}; here the
score combines two distinct camp samples, so the following lemma establishes
the required two-sample remainder directly.

\begin{lemma}[Uniform local root expansion]
\label{lem:localroot}
Under Assumption~\ref{ass:finitepool}(i), let
$\delta_n=c_n-c^{*}(r,\theta)$ and define $e_n$ and $u_n$ as below. Uniformly
on $\mathcal K$,
\[
 \|\delta_n\|_{L^q}=O(n^{-1/2})\quad\text{for every fixed }q<\infty,
 \qquad
 \left\|\delta_n-\frac{e_n}{D^{*}}\right\|_{L^2}=O(n^{-1}),
\]
and
\begin{equation}
  \delta_n
  =\frac{e_n}{D^{*}}
   +\frac{u_ne_n}{(D^{*})^{2}}
   +\frac{H_{cc}(c^{*};r,\theta)e_n^{2}}{2(D^{*})^{3}}
   +\frac{\eta_n(a_0-c^*)}{D^*}
   +\mathfrak r_n,
  \qquad
  \sup_{(r,\theta)\in\mathcal K}n\E_\theta|\mathfrak r_n|\longrightarrow0.
  \label{eq:uniformroot}
\end{equation}
\end{lemma}

\begin{proof}
Suppress $(r,\theta)$ and write $P_{n1}$ and $P_{n2}$ for the two empirical
measures, with population laws $P_1$ and $P_2$. At $c^*$ set
\begin{align*}
 g_1(x)&=(x-c^*)^+,& \iota_1(x)&=\ind{x>c^*},\\
 g_2(x)&=(c^*-x)^+,& \iota_2(x)&=\ind{x<c^*},
\end{align*}
and define the centered score and slope fluctuations
\begin{align*}
 e_n&=r(P_{n1}-P_1)g_1-(1-r)(P_{n2}-P_2)g_2,\\
 u_n&=-r(P_{n1}-P_1)\iota_1-(1-r)(P_{n2}-P_2)\iota_2.
\end{align*}
Then
\[
 D_n=rP_{n1}\iota_1+(1-r)P_{n2}\iota_2
     =D^{*}-u_n.
\]
For $t\in\R$, define the crossing losses
\begin{align*}
 \ell_1(x,t)
 &= (x-c^{*}-t)^+-(x-c^{*})^+ +t\ind{x>c^{*}},\\
 \ell_2(x,t)
 &= (c^{*}+t-x)^+-(c^{*}-x)^+ -t\ind{x<c^{*}}.
\end{align*}
For each $k\in\{1,2\}$ they satisfy
$|\ell_k(x,t)|\leq |t|\ind{|x-c^{*}|\leq|t|}$. If the anchor is omitted,
the empirical excess-advocacy function therefore has the exact expansion
\begin{equation}
 H_n^0(c^{*}+t)
 =e_n-D_nt+rP_{n1}\ell_1(\cdot,t)
              -(1-r)P_{n2}\ell_2(\cdot,t).
 \label{eq:exactkink}
\end{equation}
Assumption~\ref{ass:finitepool}(i) and Taylor's theorem imply, uniformly for
$|t|\leq\varepsilon$,
\begin{align*}
 P_1\ell_1(\cdot,t)&=\tfrac12f_{1,\theta}(c^{*})t^2+O(|t|^3),\\
 P_2\ell_2(\cdot,t)&=\tfrac12f_{2,\theta}(c^{*})t^2+O(|t|^3).
\end{align*}
Moreover, the two one-dimensional classes indexed by $t$ are VC-subgraph
classes with envelope
$F_a(X)=a\ind{|X-c^{*}|\leq a}$ on $|t|\leq a$. Uniformly on
$\mathcal K$, $P_k F_a^2\leq 2Ma^3$ for $k\in\{1,2\}$ and
$\|F_a\|_\infty\leq a$.
The maximal-moment bound for bounded VC-subgraph classes
\citep[Section~2.14]{vanderVaartWellner1996} consequently gives, for either
sender and $0<a\leq\varepsilon$,
\begin{equation}
 \max_{k\in\{1,2\}}
 \left\|\sup_{|t|\leq a}
 |(P_{nk}-P_k)\ell_k(\cdot,t)|\right\|_{L^2}
 \leq K\max_{k\in\{1,2\}}
       \left(\frac{a^{3/2}}{\sqrt{n_k}}+\frac{a}{n_k}\right)
 \leq K'\left(\frac{a^{3/2}}{\sqrt n}+\frac{a}{n}\right),
 \label{eq:crossingmax}
\end{equation}
where the second inequality uses the lower bound on $n_k/n$, and $K,K'$
are independent of $(r,\theta)\in\mathcal K$. The classes are translations
of the same one-dimensional VC-subgraph class as $c^*$ varies; the density
bound controls their localized envelopes uniformly.

For completeness, the root is localized at the parametric rate. Put
$\alpha_{k,n}=n_k/n$ for $k\in\{1,2\}$. Because the sender fractions and $r$ on
$\mathcal K$ are bounded away from zero,
the squared weights in the centered part of $H_n(c)$ sum to at most $K_0/n$
for a common constant $K_0$. The population function has derivative
$-D^{*}$ at $c^{*}$, and the uniform density bound and $D^{*}\geq d_0$
imply, after reducing $\varepsilon$ if needed, that
$|H(c^{*}+t;r,\theta)|\geq d_0|t|/2$ for $|t|\leq\varepsilon$.
For $t\geq n^{-1/2}$, the $O(n^{-1})$ anchor is smaller than this separation.
Monotonicity of $H_n$, followed by weighted Hoeffding inequalities at
$c^{*}\pm t$, therefore yields constants $A,k>0$ such that
\begin{equation}
 \sup_{(r,\theta)\in\mathcal K}
 \Prob_\theta(|\delta_n|>t)
 \leq A\exp(-knt^2),
 \qquad n^{-1/2}\leq t\leq\varepsilon,
 \label{eq:roottail}
\end{equation}
and the probability of $|\delta_n|>\varepsilon$ is $O(e^{-kn})$.
Since both $c_n$ and $c^*$ lie in the fixed compact interval $\mathcal X$,
integrating this tail bound proves
$\|\delta_n\|_{L^q}=O(n^{-1/2})$ for every fixed finite $q$.

Take $b_n=n^{-1/2}\log n$ and let $\mathcal B_n=\{|\delta_n|\leq b_n\}$. On
$\mathcal B_n$,
the empirical part of each crossing loss is bounded in $L^2$ by
\eqref{eq:crossingmax}, which at $b_n$ is $o(n^{-1})$; on $\mathcal B_n^c$ its
bounded-support contribution is $o(n^{-1})$ because
\eqref{eq:roottail} makes $\Pr(\mathcal B_n^c)$ smaller than every power of
$n$.  Taking the supremum over every $|t|\leq b_n$ before substituting the
random value $t=\delta_n$ is what avoids any independence requirement between
the root and the empirical process.
The population Taylor remainder is $O(|\delta_n|^3)$, whose $L^2$ norm is
$O(n^{-3/2})$ by the sixth-moment bound just established. Consequently,
\begin{equation}
 rP_{n1}\ell_1(\cdot,\delta_n)
 -(1-r)P_{n2}\ell_2(\cdot,\delta_n)
 =\tfrac12H_{cc}(c^{*};r,\theta)\delta_n^2+\zeta_n,
 \qquad \|\zeta_n\|_{L^2}=o(n^{-1}),
 \label{eq:kinkremainder}
\end{equation}
uniformly on $\mathcal K$. Substituting $t=\delta_n$ into
\eqref{eq:exactkink} and restoring the anchor gives
\[
 D^{*}\delta_n
 =e_n+u_n\delta_n
  +\tfrac12H_{cc}(c^{*};r,\theta)\delta_n^2
  +\zeta_n+\eta_n(a_0-c^{*}-\delta_n).
\]
Bounded independent summands give
$\|e_n\|_{L^q}+\|u_n\|_{L^q}=O(n^{-1/2})$ for every fixed finite $q$.
The last display, the moment bound for $\delta_n$, and
$D^{*}\geq d_0$ first imply
\[
 \|u_n\delta_n\|_{L^2}
 \leq\|u_n\|_{L^4}\|\delta_n\|_{L^4}=O(n^{-1}),
 \qquad
 \|\delta_n^2\|_{L^2}=\|\delta_n\|_{L^4}^2=O(n^{-1}).
\]
Together with $\|\zeta_n\|_{L^2}=o(n^{-1})$ and
$\eta_n=O(n^{-1})$, these inequalities give
$\|\delta_n-e_n/D^{*}\|_{L^2}=O(n^{-1})$.
Write $\rho_n=\delta_n-e_n/D^*$, so $\|\rho_n\|_{L^2}=O(n^{-1})$.
Replacing $u_n\delta_n$ by $u_ne_n/D^*$ changes its expectation in absolute
value by at most
$\|u_n\|_{L^2}\|\rho_n\|_{L^2}=O(n^{-3/2})$. Likewise,
\[
 \E\left|\delta_n^2-\frac{e_n^2}{(D^*)^2}\right|
 \leq \|\rho_n\|_{L^2}
       \left\|\delta_n+\frac{e_n}{D^*}\right\|_{L^2}
 =O(n^{-3/2}).
\]
The $L^1$ norm of $\zeta_n$ is $o(n^{-1})$. Retain
$\eta_n(a_0-c^*)/D^*$ explicitly; the remaining anchor term satisfies
$\eta_n\E|\delta_n|=O(n^{-3/2})$ uniformly because
$n\eta_n\to \nu$ and $\|\delta_n\|_{L^1}=O(n^{-1/2})$.
Dividing by $D^*\geq d_0$ and collecting these errors gives
\eqref{eq:uniformroot}, including
$\sup_{\mathcal K}n\E|\mathfrak r_n|\to0$.
\end{proof}

\begin{proof}[Proof of Proposition~\ref{thm:finitewelfare}]
Fix $(r,\theta)$ and abbreviate $c^{*}=c^{*}(r,\theta)$ and
$D^{*}=D^{*}(r,\theta)$. Put
\[
  g_1(x)=(x-c^{*})^{+},
  \qquad
  g_2(x)=(c^{*}-x)^{+}.
\]
The centered empirical excess advocacy at the population zero is
\[
  e_{n}
  =r\big(\bar g_1-\E g_1\big)
   -(1-r)\big(\bar g_2-\E g_2\big).
\]
The two-sample central limit theorem gives
$\sqrt n e_n\Rightarrow N(0,\Omega)$, while
$\sqrt n\eta_n(a_0-c^*)\to0$ under \eqref{eq:anchorrate}.
Uniform convergence of the empirical tail probabilities in a neighborhood
of $c^{*}$ gives the Bahadur representation
\begin{equation}
  c_n-c^{*}=\frac{e_n}{D^{*}}+o_p(n^{-1/2}),
  \label{eq:bahadur}
\end{equation}
because $\partial H(c^{*};r,\theta)/\partial c=-D^{*}$. This proves
\eqref{eq:cltcurated}. Complete disclosure is the weighted sum of the two
sample means up to an order-$1/n$ anchor; the ordinary two-sample central
limit theorem gives \eqref{eq:cltfull}.

For the second-order statements, let
\[
  u_n
  =-r(\bar\iota_1-p_1)-(1-r)(\bar\iota_2-p_2),
  \quad
  \iota_1=\ind{X_1>c^{*}},\quad \iota_2=\ind{X_2<c^{*}}.
\]
Thus $u_n$ is the centered empirical derivative of $H_n$ at $c^{*}$, and
Lemma~\ref{lem:localroot} gives \eqref{eq:uniformroot}, where
$H_{cc}(c^{*};r,\theta)=r f_{1,\theta}(c^{*})
-(1-r)f_{2,\theta}(c^{*})$. Independence of the two pools and
$\operatorname{Cov}(\iota_1,g_1)=U(c^{*})(1-p_1)$,
$\operatorname{Cov}(\iota_2,g_2)=L(c^{*})(1-p_2)$ give the exact
finite-$n$ identities
\[
 \begin{split}
 n\E[u_ne_n]
 &=-\frac{r^2}{\alpha_{1,n}}U(c^*)(1-p_1)
   +\frac{(1-r)^2}{\alpha_{2,n}}L(c^*)(1-p_2)
   \equiv C_n(r,\theta),\\
 n\E[e_n^2]
 &=\frac{r^2}{\alpha_{1,n}}\operatorname{Var}(g_1)
   +\frac{(1-r)^2}{\alpha_{2,n}}\operatorname{Var}(g_2)
   \equiv\Omega_n(r,\theta),
 \end{split}
\]
where $\alpha_{k,n}=n_k/n$ for $k\in\{1,2\}$. Hence, uniformly on $\mathcal K$,
$C_n\to C$ and $\Omega_n\to\Omega$, because
$\alpha_{k,n}\to\alpha_k$ for both $k$ and all displayed moments are uniformly
bounded by compact support. Taking expectations in \eqref{eq:uniformroot}
therefore gives
\[
 n\E[c_n-c^*]
 =\frac{C_n}{(D^*)^2}
  +\frac{H_{cc}(c^*;r,\theta)\Omega_n}{2(D^*)^3}
  +\frac{n\eta_n(a_0-c^*)}{D^*}+o(1)
 \longrightarrow B^{C,\nu}(r,\theta),
\]
uniformly, which proves \eqref{eq:biasexpansion}. Moreover,
$\|c_n-c^*-e_n/D^*\|_{L^2}=O(n^{-1})$ implies
\[
 n\E[(c_n-c^*)^2]
 =\frac{\Omega_n}{(D^*)^2}+o(1)
 \longrightarrow V^C(r,\theta)
\]
uniformly. Finally,
$(c_n-\theta)^2=(c^*-\theta)^2
+2(c^*-\theta)(c_n-c^*)+(c_n-c^*)^2$ proves
\eqref{eq:riskexpansion} and its uniform remainder claim.
\end{proof}

\subsection{Derivative risk expansion and the exact finite-pool policy}
\label{app:exactpolicy}

For this subsection write
\(
 R_n(r,\theta)=\E_\theta[(c_n(r)-\theta)^2]
\)
for conditional exact risk. The next argument upgrades the value expansion
in Proposition~\ref{thm:finitewelfare} to the local $C^1$ expansion needed for
the exact optimizer.

For $a,z\in\operatorname{int}\mathcal X$, define the triangular crossing
loss
\begin{equation}
 \psi(x;z,a)
 =|x-z|\ind{\min\{a,z\}<x<\max\{a,z\}}.
 \label{eq:triangular-crossing}
\end{equation}
Restoring the cutoff argument in Lemma~\ref{lem:localroot} shows that both
camp-specific crossing losses are instances of the same function:
\begin{equation}
 \ell_1(x,z-a;a)=\ell_2(x,z-a;a)=\psi(x;z,a)
 \quad\text{almost surely}.
 \label{eq:crossing-identity}
\end{equation}

\begin{lemma}[$C^1$ cutoff-crossing bound]
\label{lem:C1crossing}
Let $\mathcal K_1$ be a compact subset of the interior of the design domain
on which the interiority and smoothness conditions in
Assumption~\ref{ass:finitepool} hold. For $k\in\{1,2\}$, set
\begin{equation}
 T_{nk}(r,\theta)
 =\E_\theta\big[(P_{nk}-P_k)
       \psi(\cdot;c_n(r),c^*(r,\theta))\big].
 \label{eq:crossing-term}
\end{equation}
Then $T_{nk}$ is continuously differentiable in $r$ and, uniformly on
$\mathcal K_1$,
\begin{equation}
 |T_{nk}(r,\theta)|
 +|\partial_rT_{nk}(r,\theta)|
 =O(n^{-3/2}\log n).
 \label{eq:C1-crossing-rate}
\end{equation}
Moreover, for every fixed finite $q$,
\begin{equation}
 \|c_n-c^*\|_{L^q}
 +\|\dot c_n-\dot c^*\|_{L^q}=O(n^{-1/2}).
 \label{eq:C1-localization}
\end{equation}
\end{lemma}

\begin{proof}
All constants below are uniform on $\mathcal K_1$. We suppress
$(r,\theta)$ when no confusion results. Set $a=c^*(r,\theta)$ and write
\begin{align*}
 A_n(t)&=P_{n1}(X_1-t)^+ +P_{n2}(t-X_2)^+,\\
 \widehat D_n(t;r)
 &=\eta_n+rP_{n1}\ind{X_1>t}+(1-r)P_{n2}\ind{X_2<t}.
\end{align*}
Off the sample switching points, implicit differentiation gives
\begin{equation}
 \dot c_n(r)=\frac{A_n(c_n(r))}{\widehat D_n(c_n(r);r)}.
 \label{eq:sample-cutoff-derivative}
\end{equation}
The same formula holds after deleting one observation, with the original
camp normalizations retained. For fixed $n$, all these roots are
continuous and piecewise continuously differentiable in $r$, with
derivatives bounded by $K\eta_n^{-1}$. A tie has probability zero at each
fixed $r$: delete the tied observation and condition on all the others.

Choose a fixed $\varepsilon_1>0$ small enough that
$J=[a-\varepsilon_1,a+\varepsilon_1]$ lies in the smoothness neighborhood
and the population slope magnitude is at least $3d_0/4$ throughout $J$.
Hoeffding bounds for the scores at $a\pm\varepsilon_1/2$, and the uniform
empirical tail bound on $J$, give an event $\mathcal E_n$ on which every
full and deleted root lies in $J$ and every corresponding denominator on
$J$ is at least $d_0/2$, with
\begin{equation}
 \sup_{\mathcal K_1}\Pr_\theta(\mathcal E_n^c)
 \leq K e^{-kn}.
 \label{eq:C1-uniform-good-event}
\end{equation}
The same construction without observation $i$ gives an
$\mathcal F_{-i}$-measurable event $\mathcal E_{n,-i}$ with this exponential
probability bound, uniformly in $i$, on which these conclusions hold
\emph{for every replacement value} $X_i=x\in\mathcal X$. Indeed the
omitted score and tail mass are uniformly $O(n^{-1})$; by choosing the
score and slope tolerances with fixed slack, adding any one observation
preserves the conclusions. On this event all full, deleted, and
replacement-root derivatives are bounded by a uniform constant. On its
complement their polynomial bound $K\eta_n^{-1}$ makes every fixed moment
negligible. No indicator of a good event is differentiated below; the
events are used only to bound already differentiated integrands at a
fixed parameter value.

Uniform empirical half-line moment bounds, the lower bound on the
denominator, and the local Lipschitz properties of the population
numerator and denominator yield
\[
 \|c_n-a\|_{L^q}+\|\dot c_n-\dot a\|_{L^q}=O(n^{-1/2}),
\]
for every fixed finite $q$. The identical bounds hold for every deleted
root. These prove \eqref{eq:C1-localization}.

We record the local interval estimate used below and in the following
lemma. For deterministic $h$ in the cutoff neighborhood,
\begin{equation}
 \max_k\left\|
  \sup_{\substack{I\subset J\\ |I|\leq h}}
  |(P_{nk}-P_k)\ind{X_k\in I}|
 \right\|_{L^q}
 \leq K\left(\sqrt{\frac{h\log n}{n}}
                  +\frac{\log n}{n}\right).
 \label{eq:C1-centered-interval-bound}
\end{equation}
For completeness, partition $J$ into cells of length at most $n^{-2}$.
Apply Bernstein's inequality to the $O(n^4)$ intervals with endpoints on
this grid; each interval of length at most $h+2n^{-2}$ has population
mass at most $M(h+2n^{-2})$. An arbitrary interval differs from a grid
interval by at most two cells. The same bound for cell counts controls
these endpoint errors. Integration of the resulting tail bound gives
the displayed $L^q$ estimate. Deleting one observation changes the
bound only by $O(n^{-1})$.

Put $z_i=c_n^{(i)}$, $\Delta_i=c_n-z_i$, and $d_i=z_i-a$. Monotonicity of
the sample score and its slope lower bound give, on
$\mathcal E_{n,-i}$, uniformly in the replacement value $x$,
\begin{align}
 0\leq\Delta_i(x,r)&\leq \frac K n(x-z_i)^+,
 && i\in I_1,\label{eq:plus-self-influence}\\
 0\leq-\Delta_i(x,r)&\leq \frac K n(z_i-x)^+,
 && i\in I_2.\label{eq:minus-self-influence}
\end{align}
In particular, $|\Delta_i(x,r)|\leq K_0/n$ for a fixed $K_0$.
Let $P_{nk}^{(i)}$ denote the empirical measure with observation $i$
omitted if it belongs to camp $k$, still divided by the original $n_k$;
the other camp's measure is unchanged. Set
\[
 Z_{ni}=\sum_{k=1}^2
  \sup_{\substack{I\subset J\\ |I|\leq K_0/n}}
    P_{nk}^{(i)}\ind{X_k\in I}.
\]
The density bound and \eqref{eq:C1-centered-interval-bound} imply
$\sup_i\|Z_{ni}\|_{L^q}\leq K\log n/n$. Subtracting the full and deleted
ratios in \eqref{eq:sample-cutoff-derivative} gives the deterministic
envelope
\begin{equation}
 \sup_{x\in\mathcal X}|\dot\Delta_i(x,r)|
 \leq K\{n^{-1}+Z_{ni}\}
 \quad\text{on }\mathcal E_{n,-i},
 \label{eq:C1-replacement-derivative-envelope}
\end{equation}
where derivatives are evaluated off switching points. To see this, the
numerator changes by at most $K(n^{-1}+|\Delta_i|)$, and the denominator
changes by at most $Kn^{-1}$ plus the deleted empirical mass between the
two roots. Their distance is at most $K_0/n$, so this last mass is bounded
by $Z_{ni}$. Both denominators are bounded below. Consequently
\begin{equation}
 \sup_i\|\Delta_i\|_{L^q}=O(n^{-1}),\qquad
 \sup_i\|\dot\Delta_i\|_{L^q}=O(n^{-1}\log n).
 \label{eq:C1-loo}
\end{equation}
Here and in subsequent estimates, the exponentially small complement is
absorbed using the polynomial bound on the derivatives.

Write $\bar\psi_k(z,a)=P_k\psi(\cdot;z,a)$. Direct integration gives
\begin{equation}
 \partial_z\bar\psi_k(z,a)=F_{k,\theta}(z)-F_{k,\theta}(a),\qquad
 \partial_a\bar\psi_k(z,a)=-(z-a)f_{k,\theta}(a).
 \label{eq:triangle-derivatives}
\end{equation}
Both derivatives vanish at $z=a$. Conditional on $\mathcal F_{-i}$,
$z_i$ is independent of $X_i$, so
\begin{equation}
 \E[\psi(X_i;z_i,a)-\bar\psi_k(z_i,a)\mid\mathcal F_{-i}]=0.
 \label{eq:loo-centering}
\end{equation}
We subtract this identity from a summand of $T_{nk}$ and treat its
population and observation parts separately.

For the population part, the density and density-derivative bounds and
\eqref{eq:triangle-derivatives} give on the good event
\begin{align}
 |\bar\psi_k(z_i+\Delta_i,a)-\bar\psi_k(z_i,a)|
 &\leq K(|d_i||\Delta_i|+\Delta_i^2),
 \label{eq:population-crossing-envelope}\\
 \left|\partial_r\{
 \bar\psi_k(z_i+\Delta_i,a)-\bar\psi_k(z_i,a)\}\right|
 &\leq K\left[
 (|d_i|+|\dot z_i-\dot a|)|\Delta_i|
 +( |d_i|+|\Delta_i|)|\dot\Delta_i|+\Delta_i^2\right].
 \label{eq:population-crossing-derivative-envelope}
\end{align}
For example, cancellation of the first-order density terms leaves
$f_{k,\theta}(a)\Delta_i(\dot z_i-\dot a)$, which explains why the first
derivative is small. H\"older's inequality, localization, and
\eqref{eq:C1-loo} bound the expected value by $O(n^{-3/2})$ and its
expected absolute derivative by $O(n^{-3/2}\log n)$.

The observation part has a particularly useful exact form. If $i\in I_1$
and $x>z_i$, the full root lies strictly between $z_i$ and $x$; otherwise
it equals $z_i$. The downward case reverses the inequalities. Thus, up to
density-zero endpoints,
\begin{align}
 \psi(x;z_i+\Delta_i,a)-\psi(x;z_i,a)
 &=-\Delta_i(x,r)\ind{z_i<a,\ z_i<x<a}, &&i\in I_1,
 \label{eq:plus-exact-observation-crossing}\\
 \psi(x;z_i+\Delta_i,a)-\psi(x;z_i,a)
 &=\Delta_i(x,r)\ind{z_i>a,\ a<x<z_i}, &&i\in I_2.
 \label{eq:minus-exact-observation-crossing}
\end{align}
Their conditional expectations are respectively
\begin{equation}
 -\ind{z_i<a}\int_{z_i}^a\Delta_i(x,r)f_{1,\theta}(x)\,dx,
 \qquad
 \ind{z_i>a}\int_a^{z_i}\Delta_i(x,r)f_{2,\theta}(x)\,dx.
 \label{eq:exact-observation-integrals}
\end{equation}
The expectation of the absolute value is $O(n^{-1}\E d_i^2)=O(n^{-2})$
on the good event, by the sharper self-influence inequalities. For
example, on $z_i<a$ the derivative of the first integral is
\begin{equation}
 -\int_{z_i}^a\dot\Delta_i(x,r)f_{1,\theta}(x)\,dx
 -\Delta_i(a,r)f_{1,\theta}(a)\dot a.
 \label{eq:plus-observation-integral-derivative}
\end{equation}
The endpoint at $z_i$ contributes zero because $\Delta_i(z_i,r)=0$.
The integrals and their derivatives join continuously when $z_i=a$.
By \eqref{eq:C1-replacement-derivative-envelope}, their expected absolute
derivative is bounded by
\begin{align*}
 K\E\bigl[|d_i|\{n^{-1}+Z_{ni}\}\bigr]
   +K n^{-1}\E|d_i|
 &\leq K\|d_i\|_{L^2}
        \{n^{-1}+\|Z_{ni}\|_{L^2}\}\\
 &=O(n^{-3/2}\log n).
\end{align*}
These are expectation bounds obtained from a random envelope; no
deterministic bound on a random interval count is asserted.

For each fixed $n$, the polynomial derivative bounds justify conditional
integration and differentiation by domination. They also control the
complement of the good event. At a fixed $r$, ties among deleted roots
and remaining observations have probability zero. Dominated convergence
then makes the expected derivatives continuous; the integral formulas
above also account explicitly for the moving endpoint $a$. Combining
the population and observation bounds, and averaging over the
observations in camp $k$, gives
$|T_{nk}|+|\partial_rT_{nk}|=O(n^{-3/2}\log n)$, uniformly on
$\mathcal K_1$, as claimed.
\end{proof}

\begin{lemma}[$C^1$ local-root moments]
\label{lem:C1localroot}
Let $\mathcal K_1$ be as in Lemma~\ref{lem:C1crossing}.  Define
$\delta_n=c_n-c^*$, let $e_n,u_n$ be the centered score and slope in
Lemma~\ref{lem:localroot}, and put
\begin{equation}
 C_n(r,\theta)=n\E_\theta[u_ne_n],
 \qquad
 \Omega_n(r,\theta)=n\E_\theta[e_n^2].
 \label{eq:C1-finite-leading-moments}
\end{equation}
Then, uniformly on $\mathcal K_1$,
\begin{align}
 \partial_r\E_\theta\delta_n
 &=\frac1n\partial_r\left\{
    \frac{C_n}{(D^*)^2}
   +\frac{H_{cc}(c^*;r,\theta)\Omega_n}{2(D^*)^3}
   +\frac{n\eta_n(a_0-c^*)}{D^*}\right\}+o(n^{-1}),
 \label{eq:C1-root-mean}\\
 \partial_r\E_\theta\delta_n^2
 &=\frac1n\partial_r\left\{
   \frac{\Omega_n}{(D^*)^2}\right\}+o(n^{-1}).
 \label{eq:C1-root-second-moment}
\end{align}
\end{lemma}

\begin{proof}
We give the remainder calculation because an $L^1$ expansion alone cannot
be differentiated.  Suppress $(r,\theta)$ and write
\[
 D=D^*,\qquad h=H_{cc}(c^*;r,\theta),\qquad
 \rho_n=\delta_n-\frac{e_n}{D},\qquad
 \kappa=\dot c^*.
\]
We first record a value bound using only empirical half-line moments.
Let $S_n$ be the sum of the two empirical distribution functions' uniform
absolute deviations from their population counterparts. Then
$\|S_n\|_{L^q}=O(n^{-1/2})$ for every fixed finite $q$. On the localized
root event, each empirical crossing loss is bounded by
$K|\delta_n|(|\delta_n|+S_n)$, since its support is an interval of length
$|\delta_n|$. The exact root equation therefore gives
\[
 |\rho_n|\leq K\{ |u_n\delta_n|+\delta_n^2
                  +|\delta_n|S_n+\eta_n\}.
\]
H\"older's inequality and the fixed-order root and half-line moment
bounds, with the exponentially small complement absorbed by boundedness,
yield for every fixed finite $q$
\begin{equation}
 \|\rho_n\|_{L^q}=O(n^{-1}).
 \label{eq:rho-value-rate}
\end{equation}
We first record its derivative counterpart.  Let
\begin{align*}
 \gamma_n
 &=(P_{n1}-P_1)g_1 +(P_{n2}-P_2)g_2,\\
 k&=p_2-p_1,
 &j&=-r f_{1,\theta}(c^*)+(1-r)f_{2,\theta}(c^*).
\end{align*}
Direct differentiation of the centered score, including the motion of
$c^*$, gives the exact identity
\begin{equation}
 \dot e_n=\gamma_n+\kappa u_n.
 \label{eq:score-derivative-identity}
\end{equation}
Expand the numerator and denominator of
\eqref{eq:sample-cutoff-derivative} at $c^*$.  The population linear terms
are $k\delta_n$ and $j\delta_n$, respectively.  More explicitly, with
$A=A(c^*)=\kappa D$,
\begin{align}
 A_n(c_n)
 &=A+\gamma_n+k\delta_n+\alpha_n,
 \label{eq:C1-numerator-expansion}\\
 \widehat D_n(c_n;r)
 &=D-u_n+j\delta_n+\eta_n+\beta_n.
 \label{eq:C1-denominator-expansion}
\end{align}
To bound the centered changes over the random interval between $c^*$ and
$c_n$, put $h_n=n^{-1/2}\log n$. The root tail bound makes
$\Pr(|\delta_n|>h_n)$ smaller than every fixed inverse power of $n$.
On $|\delta_n|\leq h_n$, take the supremum in
\eqref{eq:C1-centered-interval-bound} before substituting the random
interval. Its $L^q$ bound is $O(n^{-3/4}\log n)$; boundedness controls
the complement. The numerator's centered hinge increment is bounded in
the same way, or directly by its Lipschitz factor times the uniform
half-line fluctuation. The population Taylor errors are
$O_{L^q}(n^{-1})$.
Thus
\begin{equation}
 \|\alpha_n\|_{L^q}+\|\beta_n\|_{L^q}
 \leq K(n^{-3/4}\log n+n^{-1}).
 \label{eq:C1-random-interval-errors}
\end{equation}
Substitute \eqref{eq:C1-numerator-expansion}--
\eqref{eq:C1-denominator-expansion}, use
$\rho_n=O_{L^q}(n^{-1})$, and expand the ratio on the event where its
denominator is at least $d_0/2$.  Products of two global empirical
fluctuations and the anchor are $O_{L^q}(n^{-1})$; the exponentially small
complement is negligible.  This yields
\begin{equation}
 \dot\delta_n
 =\frac{\gamma_n+\kappa u_n}{D}
  +\frac{(k-\kappa j)e_n}{D^2}+q_n,
 \qquad
 \|q_n\|_{L^q}=O(n^{-3/4}\log n).
 \label{eq:C1-Bahadur-ratio}
\end{equation}
Here and below derivatives are taken off the finitely many sample switching
points; the integral formulas in Lemma~\ref{lem:C1crossing} supply the
continuous expected extension.  Since
\begin{equation}
 \dot D=p_1-p_2+\kappa j=-k+\kappa j,
 \label{eq:D-derivative-identity}
\end{equation}
the first two terms in \eqref{eq:C1-Bahadur-ratio} equal
$\partial_r(e_n/D)$.  Consequently,
\begin{equation}
 \|\dot\rho_n\|_{L^q}
 =O(n^{-3/4}\log n)=o(n^{-1/2}).
 \label{eq:rho-derivative-rate}
\end{equation}
This is the $C^1$ Bahadur bound used below.

We next name every term in the second-order remainder.  With the crossing
losses of Lemma~\ref{lem:localroot}, define
\begin{align*}
 K(t)
 &=rP_1\ell_1(\cdot,t)-(1-r)P_2\ell_2(\cdot,t),\\
 G_n(t)
 &=r(P_{n1}-P_1)\ell_1(\cdot,t)
   -(1-r)(P_{n2}-P_2)\ell_2(\cdot,t).
\end{align*}
The sample root equation is exactly
\begin{equation}
 D\delta_n=e_n+u_n\delta_n+K(\delta_n)+G_n(\delta_n)
                    +\eta_n(a_0-c^*-\delta_n).
 \label{eq:C1-exact-root-equation}
\end{equation}
It follows algebraically that
\begin{equation}
 \delta_n=\frac{e_n}{D}+\frac{u_ne_n}{D^2}
     +\frac{he_n^2}{2D^3}+\frac{\eta_n(a_0-c^*)}{D}
     +\mathfrak r_n,
 \label{eq:C1-named-root-expansion}
\end{equation}
where
\begin{equation}
 \mathfrak r_n
 =\frac{u_n\rho_n}{D}
  +\frac{K(\delta_n)-\tfrac12h\delta_n^2}{D}
  +\frac{h}{2D}\left(\delta_n^2-\frac{e_n^2}{D^2}\right)
  +\frac{G_n(\delta_n)}{D}
  -\frac{\eta_n\delta_n}{D}.
 \label{eq:C1-remainder-decomposition}
\end{equation}

Four of the five terms in \eqref{eq:C1-remainder-decomposition} are now
immediate.  Lemma~\ref{lem:C1crossing} gives
\begin{equation}
 \left|\partial_r\E G_n(\delta_n)\right|
 =O(n^{-3/2}\log n).
 \label{eq:C1-G-bound}
\end{equation}
Taylor's theorem and Assumption~\ref{ass:finitepool}(ii) give
\begin{equation}
 \left|\partial_r\E
 [K(\delta_n)-\tfrac12h\delta_n^2]\right|
 \leq K\E[|\delta_n|^3+\delta_n^2|\dot\delta_n|]
 =O(n^{-3/2}).
 \label{eq:C1-cubic-bound}
\end{equation}
Equations \eqref{eq:rho-value-rate}--\eqref{eq:rho-derivative-rate} and
$\delta_n^2-e_n^2/D^2=\rho_n(\delta_n+e_n/D)$ imply
\begin{equation}
 \left|\partial_r\E
 \left[\delta_n^2-\frac{e_n^2}{D^2}\right]\right|
 =O(n^{-5/4}\log n)=o(n^{-1}),
 \label{eq:C1-quadratic-replacement}
\end{equation}
while $\eta_n=O(n^{-1})$ and
\eqref{eq:C1-localization} give
\begin{equation}
 \left|\partial_r\E[\eta_n\delta_n]\right|=O(n^{-3/2}).
 \label{eq:C1-anchor-remainder}
\end{equation}

It remains to control the derivative of $\E[u_n\rho_n]$.  We differentiate
the fixed-population-cutoff indicators under their density integrals;
only a level, rather than a derivative, leave-one-out bound is needed.
Put $a=c^*$ and $\kappa=\dot a$.  For $i\in I_k$, define
\[
 I_i(r)=
 \begin{cases}\ind{X_i>a},&k=1,\\ \ind{X_i<a},&k=2,\end{cases}
 \qquad p_i(r)=\E I_i(r),\qquad \xi_i=I_i-p_i,
\]
and let $\beta_i(r)=-r/n_1$ in camp 1 and
$\beta_i(r)=-(1-r)/n_2$ in camp 2.  Then
$u_n=\sum_i\beta_i\xi_i$.  Write
$\rho_{n,i}^{a}(r)$ for $\rho_n(r)$ evaluated after replacing observation
$i$ by the deterministic value $a(r)$; expectation of this object is over
the other observations.  Conditional integration at the moving boundary
$X_i=a(r)$ gives the exact identity
\begin{align}
 \partial_r\E[u_n\rho_n]
 &=\E[u_n\dot\rho_n]
   +\E\!\left[\left(\sum_i\dot\beta_i\xi_i\right)\rho_n\right]
   +\sum_i\beta_i\dot p_i
      \left\{\E_{-i}\rho_{n,i}^{a}-\E\rho_n\right\},
 \label{eq:C1-slope-boundary-identity}
\end{align}
where $\dot p_i=-\kappa f_{1,\theta}(a)$ in camp 1 and
$\dot p_i=\kappa f_{2,\theta}(a)$ in camp 2.  In fact, the boundary
contribution to $\partial_r\E[I_i\rho_n]$ is
$\dot p_i\E_{-i}\rho_{n,i}^{a}$, whereas differentiating
$p_i\E\rho_n$ contributes $\dot p_i\E\rho_n$.
The sample root and $\rho_n$ are continuous in the observation values
and absolutely continuous in $r$.  For fixed $n$, the positive anchor
bounds their $r$-derivatives by a finite deterministic constant.  Thus
conditional integration and dominated convergence justify this identity;
ties at a fixed $r$ have probability zero in its ordinary expectation
terms.  The right-hand side is continuous in $r$: in the ordinary expectations,
fixed-$r$ ties have probability zero, while the boundary evaluation is
continuous in $r$ because the root is continuous in both $r$ and the
inserted value $a(r)$. The same fixed-$n$ domination applies.
No sample indicator is differentiated pointwise at its boundary.

The first two terms in \eqref{eq:C1-slope-boundary-identity} satisfy
\begin{align}
 |\E[u_n\dot\rho_n]|
 &\leq \|u_n\|_{L^2}\|\dot\rho_n\|_{L^2}
   =O(n^{-5/4}\log n),\notag\\
 \left|\E\!\left[\left(\sum_i\dot\beta_i\xi_i\right)\rho_n\right]\right|
 &\leq\left\|\sum_i\dot\beta_i\xi_i\right\|_{L^2}
        \|\rho_n\|_{L^2}=O(n^{-3/2}),
 \label{eq:C1-slope-bulk-bounds}
\end{align}
by \eqref{eq:rho-value-rate} and \eqref{eq:rho-derivative-rate} and the
centered empirical moment bounds.  The cancellation inside the braces
in the final term of \eqref{eq:C1-slope-boundary-identity} is essential.
We next establish its rate directly.

Let $a_{ni}(c,r)$ be observation $i$'s signed normalized contribution to
the unanchored score: $r(X_i-c)^+/n_1$ in camp 1 and
$-(1-r)(c-X_i)^+/n_2$ in camp 2.  Delete this contribution without changing
the camp normalization and put
\[
 z_i=c_n^{(i)},\qquad
 e_n^{(i)}=e_n-a_{ni}(a,r),\qquad
 \rho_n^{(i)}=z_i-a-\frac{e_n^{(i)}}D.
\]
These three objects are independent of $X_i$.  Uniformly in $i$,
$\|z_i-a\|_{L^q}=O(n^{-1/2})$ for every fixed finite $q$: the deletion
changes the score by $O(n^{-1})$, so the root localization argument
applies with the same constants.

For the following bounds we use a good event measurable with respect to
the deleted sample.  Choose a fixed neighborhood
$J=[a-\varepsilon_0,a+\varepsilon_0]$ on which the population tail
denominator is uniformly bounded below.  Require the deleted score to
have its population signs at the endpoints with a fixed positive margin,
and the deleted empirical tail averages to be uniformly sufficiently close
to their population counterparts on $J$.  Hoeffding's inequality and the
half-line maximal bound make this event's complement at most $Ke^{-kn}$,
uniformly in the parameters and $i$.  Because inserting any
$x\in\mathcal X$ changes the score and tail averages by only $K/n$,
on this event, for all sufficiently large $n$, both the deleted root and
the root after any such insertion lie in $J$, with all denominators
along their connecting interval at least $d_0/2$ (after reducing the
fixed neighborhood and its concentration tolerances if necessary).
In particular, the same event remains valid when $x=a$.

For the observed value $X_i$, the integral mean-value formula gives
\[
 c_n-z_i=\frac{a_{ni}(z_i,r)}{\overline D_{ni}},\qquad
 \overline D_{ni}
 =-\int_0^1\partial_cH_n\bigl(z_i+t(c_n-z_i);r\bigr)\,dt.
\]
On the good event, $|c_n-z_i|\leq K/n$ and
\[
 \E\!\left[\ind{\text{good}}\,
         |\overline D_{ni}-D|\right]=O(n^{-1/2}).
\]
To verify the latter bound, separate the deleted empirical tails from
their population tails uniformly on $J$, whose maximal $L^1$ deviation is
$O(n^{-1/2})$.  The inserted observation and the anchor contribute
$O(n^{-1})$, while Lipschitz continuity of the population tails bounds
their displacement along the integration interval by
$K(|z_i-a|+|c_n-z_i|)$.  Thus this estimate uses only levels of the
averaged denominator.
The exact identity
\[
 \rho_n-\rho_n^{(i)}
 =\frac{a_{ni}(z_i,r)-a_{ni}(a,r)}D
  +a_{ni}(z_i,r)
       \left(\frac1{\overline D_{ni}}-\frac1D\right)
\]
and the $K/n$ Lipschitz constant and uniform bound of $a_{ni}$ now yield
\begin{equation}
 \sup_i\E|\rho_n-\rho_n^{(i)}|=O(n^{-3/2}).
 \label{eq:C1-slope-level-influence}
\end{equation}
On the exponentially small complement, the roots, scores, and $D^{-1}$
are uniformly bounded, so its contribution to this level estimate is
negligible.

If instead $X_i=a$, then $a_{ni}(a,r)=0$ exactly and
$|a_{ni}(z_i,r)|\leq K|z_i-a|/n$.  The same deleted-sample good event
and the mean-value formula therefore give
\begin{equation}
 \sup_i\E_{-i}|\rho_{n,i}^{a}-\rho_n^{(i)}|
 =O(n^{-3/2}).
 \label{eq:C1-slope-boundary-influence}
\end{equation}
Subtracting the common expectation of $\rho_n^{(i)}$ in
\eqref{eq:C1-slope-level-influence} and
\eqref{eq:C1-slope-boundary-influence} proves
\[
 \sup_i\left|\E_{-i}\rho_{n,i}^{a}-\E\rho_n\right|
 =O(n^{-3/2}).
\]
Finally, $\sum_i|\beta_i|=1$ and $\sup_i|\dot p_i|\leq K$.
Combining this boundary estimate with
\eqref{eq:C1-slope-boundary-identity}--\eqref{eq:C1-slope-bulk-bounds}
yields
\begin{equation}
 \left|\partial_r\E[u_n\rho_n]\right|
 =O(n^{-5/4}\log n)=o(n^{-1}),
 \label{eq:C1-slope-root-bound}
\end{equation}
uniformly on $\mathcal K_1$.  In particular, the argument does not
differentiate $\overline D_{ni}$ or require a derivative moment for it.

Together, \eqref{eq:C1-G-bound}--\eqref{eq:C1-slope-root-bound} prove
\begin{equation}
 \partial_r\E\mathfrak r_n=o(n^{-1})
 \label{eq:C1-remainder-derivative}
\end{equation}
uniformly on $\mathcal K_1$.

Finally, $\E e_n=0$, $n\E[u_ne_n]=C_n$, and
$n\E e_n^2=\Omega_n$ exactly.  Taking expectations and differentiating
\eqref{eq:C1-named-root-expansion}, then using
\eqref{eq:C1-remainder-derivative}, proves
\eqref{eq:C1-root-mean}.  Equation
\eqref{eq:C1-quadratic-replacement} gives
\[
 \partial_r\E\delta_n^2
 =\partial_r\E(e_n^2/D^2)+o(n^{-1}),
\]
which proves \eqref{eq:C1-root-second-moment}.
\end{proof}

\begin{lemma}[Uniform derivative expansion]
\label{lem:riskderivative}
Let $\mathcal K_1$ be a compact subset of the interior of the design domain
on which the interiority and smoothness conditions in
Assumption~\ref{ass:finitepool} hold. Then
$R_n(\cdot,\theta)$ is continuously
differentiable in $r$ and, uniformly on $\mathcal K_1$,
\begin{equation}
 n\left|
 \partial_r R_n(r,\theta)
 -\partial_r\left\{
 [c^*(r,\theta)-\theta]^2+\frac{Q_{\nu}(r,\theta)}{n}
 \right\}
 \right|\longrightarrow0.
 \label{eq:conditional-risk-derivative}
\end{equation}
\end{lemma}

\begin{proof}
Suppress $\theta$ where no confusion results. Equation
\eqref{eq:sample-cutoff-derivative} shows that, for every realization,
$r\mapsto c_n(r)$ is continuous, absolutely continuous, and piecewise
continuously differentiable. Its derivative is bounded for fixed $n$ by a
finite constant proportional to $\eta_n^{-1}$. Hence the sample squared
loss is absolutely continuous as well. At any fixed $r$, a cutoff tie has
probability zero: deleting the tied observation leaves a root independent
of that observation, and the evidence law is atomless. Difference
quotients, Fubini's theorem, and dominated convergence therefore give
\begin{equation}
 \partial_r R_n(r,\theta)
 =2\E_\theta[(c_n(r)-\theta)\dot c_n(r)].
 \label{eq:exact-risk-score}
\end{equation}
The density smoothness in Assumption~\ref{ass:finitepool}(ii) and the same
domination make this derivative continuous. The population counterpart of
\eqref{eq:sample-cutoff-derivative} is
\begin{equation}
 \dot c^*(r,\theta)
 =\frac{A(c^*)}{D^*},
 \qquad
 A(c)=\E_\theta(X_1-c)^++\E_\theta(c-X_2)^+,
 \label{eq:population-cutoff-derivative}
\end{equation}
which is \eqref{eq:cs}.

Put $\delta_n=c_n-c^*$ and
\begin{equation}
 m_n(r,\theta)=\E_\theta\delta_n,
 \qquad
 s_n(r,\theta)=\E_\theta\delta_n^2.
 \label{eq:root-moments}
\end{equation}
Uniformly on $\mathcal K_1$,
\begin{align}
 \partial_r m_n(r,\theta)
 &=\frac{\partial_r B^{C,\nu}(r,\theta)}{n}+o(n^{-1}),
 \label{eq:C1-bias}\\
 \partial_r s_n(r,\theta)
 &=\frac{\partial_r V^C(r,\theta)}{n}+o(n^{-1}).
 \label{eq:C1-variance}
\end{align}
Here and below the $o(n^{-1})$ terms are uniform. We establish these two
statements from Lemma~\ref{lem:C1localroot}, rather than infer them by
differentiating the value expansion in Proposition~\ref{thm:finitewelfare}.
Indeed, the explicit formulas in the proof of that theorem show that
$C_n,\Omega_n$ and their first $r$-derivatives converge uniformly to
$C,\Omega$ and their derivatives.  This follows from
$\alpha_{k,n}\to\alpha_k$ for $k\in\{1,2\}$, the $C^1$ tail-moment formulas, compact
support, and the uniform density bounds.  Also $n\eta_n\to \nu$.  Substitution
in \eqref{eq:C1-root-mean}--\eqref{eq:C1-root-second-moment} therefore gives
\eqref{eq:C1-bias}--\eqref{eq:C1-variance}.

Finally, let $d(r,\theta)=c^*(r,\theta)-\theta$ and
$\kappa(r,\theta)=\partial_r c^*(r,\theta)$. The exact identity
\[
 R_n(r,\theta)=d(r,\theta)^2+2d(r,\theta)m_n(r,\theta)+s_n(r,\theta)
\]
and Proposition~\ref{thm:finitewelfare} give
\begin{align*}
 \partial_rR_n(r,\theta)
 &=2d\kappa
   +\frac{2\kappa B^{C,\nu}+2d\,\partial_rB^{C,\nu}
           +\partial_rV^C}{n}
   +o(n^{-1})\\
 &=\partial_r\left\{d(r,\theta)^2+\frac{Q_{\nu}(r,\theta)}{n}\right\}
   +o(n^{-1}),
\end{align*}
uniformly, since $Q_{\nu}=V^C+2dB^{C,\nu}$. This is
\eqref{eq:conditional-risk-derivative}.
\end{proof}

\begin{proof}[Proof of Theorem~\ref{thm:exactpolicy}]
Integrating \eqref{eq:conditional-risk-derivative} with respect to $\mu$
is justified by the uniform bounds on
$\mathcal N_\mu\times\operatorname{supp}\mu$ and gives
\eqref{eq:riskderivativeexpansion}. Proposition~\ref{prop:finitepolicy}(i)
already shows that every exact minimizer $\widehat r_{n,\mu}$ converges to
the interior point $r_\mu$. It is therefore interior for all sufficiently
large $n$ and satisfies
\(
 \mathcal W_{n,\mu}'(\widehat r_{n,\mu})=0.
\)

Let $h_n=\widehat r_{n,\mu}-r_\mu$. Positive curvature and Taylor's theorem
give
\(
 \mathcal L_\mu'(r_\mu+h_n)
 =\mathcal L_\mu''(r_\mu)h_n+o(|h_n|).
\)
Equation~\eqref{eq:riskderivativeexpansion}, boundedness of
$\mathcal Q_{\mu,\nu}'$ near $r_\mu$, and
$\mathcal L_\mu''(r_\mu)>0$ first imply $h_n=O(n^{-1})$. Multiplying the
exact first-order condition by $n$ then gives
\[
 0=\mathcal L_\mu''(r_\mu)\,n h_n
   +\mathcal Q_{\mu,\nu}'(r_\mu)+o(1),
\]
which proves the first limit in \eqref{eq:exactpolicyshift}. The second
follows by comparison with \eqref{eq:secondorderpolicy}.

For a point mass at $\theta$,
$\mathcal L_\mu''(r^*)=2\kappa(\theta)^2$ and
\[
 \partial_rQ_{\nu}(r^*,\theta)
 =\partial_rV^C(r^*,\theta)
  +2\kappa(\theta)B^{C,\nu}(r^*,\theta).
\]
Substitution proves \eqref{eq:exactbiasvarianceshift}.
\end{proof}

\subsection{Leave-one-out stability and the proof of
Proposition~\ref{prop:restoration}}
\label{app:restoration}

Throughout, $d=\overline x-\underline x$. For each item $i$, let
$H_{n}^{(i)}$ be $H_{n}$ with item $i$'s term deleted, and let
$c_{n}^{(i)}$ be its unique zero (the endpoint sign conditions and strict
monotonicity of Theorem~\ref{thm:finite} hold verbatim). Let
$Z_{n}(c)=\sum_{i\in I_1}v_{ni}\ind{X_{i}>c}
+\sum_{i\in I_2}v_{ni}\ind{X_{i}<c}$ denote visible attention mass, and
$Z_{n}^{(i)}$ its leave-one-out version, so that
$-H_{n}'(c)=\eta_{n}+Z_{n}(c)$ almost everywhere and
$Z_{n}\geq Z_{n}^{(i)}$ pointwise.

\begin{lemma}[Leave-one-out stability]
\label{lem:loo}
Under the conditions of Proposition~\ref{thm:largepool}, fix
$\varepsilon_{0}>0$ with
$J_{0}=[c^{*}-\varepsilon_{0},c^{*}+\varepsilon_{0}]\subset
\mathrm{int}\,\mathcal X$ and let
$D_{*}=\tfrac12\min_{c\in J_{0}}
[r\,\Prob_{\theta}(X_1>c)+(1-r)\,\Prob_{\theta}(X_2<c)]>0$. There are
events $E_{n}^{(i)}$, each independent of $X_{i}$ and with
$\min_{i}\Prob(E_{n}^{(i)})\to1$, such that on $E_{n}^{(i)}$:
$c_{n}\in J_{0}$, $c_{n}^{(i)}\in J_{0}$, and
\[
  |c_{n}-c_{n}^{(i)}|
  \;\leq\;
  \delta_{n}
  :=\frac{d}{D_{*}}\max_{j}v_{nj}
  \;\longrightarrow\;0 .
\]
\end{lemma}

\begin{proof}
The argument of Lemma~\ref{lem:uniform} applies verbatim to
$H_{n}^{(i)}$ and to $Z_{n}^{(i)}$ (the summands of $Z_{n}^{(i)}$ are
monotone in $c$ and uniformly bounded, and its mean converges uniformly to
the continuous function
$D(c)=r\Prob_{\theta}(X_1>c)+(1-r)\Prob_{\theta}(X_2<c)$), with all
bounds uniform in $i$: deleting one term changes $H_n$ by at most
$d\max_j v_{nj}$ and $Z_n$ by at most $\max_j v_{nj}$ in sup norm,
and both bounds converge to zero. Let
$\gamma=\min\{H(c^{*}-\varepsilon_{0}),-H(c^{*}+\varepsilon_{0})\}>0$ and
define
\[
  E_{n}^{(i)}
  =\Big\{
  H_{n}^{(i)}(c^{*}-\varepsilon_{0})>\tfrac{\gamma}{2},\;
  H_{n}^{(i)}(c^{*}+\varepsilon_{0})<-\tfrac{\gamma}{2},\;
  \inf_{c\in J_{0}}Z_{n}^{(i)}(c)\geq D_{*}
  \Big\},
\]
which is measurable with respect to $\{X_{j}\}_{j\neq i}$, hence
independent of $X_{i}$, and satisfies
$\min_{i}\Prob(E_{n}^{(i)})\to1$ by the uniform convergence just noted.
On $E_{n}^{(i)}$: the sign conditions force $c_{n}^{(i)}\in J_{0}$; since
$\sup_{c}|H_{n}-H_{n}^{(i)}|\leq v_{ni}d\leq\max_{j}v_{nj}\,d$, for $n$
large ($\max_{j}v_{nj}d<\gamma/2$) also $c_{n}\in J_{0}$. Finally, on
$J_{0}$ the map $H_{n}$ decreases at rate at least
$\eta_{n}+Z_{n}\geq Z_{n}^{(i)}\geq D_{*}$, while
$|H_{n}(c_{n}^{(i)})|=|H_{n}(c_{n}^{(i)})-H_{n}^{(i)}(c_{n}^{(i)})|
\leq v_{ni}d$; hence $|c_{n}-c_{n}^{(i)}|\leq v_{ni}d/D_{*}\leq\delta_{n}$.
\end{proof}

\begin{proof}[Proof of Proposition~\ref{prop:restoration}]
Here $F_{1,\theta}=F_{2,\theta}=F_{\theta}$, $r=1/2$, and
$c^{*}=\theta$ (Corollary~\ref{thm:countervailing}). Fix $x\neq\theta$ and
$i\in I_1$. Since $c_{n}^{(i)}$ and $E_{n}^{(i)}$ are independent of
$X_{i}$, Lemma~\ref{lem:loo} gives the sandwich
\[
  \Prob\big(c_{n}^{(i)}<x-\delta_{n}\big)-\Prob\big((E_{n}^{(i)})^{c}\big)
  \;\leq\;
  \Prob\big(A_{i}=1\mid X_{i}=x\big)
  \;\leq\;
  \Prob\big(c_{n}^{(i)}<x+\delta_{n}\big)+\Prob\big((E_{n}^{(i)})^{c}\big),
\]
because on $E_{n}^{(i)}$ we have
$c_{n}^{(i)}<x-\delta_{n}\Rightarrow c_{n}<x$ and
$c_{n}<x\Rightarrow c_{n}^{(i)}<x+\delta_{n}$. The uniform convergence in
the proof of Lemma~\ref{lem:loo} implies
$\max_{i}|c_{n}^{(i)}-\theta|\xrightarrow{p}0$, and $\delta_{n}\to0$, so
both outer bounds converge to $\ind{x>\theta}$, uniformly over
$i\in I_1$. Symmetrically, for $i\in I_2$ the disclosure probability
given $X_{i}=x$ converges to $\ind{x<\theta}$, uniformly. Since
$X_{J}\sim F_{\theta}$ for every $J=i$ and
$\sum_{i\in I_1}v_{ni}=r_{n}\to1/2$,
\[
  \varphi_{n}(x)
  :=\Prob(A_{J}=1\mid X_{J}=x)
  =\sum_{i\in I_1}v_{ni}\,\Prob(A_{i}=1\mid X_{i}=x)
  +\sum_{i\in I_2}v_{ni}\,\Prob(A_{i}=1\mid X_{i}=x)
  \longrightarrow\tfrac12
\]
for $F_{\theta}$-almost every $x$, with $0\leq\varphi_{n}\leq1$. By
dominated convergence $\Prob(A_{J}=1)=\int\varphi_{n}\,dF_{\theta}
\to\tfrac12$. The conditional law of $X_{J}$ given $A_{J}=1$ has density
$\varphi_{n}/\!\int\!\varphi_{n}dF_{\theta}$ with respect to
$F_{\theta}$; this density is eventually bounded and converges to $1$
$F_{\theta}$-a.e., so it converges to $1$ in $L^{1}(F_{\theta})$ by
Scheff\'e's theorem, which is exactly total-variation convergence to
$F_{\theta}$.
\end{proof}

\subsection{Proof of Proposition~\ref{thm:leverage}}
\label{app:leverage}

Let $U(c)=\E_{\theta}(X-c)^{+}=\int_{c}^{\overline x}[1-F_{\theta}(z)]dz$
and $L(c)=\E_{\theta}(c-X)^{+}=\int_{\underline x}^{c}F_{\theta}(z)dz$.
The equilibrium condition is $rU(c^{*})=(1-r)L(c^{*})$. As
$c\uparrow\overline x$, continuity and positivity of $f_{\theta}$ at
$\overline x$ give $1-F_{\theta}(z)=f_{\theta}(\overline x)(\overline
x-z)(1+o(1))$, hence
$U(c)=\tfrac12 f_{\theta}(\overline x)(\overline x-c)^{2}(1+o(1))$, while
$L(c)\to L(\overline x)=\overline x-\theta$. By
Proposition~\ref{prop:extremization}, $c^{*}(r)\uparrow\overline x$ as
$r\uparrow1$. Substituting the expansions into the equilibrium condition,
\[
  \tfrac12 f_{\theta}(\overline x)\,(\overline x-c^{*})^{2}(1+o(1))
  =\frac{1-r}{r}\,(\overline x-\theta)(1+o(1)),
\]
and solving for $\overline x-c^{*}$ gives the first display. The second is
symmetric. \qed

\subsection{Proof of Lemma~\ref{lem:insufficient}}
\label{app:insufficient}

The receiver action equals the anchored visible average
\eqref{eq:action} evaluated at $S^{*}$, so $c_{n}$ is a function of the
record $R$. Suppose, for contradiction, that $R$ is sufficient. By the
factorization criterion, the joint density of the evidence pool factors as
$g_{\theta}(R(x))\,h(x)$ for some measurable $g_{\theta},h$. Fix a
labeling $s(\ell)\in\{1,2\}$ that identifies the sender who owns observation
$\ell$. Fix a
realization $x$ in which some upward item $i$ is withheld, i.e.\
$x_{i}<c_{n}(R(x))$. If Sender 1 is the identifying camp, such
realizations have positive probability for every finite $n$: take $x_i$ in
an interval below $a_0$ and all other observations in a sufficiently small
neighborhood of $a_0$. The anchor and the other observations then imply
$H_n(x_i)>0$, hence $c_n>x_i$. (If Sender 2 is identifying, use
the symmetric event with one downward item above $a_0$ and repeat the
argument below with inequalities reversed.) Consider
perturbations $x_{i}'\in
(\underline x,\,c_{n})$ of coordinate $i$, holding all other coordinates
fixed. The disclosed sets and disclosed values are unchanged (item $i$
remains withheld: its value stays below the threshold, and the threshold
$c_{n}$, being a function of disclosed items and the anchor only, is
unchanged). Hence $R$ is constant along the perturbation, and the
factorization implies that
\[
  \frac{f_{1,\theta}(x_{i}')}{f_{1,\theta'}(x_{i}')}
  =
  \frac{g_{\theta}(R)\,h(x')/\prod_{\ell\neq i}f_{s(\ell),\theta}(x_{\ell})}
       {g_{\theta'}(R)\,h(x')/\prod_{\ell\neq i}f_{s(\ell),\theta'}(x_{\ell})}
\]
is constant in $x_{i}'$ on the open interval $(\underline x,c_{n})$,
contradicting the likelihood-ratio condition in
Lemma~\ref{lem:insufficient}. To interpret the factorization equalities
almost everywhere, fix a pair $\theta\ne\theta'$ and choose the other
coordinates from their Fubini-full-measure set. The resulting likelihood
ratio is constant for almost every withheld coordinate on the interval;
continuity and positivity of the densities extend constancy to the entire
interval. \qed

\subsection{Proof of Proposition~\ref{prop:private}}
\label{app:private}

For every $n$, Sender $k$'s type space $\mathcal X^{n_k}$, $k\in\{1,2\}$, is compact
metric and its disclosure-action space is finite, endowed with the discrete
topology. Payoffs are bounded and continuous in actions, so they are
equicontinuous; conditional on the commonly known state, the two type
vectors are independent, so the information structure is absolutely
continuous with respect to the product of its marginals. The existence
theorem of \citet{milgromweber1985} therefore yields an equilibrium in
distributional strategies. Regular conditional probabilities give the
equivalent behavioral-strategy representation used in the statement. The
argument below applies to every equilibrium sequence and every deterministic
$\eta_n\to\bar\eta\in[0,\infty)$.

\paragraph{Step 1: tails.}
Fix any strategy of Sender 2 and any realization of the
Sender 1's pool. Sender 1's expected payoff from a set $S_1$ is
$\E[(A+\sum_{S_1}v x_{i})/(B+\sum_{S_1}v)]$, where
$(A,B)=(\eta_n a_0+\sum_{S_2}v X_{j},\,\eta_n+\sum_{S_2}v)$ is generated by
the opponent's play and is independent of Sender 1's pool, with
$B\geq\eta_n>0$. If $S_1$ contains $i$ and excludes $j$ with
$x_{j}>x_{i}$, swapping $j$ for $i$ changes the payoff by
$v\,\E[(x_{j}-x_{i})/(B+\sum_{S_1}v)]>0$ pointwise in $(A,B)$: a strict
improvement. Hence, conditional on disclosing $k$ observations, the unique
best set consists of the top $k$ values. Since the pool is atomless, ties
are null, and every pure action in the support of a best response (hence,
for almost every realization and private-randomization draw, every
equilibrium action) is a tail in own values. Under the convention in the
statement, its cutoff $T_1^n$ lies strictly between the marginal disclosed
and excluded order statistics and satisfies
$S_1=\{i:X_i>T_1^n\}$ exactly, apart from the null endpoint events. The
bottom-tail argument gives $S_2=\{j:X_j<T_2^n\}$. Moreover,
$T_{1}^{n}\perp T_{2}^{n}$ because each is a measurable function of its
camp's independent pool and private randomization only.

\paragraph{Step 2: security against arbitrary disclosure.}
Put $c_n=c_n^{\circ}$. Because the left side of
\eqref{eq:private-target} is strictly decreasing and converges uniformly
to $\bar\eta(a_0-c)+H(c;r,\theta)$, its zeros satisfy
$c_n\to c^*_{\bar\eta}\in(\underline x,\overline x)$. At the deterministic
point $c_n$, independence and camp-uniform weights give
\[
  \E[H_n(c_n)]=0,
  \qquad
  \operatorname{Var}(H_n(c_n))
  \leq (\overline x-\underline x)^2\max_i v_{ni}=o(1),
\]
so $H_n(c_n)=o_p(1)$.

Suppose Sender 1 discloses exactly $\{i:X_i>c_n\}$, while the
Sender 2 uses an arbitrary, possibly randomized disclosure set
$S_2$. Subtracting $c_n$ from the induced action gives
\[
 a_n-c_n=
 \frac{\eta_n(a_0-c_n)+\sum_{i\in I_1}v_{ni}(X_i-c_n)^+
       +\sum_{j\in S_2}v_{nj}(X_j-c_n)}
      {\eta_n+\sum_{i\in I_1}v_{ni}\ind{X_i>c_n}
       +\sum_{j\in S_2}v_{nj}}.
\]
For every realization and every $S_2$, the numerator is at least
$H_n(c_n)$: among downward items, its smallest value is obtained by
including every item below $c_n$ and no item above it. Moreover, the
denominator is bounded away from zero with probability approaching one,
because it contains the upward tail mass, whose probability limit is
$r[1-F_{1,\theta}(c^*_{\bar\eta})]>0$. Therefore this cutoff guarantees
\begin{equation}
  a_n\geq c_n-o_p(1)
  \quad\text{uniformly over all downward disclosure rules.}
  \label{eq:private-up-security}
\end{equation}
The symmetric argument shows that disclosing exactly the downward items
below $c_n$ guarantees
\begin{equation}
  a_n\leq c_n+o_p(1)
  \quad\text{uniformly over all upward disclosure rules.}
  \label{eq:private-down-security}
\end{equation}
The empirical deviations and the tail-mass event used in these bounds do
not depend on the opponent's rule. Hence, for every $\varepsilon>0$, taking
the supremum over all measurable randomized opponent rules gives exactly
the two probability limits in part (a).
All actions lie in the compact interval $\mathcal X$, so the same bounds
hold in expectation with $o(1)$ remainders. If $V_n$ is the expected
payoff in any equilibrium, optimality and the two deviations imply
\begin{equation}
  c_n-o(1)\leq V_n\leq c_n+o(1).
  \label{eq:private-value}
\end{equation}
This also proves the asserted maxmin statement without restricting the
opponent to a threshold rule.

\paragraph{Step 3: equilibrium thresholds concentrate.}
For deterministic cutoffs $(s,t)$ define the population action
\[
 y_n(s,t)=
 \frac{\eta_n a_0+r_n\E_\theta[X_1;X_1>s]
       +(1-r_n)\E_\theta[X_2;X_2<t]}
      {\eta_n+r_n\Prob_\theta(X_1>s)
       +(1-r_n)\Prob_\theta(X_2<t)}.
\]
At $c_n$, equation \eqref{eq:private-target} gives
$y_n(c_n,c_n)=c_n$. More strongly,
\begin{equation}
  g_{2,n}(t):=y_n(c_n,t)-c_n\geq0,
  \qquad
  g_{1,n}(s):=c_n-y_n(s,c_n)\geq0,
  \label{eq:private-gaps}
\end{equation}
and, by the positive densities, either inequality is strict unless its
argument equals $c_n$. For example, moving $t$ above $c_n$ adds positive
downward observations, while moving it below $c_n$ removes negative ones;
both changes raise the action from $c_n$. The upward statement is
symmetric. Since $c_n$ converges to an interior point, compactness and
continuity imply that for every $\varepsilon>0$ there are
$\delta_\varepsilon>0$ and $N_\varepsilon$ such that, for $n\geq
N_\varepsilon$,
\begin{equation}
 \inf_{|t-c_n|\geq\varepsilon}g_{2,n}(t)\geq\delta_\varepsilon,
 \qquad
 \inf_{|s-c_n|\geq\varepsilon}g_{1,n}(s)\geq\delta_\varepsilon.
 \label{eq:private-separation}
\end{equation}

The four camp-weighted empirical tail probabilities and tail first
moments converge uniformly to their population counterparts by the grid
argument in Lemma~\ref{lem:uniform}. This remains valid at a data-dependent
threshold. If Sender 1 deviates to the constant cutoff $c_n$
against the equilibrium threshold $T_2^n$, the denominator contains the
upward tail mass at $c_n$ and is therefore bounded away from zero with
probability approaching one. Consequently its deviation payoff is
\[
  y_n(c_n,T_2^n)+o_p(1)
  =c_n+g_{2,n}(T_2^n)+o_p(1).
\]
Equilibrium optimality, this deviation, and the upper bound in
\eqref{eq:private-value} yield
$\E[g_{2,n}(T_2^n)]\to0$. Sender 2's deviation to $c_n$ and
the lower bound in \eqref{eq:private-value} similarly give
$\E[g_{1,n}(T_1^n)]\to0$. Hence \eqref{eq:private-separation} and Markov's
inequality imply
\[
  T_2^n-c_n\xrightarrow{p}0,
  \qquad
  T_1^n-c_n\xrightarrow{p}0.
\]

\paragraph{Step 4: the action and target converge.}
On the event that both thresholds lie near $c_n$, the visible population
mass is bounded away from zero uniformly in $n$. The same uniform laws of
large numbers therefore give
$a_n-y_n(T_1^n,T_2^n)\xrightarrow{p}0$. Continuity, the preceding threshold
convergence, and $y_n(c_n,c_n)=c_n$ then imply
$a_n-c_n\xrightarrow{p}0$. Finally, the uniform convergence of
$\eta_n(a_0-c)+H(c;r_n,\theta)$ used in Step~2 gives
$c_n\to c^*_{\bar\eta}$. When $\bar\eta=0$, this limit is
$c^*(r,\theta)$ by Proposition~\ref{thm:largepool}. \qed

\subsection{Proof of Proposition~\ref{prop:personalized}}
\label{app:personalized}

For each row define the weighted upper- and lower-tail processes
\begin{align*}
 Q_{1,nj}(c)&=\sum_{i\in I_1}w_{nji}\ind{X_i>c},
 &M_{1,nj}(c)&=\sum_{i\in I_1}w_{nji}X_i\ind{X_i>c},\\
 Q_{2,nj}(c)&=\sum_{i\in I_2}w_{nji}\ind{X_i<c},
 &M_{2,nj}(c)&=\sum_{i\in I_2}w_{nji}X_i\ind{X_i<c}.
\end{align*}
Their expectations are, respectively, $\rho_{nj}$ or $1-\rho_{nj}$ times
the corresponding camp tail probability or truncated first moment. Since
\[
 \sum_{i\in I_k}w_{nji}^2
 \leq \omega_n\sum_{i\in I_k}w_{nji}\leq\omega_n,
 \qquad k\in\{1,2\},
\]
weighted Hoeffding gives, at every fixed cutoff and, for example, for an
upper-tail probability,
\[
 \Pr\!\left(
   \max_{j\in J_n}|Q_{1,nj}(c)-\rho_{nj}\Prob(X_1>c)|>\varepsilon
 \right)
 \leq 2|J_n|\exp\{-2\varepsilon^2/\omega_n\}.
\]
The other three processes satisfy the same type of bound, with a constant
depending only on the fixed bounded interval $\mathcal X$ for the first
moments.

Take a finite grid whose population tail probabilities and truncated
moments vary by at most an arbitrarily small amount between adjacent grid
points. Empirical tail probabilities are monotone in the cutoff. For the
possibly signed first moments, write
$X=\underline x+(X-\underline x)$ in the upper tail and
$X=\overline x-(\overline x-X)$ in the lower tail; the shifted summands are
nonnegative, so their empirical tail sums are also monotone and the same
bracketing argument applies. A union bound over listeners, processes, and
grid points, together with
$\omega_n\log(2|J_n|)\to0$, therefore gives uniform convergence over both
$j$ and $c$ of all four processes to their population counterparts; in
particular,
\begin{align*}
 \max_{j\in J_n}\sup_{c\in\mathcal X}
 \big|Q_{1,nj}(c)-\rho_{nj}\Prob(X_1>c)\big|&\xrightarrow{p}0,\\
 \max_{j\in J_n}\sup_{c\in\mathcal X}
 \big|M_{1,nj}(c)-\rho_{nj}\E[X_1;X_1>c]\big|&\xrightarrow{p}0,
\end{align*}
with the analogous two statements for the lower camp, replacing
$\rho_{nj}$ by $1-\rho_{nj}$.

This uniformity permits evaluation at the endogenous cutoff without an
independence argument. Since $c_n\xrightarrow{p}c^*$ and the population
tail functions are continuous, listener numerators and denominators
converge uniformly to the numerator and denominator defining
$y(\rho_{nj};c^*)$. The limiting denominator is bounded below uniformly
over $\rho\in[0,1]$ by
$\min\{\Prob(X_1>c^*),\Prob(X_2<c^*)\}>0$. Thus the probability that any
listener has a zero denominator vanishes, and uniform ratio convergence
proves \eqref{eq:uniform-personalization}.

Write $E_{1}=\E_{\theta}[X_1\mid X_1>c^{*}]$ and
$E_{2}=\E_{\theta}[X_2\mid X_2<c^{*}]$, with $E_{1}>c^{*}>E_{2}$
since both tails are nondegenerate. Then $y$ is a weighted average of
$E_{1}$ and $E_{2}$ with weight on $E_{1}$ strictly increasing in $\rho$;
differentiation gives
$\partial y/\partial\rho=
\Prob(X_1>c^{*})\Prob(X_2<c^{*})(E_{1}-E_{2})/(\text{denominator})^{2}>0$.
Part (a) follows. For (b): at $r=r^{*}$, $c^{*}=\theta$
(Theorem~\ref{thm:design}), and $y(\rho;\theta)=\theta$ if and only if
$\rho\,\Gamma_1(\theta)=(1-\rho)\Gamma_2(\theta)$, i.e.\
$\rho=r^{*}(\theta)$. Let
$d_n=\max_j|\rho_{nj}-r^*(\theta)|$. If $d_n\to0$, uniform continuity of
$y$ and \eqref{eq:uniform-personalization} give uniform receiver accuracy.
Conversely, for every
$\delta\in(0,\max\{r^*(\theta),1-r^*(\theta)\}]$, strict monotonicity and
compactness imply
\[
 m_\delta
 =\min_{\rho\in[0,1]:\,|\rho-r^*(\theta)|\geq\delta}
   |y(\rho;\theta)-\theta|>0.
\]
If $d_n\not\to0$, some subsequence has $d_n\geq\delta$; on that subsequence
\eqref{eq:uniform-personalization} implies
$\Pr(\max_j|a_{nj}-\theta|>m_\delta/2)\to1$, contradicting uniform accuracy.
This proves the reverse implication. \qed

\subsection{Proof of Proposition~\ref{prop:neutral-backfill}}
\label{app:neutral-backfill}

\begin{proof}
\emph{Part (i): finite equilibrium.}
For \(\lambda<1\), use \(\sum_i v_{ni}=1\) to rewrite
\eqref{eq:partial-action} as
\begin{equation}
 a_{n,\lambda}(S)
 =
 \frac{(\eta_n+\lambda)a_0
       +\sum_{i\in S}v_{ni}(X_i-\lambda a_0)}
      {\eta_n+\lambda
       +(1-\lambda)\sum_{i\in S}v_{ni}}.
 \label{eq:partial-rewrite}
\end{equation}
Fix an opponent's disclosure set and suppose the current action is \(y\).
Adding an item of value \(x\) and weight \(v>0\) raises the ratio if and
only if
\[
 x-\lambda a_0>(1-\lambda)y,
 \qquad\text{equivalently}\qquad
 x>t(y):=\lambda a_0+(1-\lambda)y.
\]
Removing a disclosed item has the reverse implication. Hence every upward
best response is an upper tail and every downward best response is a lower
tail, both around the cutoff \(t(y)\), up to items exactly at the cutoff.

The function \(K_{n,\lambda}\) is continuous and strictly decreasing. Put
\[
 t_L=\lambda a_0+(1-\lambda)\underline x,
 \qquad
 t_U=\lambda a_0+(1-\lambda)\overline x.
\]
At \(t_L\), the anchor term in \eqref{eq:partial-Kn} equals
\((\eta_n+\lambda)(a_0-\underline x)\), while the entire negative-tail term is
at most \(\lambda(a_0-\underline x)\); the positive-tail term is
nonnegative. Thus \(K_{n,\lambda}(t_L)>0\). Symmetrically,
\(K_{n,\lambda}(t_U)<0\). Its unique zero \(t_{n,\lambda}\) therefore lies
in \((t_L,t_U)\), and the corresponding
\(y_{n,\lambda}=(t_{n,\lambda}-\lambda a_0)/(1-\lambda)\) is interior.

At the proposed tail profile, subtracting \(y_{n,\lambda}\) from
\eqref{eq:partial-rewrite} gives a numerator proportional to
\[
 \frac{\eta_n+\lambda}{1-\lambda}(a_0-t_{n,\lambda})
 +\sum_{i\in I_1}v_{ni}(X_i-t_{n,\lambda})^+
 -\sum_{i\in I_2}v_{ni}(t_{n,\lambda}-X_i)^+
 =0.
\]
Thus the induced action is \(y_{n,\lambda}\), and the preceding marginal
comparison makes the two tails mutual global best responses. Conversely,
in any pure-strategy equilibrium the same comparison forces both disclosure
sets to be tails around \(t(y)\); the induced-action identity then implies
\(K_{n,\lambda}(t(y))=0\). Uniqueness of the zero gives
\(t(y)=t_{n,\lambda}\) and \(y=y_{n,\lambda}\).

If \(\lambda=1\), equation \eqref{eq:partial-action} becomes
\[
 a_{n,1}(S)
 =a_0+\frac{\sum_{i\in S}v_{ni}(X_i-a_0)}{1+\eta_n}.
\]
Sender 1 therefore includes exactly its items above \(a_0\), and
Sender 2 exactly its items below \(a_0\). Substitution gives
\eqref{eq:full-imputation-action}, with uniqueness up to ties at \(a_0\).

\emph{Part (ii): large pools and comparative statics.}
For fixed \(\lambda<1\), the grid argument of
Lemma~\ref{lem:uniform} gives
\[
 \sup_{t\in\mathcal X}
 |K_{n,\lambda}(t)-K_\lambda(t;r,\theta)|
 \xrightarrow{p}0.
\]
The population function is continuous and strictly decreasing, is positive
at \(\underline x\), and is negative at \(\overline x\); it therefore has a
unique interior zero \(t_\lambda^*(r,\theta)\). Uniform convergence and
strict crossing imply convergence of the finite zero, and
\eqref{eq:partial-y} gives convergence of the action.

Write
\[
 U(t)=\E_\theta(X_1-t)^+,\qquad
 L(t)=\E_\theta(t-X_2)^+.
\]
At the population zero,
\[
 \frac{\partial K_\lambda}{\partial t}
 =-\frac{\lambda}{1-\lambda}
  -D_\lambda(t;r,\theta)<0,
 \qquad
 \frac{\partial K_\lambda}{\partial r}=U(t)+L(t)>0.
\]
The implicit-function theorem gives
\eqref{eq:partial-cutoff-cs}; division by \(1-\lambda\) gives
\eqref{eq:partial-action-cs}. Positive densities then imply that the upward
disclosure probability falls, the downward disclosure probability rises,
and the upper-tail conditional mean rises with \(r\).

For the boundary, write
\(t_\lambda^*=a_0+(1-\lambda)z_\lambda\). The root equation becomes
\[
 -\lambda z_\lambda
 +rU\bigl(a_0+(1-\lambda)z_\lambda\bigr)
 -(1-r)L\bigl(a_0+(1-\lambda)z_\lambda\bigr)=0.
\]
Bounded support makes \(z_\lambda\) bounded for \(\lambda\) near one.
Consequently
\[
 z_\lambda\longrightarrow rU(a_0)-(1-r)L(a_0),
 \qquad
 t_\lambda^*\longrightarrow a_0,
\]
and \(y_\lambda^*=a_0+z_\lambda\) converges to the displayed formula for
\(y_1^*\). Equation \eqref{eq:partial-cutoff-cs} shows that the cutoff
response converges to zero. At \(\lambda=1\), the ordinary weighted law of
large numbers applied to \eqref{eq:full-imputation-action} gives the stated
action limit, whose derivative in \(r\) is \(U(a_0)+L(a_0)>0\).

\emph{Part (iii): implementation.}
For \(\lambda<1\), the condition
\(y_\lambda^*(r,\theta)=\theta\) is equivalent to
\(t_\lambda^*=\bar t_\lambda(\theta)\). Substituting that cutoff into
\eqref{eq:partial-K} gives
\[
 0=\lambda(a_0-\theta)
   +r\Gamma_{1,\lambda}(\theta)
   -(1-r)\Gamma_{2,\lambda}(\theta).
\]
Solving for \(r\) gives \eqref{eq:partial-implementing-share}. The same
equation follows directly from the formula for \(y_1^*\) when
\(\lambda=1\). The resulting share is strictly between zero and one exactly
under \eqref{eq:partial-feasibility}. Finally,
\eqref{eq:partial-action-cs} (and \(U(a_0)+L(a_0)>0\) at \(\lambda=1\)) makes
the action strictly increasing in attention, so a feasible implementing
share is unique. The reductions at \(\lambda=0\) and under \(a_0=\theta\)
follow by substitution.
\end{proof}

\input{partial-reallocation-proof}
\input{private-learning-proof}

%% file: private-learning.tex
\subsection{Private evidence when the state is unknown}
\label{sec:unknown-state}

Neither sender needs to observe the state if its own growing pool permits
consistent estimation of it. Let the state have a common prior $\pi$ on a
compact set $\Theta\subset\mathrm{int}\,\mathcal X$. Maintain
Assumption~\ref{ass:evidence}, and suppose $(c,\theta)\mapsto
F_{k,\theta}(c)$ is continuous on $\mathcal X\times\Theta$, for each
$k$. The prior and conditional evidence laws are common knowledge. Each
sender observes only its own pool, not the state or the other pool. Retain
linear sender payoffs, camp-uniform weights, $n_k\to\infty$,
$r_n\to r\in(0,1)$, and any deterministic $\eta_n>0$ converging to
$\bar\eta\in[0,\infty)$. Write $c_n^{\circ}(\theta)$ for the zero of
\eqref{eq:private-target}, and $c_{\bar\eta}^{*}(r,\theta)$ for the zero
of $\bar\eta(a_0-c)+H(c;r,\theta)$.

Let $\widehat\theta_{kn}$ be a closest point in $\Theta$ to sender $k$'s
sample mean, selecting the smaller point in case of a tie. Define the
feasible disclosure rules
\[
 \widehat\tau_{1n}=\{i:X_i>c_n^{\circ}(\widehat\theta_{1n})\},
 \qquad
 \widehat\tau_{2n}=\{j:X_j<c_n^{\circ}(\widehat\theta_{2n})\}.
\]
These rules use only the sender's own observations and the known evidence
laws. They estimate a suitable disclosure standard without observing the
opponent's evidence.

\begin{proposition}[Private learning about the state]
\label{prop:unknown-state}
Under the preceding conditions, a behavioral Bayesian Nash equilibrium
exists for every $n$. Every pure disclosure action used in equilibrium is
almost surely a tail of the sender's own pool. Along every sequence of
equilibria, using the cutoff convention of Proposition~\ref{prop:private},
\[
 T_1^n-c_n^{\circ}(\theta)\xrightarrow{p}0,\qquad
 T_2^n-c_n^{\circ}(\theta)\xrightarrow{p}0,\qquad
 a_n-c_n^{\circ}(\theta)\xrightarrow{p}0.
\]
Here probability integrates over the state, samples, and private
randomizations. Moreover, for every $\varepsilon>0$,
\begin{align*}
 \sup_{\sigma_2}\Pr\{a_n(\widehat\tau_{1n},\sigma_2)
               <c_n^{\circ}(\theta)-\varepsilon\}&\longrightarrow0,\\
 \sup_{\sigma_1}\Pr\{a_n(\sigma_1,\widehat\tau_{2n})
               >c_n^{\circ}(\theta)+\varepsilon\}&\longrightarrow0.
\end{align*}
The target converges uniformly on $\Theta$ to
$c_{\bar\eta}^{*}(r,\theta)$. Thus, whenever $\eta_n\to0$, every
equilibrium action and both disclosure cutoffs converge in probability to
$c^{*}(r,\theta)$, with no restriction on the relative rates of sample
growth and anchor decay.
\end{proposition}

The result concerns the joint distribution induced by the common prior;
it does not assert convergence conditional on every state of prior
probability zero. It also retains a selection-naive receiver. The new
conclusion is that private uncertainty about the state and the other
pool need not change the limiting disclosure outcome when both senders
can learn the state from their own samples.

%% file: partial-reallocation-supplement.tex
\subsection{Finite foundations and an example for partial reallocation}
\label{sec:partial-reallocation-supplement}
\paragraph{Finite-pool foundation.}
The population calculation is the limit of a finite disclosure game. For
a disclosure set $S_k$, define its within-camp weight
$p_{k,n}=\sum_{i\in S_k}\omega_{ni}^k$ and its weighted mean
$m_{k,n}=\sum_{i\in S_k}\omega_{ni}^kX_i/p_{k,n}$ when $p_{k,n}>0$.
Set $m_{k,n}=a_0$ if $S_k$ is empty. With the positive anchor retained,
the finite action is
\begin{equation}
 a_n^\tau(S_1,S_2)=
 \frac{\eta_n a_0+r_n w_\tau(p_{1,n})m_{1,n}
       +(1-r_n)w_\tau(p_{2,n})m_{2,n}}
      {\eta_n+r_n w_\tau(p_{1,n})+(1-r_n)w_\tau(p_{2,n})}.
 \label{eq:partial-reallocation-finite}
\end{equation}
The rule assigns the fraction $\tau$ of the weight freed by withholding
proportionally among the camp's disclosed items; if there are none, that
fraction is assigned to the benchmark. Hence
\eqref{eq:partial-reallocation-finite} reduces exactly to the baseline
finite rule at $\tau=0$ and the full-reallocation rule at $\tau=1$.

\begin{lemma}[Finite equilibrium and the population limit]
\label{lem:partial-reallocation-finite}
For every finite evidence profile, positive within-camp weights,
$\eta_n>0$, $r_n\in(0,1)$, and $\tau\in[0,1]$, the finite game
\eqref{eq:partial-reallocation-finite} has a pure Nash equilibrium. All
Nash equilibria, including mixed equilibria, induce one action $c_n^\tau$
with probability one.

For the population environment of
Proposition~\ref{prop:partial-reallocation}, additionally assume uniform
within-camp item weights, conditionally independent observations with law
$F_\theta$ in each camp, and diverging camp sizes. If $\eta_n\to0$ and
$r_n\to r\in\mathcal R$, then for each fixed
$\tau\in[0,\bar\tau)$,
\[
 c_n^\tau\xrightarrow{p}c^\tau(r,\theta).
\]
Under these additional assumptions, every equilibrium disclosure set is an upper or lower empirical tail,
respectively, except for ties. The two equilibrium disclosed fractions
converge in probability to the positive population equilibrium fractions,
uniformly across finite-game equilibrium selections.
\end{lemma}

The proof is in Appendix~\ref{app:partial-reallocation-finite}.
\paragraph{A fixed-support two-state example.}
Let the support be $[0,1]$. In state $\theta_H=2/3$, let both camps have
 density $f_H(x)=2x$; in state $\theta_L=1/3$, let both have
$f_L(x)=2(1-x)$. Give the high state prior probability $\pi\in(0,1)$.
These densities are strictly positive in the interior, integrate to one,
and have the stated means. Their tail quantities are
\[
 p_H=\frac59,\quad p_L=\frac49,\quad
 G_H=G_L=\frac8{81},\quad
 h_H=-h_L=\frac9{80}.
\]
Thus $\bar h=(2\pi-1)9/80$, while the two state-contingent adjustments
have opposite signs. At the balanced committed share,
\[
 c^\tau(1/2,\theta_H)=\frac23-\frac2{45}\tau+O(\tau^2),\qquad
 c^\tau(1/2,\theta_L)=\frac13+\frac2{45}\tau+O(\tau^2).
\]
Even after optimizing the committed share, the expected loss is
\[
 \frac{16\pi(1-\pi)}{2025}\tau^2+o(\tau^2).
\]
In each state, the camp with fewer observations on its preferred side of
the mean gains more from reallocating unused exposure. Here that moves
the two actions toward one another. A fixed exposure adjustment cannot
correct the opposite movements in both states.

\paragraph{Numerical illustration.}
For $\pi=4/5$, the local predictions are
$r_\tau^{EA}=1/2+(27/400)\tau+O(\tau^2)$ and
$\min\mathcal L_\tau=(64/50625)\tau^2+o(\tau^2)$.
Table~\ref{tab:partial-reallocation} reports numerical population solutions
from the exact beta tail means. The calculation brackets global optima of
each camp's centered tail contribution, solves the decreasing game
residual, and optimizes the committed share between the two implementing
shares. The grid and refinement checks are supplied in the replication
code. The table illustrates the specified example, rather than a global
monotonicity claim for arbitrary evidence laws.
\begin{table}[htbp]
\centering
\caption{Partial reallocation in the fixed-support example, $\pi=4/5$}
\label{tab:partial-reallocation}
\begin{tabular}{c c c c c}
\hline
Reassigned fraction & $r_\tau^*(\theta_L)$ & $r_\tau^*(\theta_H)$
& $r_\tau^{EA}$ & Minimum expected loss \\
\hline
0.00 & 0.500000 & 0.500000 & 0.500000 & 0 \\
0.01 & 0.498861 & 0.501139 & 0.500683 & $1.325\times10^{-7}$ \\
0.05 & 0.494012 & 0.505988 & 0.503583 & $4.029\times10^{-6}$ \\
0.10 & 0.487155 & 0.512845 & 0.507654 & $2.112\times10^{-5}$ \\
\hline
\end{tabular}
\end{table}

%% file: partial-reallocation-proof.tex
\subsection{Proof of Proposition~\ref{prop:partial-reallocation}}
\label{app:partial-reallocation}
\begin{proof}
Fix a state and suppress its subscript. Write $Q=F^{-1}$ and, for
$p\in(0,1]$, define the integrals and means of the top and bottom tails by
\[
 T_+(p)=\int_{1-p}^1Q(v)\,dv,\quad
 T_-(p)=\int_0^pQ(v)\,dv,\quad
 m_+(p)=T_+(p)/p,\quad m_-(p)=T_-(p)/p.
\]
For a candidate action $a$, define centered contributions
\[
 g_\pm(p;a,\tau)
 =(1-\tau)[T_\pm(p)-pa]+\tau[m_\pm(p)-a].
\]
Use the continuous extensions
$g_+(0;a,\tau)=\tau(\overline x-a)$ and
$g_-(0;a,\tau)=\tau(\underline x-a)$. These are limits of nonempty
tails, not the payoff of the actual empty submission. At a fixed positive
mass, the top tail uniquely maximizes the disclosed integral and the bottom
tail uniquely minimizes it, up to null sets. Because the multiplier
$(1-\tau)+\tau/p$ is positive, this rearrangement observation applies to
$g_\pm$ as well.

At $\tau=0$ and $a\in(\underline x,\overline x)$, the unique maximizing
mass for $g_+$ is $p_+(a)=1-F(a)$, and the unique minimizing mass for
$g_-$ is $p_-(a)=F(a)$. Both are interior. For $a$ in a compact
subinterval of the support, these masses are bounded away from zero and
one. Moreover, $g_+(\cdot;a,0)$ is strictly concave and
$g_-(\cdot;a,0)$ strictly convex: their second derivatives are
$-1/f(Q(1-p))$ and $1/f(Q(p))$, respectively. The perturbation from
$g_\pm(\cdot;a,0)$ is uniformly $O(\tau)$ on $[0,1]$. Hence all global
optimizers remain in a fixed neighborhood of their baseline interior
masses, uniformly for $a$ in that compact subinterval. Strict curvature
and the implicit function theorem then give unique interior optimizers
$p_\pm(a,\tau)$, with twice continuously differentiable dependence on
$(a,\tau)$. In particular, the endpoint extensions are not selected. At
$\tau>0$, their values strictly dominate the actual empty submission for
the respective sender, because
$\underline x<a_0<\overline x$; the interior optimizer therefore also
dominates the actual empty submission.

Define the optimized residual
\[
 H_\tau(a;r)
 =r\max_{p\in[0,1]}g_+(p;a,\tau)
   +(1-r)\min_{p\in[0,1]}g_-(p;a,\tau).
\]
At $\tau=0$, it is the baseline balance equation,
\[
 H_0(a;r)=r\E(X-a)^+-(1-r)\E(a-X)^+.
\]
Its roots, as $r$ varies in $\mathcal R$ and the state varies in the
finite set, lie in a compact subinterval of the support. For $\tau>0$,
every component of $g_\pm$ decreases in $a$ with slope in $[-1,-\tau]$.
Thus $H_\tau$ is strictly decreasing; uniform convergence to $H_0$
places its unique zero in the same interior neighborhood for small
$\tau$. There the optimizers just established are attained by positive-mass
tails. At the zero, replacing the upward tail by any other disclosure
weakly decreases its centered contribution, and replacing the downward
tail weakly increases its centered contribution. Dividing the resulting
centered numerator by the positive denominator in
\eqref{eq:partial-reallocation-action} proves both global best-response
inequalities. This constructs a Nash equilibrium, rather than merely
solving first-order conditions. Conversely, the action at any pure
equilibrium must make each sender's contribution optimal and must solve
$H_\tau=0$. To see this before locating its action, note that at any
candidate action an actual empty submission is strictly dominated in its
centered contribution by a sufficiently small favorable tail when
$\tau>0$. A nonempty best response must attain the supremum or infimum
over the compactified tail-mass interval, since an unattained endpoint
value can be approached by feasible nonempty tails. The same conclusion
holds at $\tau=0$, with the endpoint value equal to the actual empty
contribution. Strict monotonicity and the uniqueness of the tail
optimizers give uniqueness. At $\tau=0$, two empty submissions cannot be an
equilibrium, since either sender has a profitable tail deviation.

The envelope derivative is
\[
 \partial_a H_\tau
 =-[r w_\tau(p_+)+(1-r)w_\tau(p_-)]<0.
\]
The implicit function theorem gives the stated smoothness. At the
equilibrium, differentiation of $g_\pm$ with respect to the disclosed
mass gives the cutoff equations
\[
 w_\tau(p_+)(t_+-c^\tau)=\tau(m_+(p_+)-c^\tau),\qquad
 w_\tau(p_-)(t_--c^\tau)=\tau(m_-(p_-)-c^\tau).
\]
The two conditional means lie strictly above and below the action for
small $\tau$, by continuity from the baseline. These equations prove
\eqref{eq:partial-reallocation-gap}.

For the derivatives at $r=1/2,\tau=0$, set
$p=p_\theta$, $q=1-p$, and $G=G_\theta$. The envelope theorem gives
\[
 \partial_a H_0=-\frac12,\qquad
 \partial_r H_0=2G,\qquad
 \left.\partial_\tau H_\tau\right|_{\tau=0}
 =\frac G2\left(\frac1p-\frac1q\right).
\]
Consequently
\[
 \left.\partial_r c^\tau\right|_{(1/2,0)}=4G=\kappa_\theta,
 \qquad
 \left.\partial_\tau c^\tau\right|_{(1/2,0)}
 =\frac{G(q-p)}{pq}=-\kappa_\theta h_\theta.
\]
The implementing share follows by the implicit function theorem. Its
uniqueness on $\mathcal R$ follows because
$\partial_r H_\tau=\max g_+-\min g_->0$ in the relevant neighborhood,
so $\partial_r c^\tau>0$ uniformly there. This also proves
\eqref{eq:partial-reallocation-implementation}.

At $\tau=0$, the expected loss on $\mathcal R$ has its unique global
minimum at $r=1/2$. Uniform convergence therefore puts every minimizer
for small $\tau$ in a neighborhood of $1/2$. The second derivative of
expected loss at $(r,\tau)=(1/2,0)$ is
$2\sum_\theta\pi_\theta\kappa_\theta^2>0$; continuity and the
implicit function theorem give a unique local, hence global, minimizer,
with $r_\tau^{EA}-1/2=O(\tau)$. For a generic nearby share, write
$\delta=r-1/2$. Twice continuous differentiability gives
\[
 c^\tau(r,\theta)-\theta
   =\kappa_\theta(\delta-\tau h_\theta)
        +O((|\delta|+\tau)^2),\qquad
 \partial_r c^\tau(r,\theta)
   =\kappa_\theta+O(|\delta|+\tau).
\]
At the optimal share, $\delta=O(\tau)$, and the first-order condition for expected loss is consequently
$0=2\sum_\theta\pi_\theta\kappa_\theta^2
(\delta-\tau h_\theta)+O((|\delta|+\tau)^2)$.
This proves \eqref{eq:partial-reallocation-exante}. Substituting it into
\eqref{eq:partial-reallocation-implementation}, or into the joint Taylor
expansion of $c^\tau$ at $(1/2,0)$, gives
\eqref{eq:partial-reallocation-loss}. Finally,
\[
 \frac{d}{dp}\frac{2p-1}{4p(1-p)}
 =\frac{p^2+(1-p)^2}{4p^2(1-p)^2}>0.
\]
All prior weights and all $\kappa_\theta$ are strictly positive. The
weighted variance in \eqref{eq:partial-reallocation-loss} therefore
vanishes exactly when $p_\theta$ is constant across states.
\end{proof}

\subsection{Proof of Lemma~\ref{lem:partial-reallocation-finite}}
\label{app:partial-reallocation-finite}
\begin{proof}
For a candidate action $a$, write
\[
 U_n(a)=\max_{S_1} w_\tau(p_{1,n})(m_{1,n}-a),\qquad
 L_n(a)=\min_{S_2} w_\tau(p_{2,n})(m_{2,n}-a),
\]
where the actual empty-set option contributes $\tau(a_0-a)$. Define
\[
 H_n^\tau(a)=\eta_n(a_0-a)+r_nU_n(a)+(1-r_n)L_n(a).
\]
Each optimized term is continuous and nonincreasing, and the anchor makes
$H_n^\tau$ strictly decreasing. It is positive at $\underline x$ and
negative at $\overline x$, so it has a unique zero $c_n^\tau$.
Choose a maximizing set for $U_n(c_n^\tau)$ and a minimizing set for
$L_n(c_n^\tau)$. Their action equals $c_n^\tau$. Fixing the downward
set, any upward deviation weakly lowers the centered numerator of
$a_n^\tau-c_n^\tau$; fixing the upward set, any downward deviation
weakly raises it. The denominator is always positive. These are global
best responses and constitute a pure equilibrium.

The disclosure game is a finite zero-sum game after taking the upward
sender's utility to be the action and the downward sender's utility its
negative. Every equilibrium has expected action $c_n^\tau$. Against
the downward pure saddle strategy constructed above, all upward
strategies yield at most $c_n^\tau$; every strategy in the support of
an equilibrium upward mixture must attain equality, since the mixture
guarantees the value. Equality forces its centered contribution to equal
$U_n(c_n^\tau)$. The symmetric argument forces every downward support
set's contribution to equal $L_n(c_n^\tau)$. Any pair of such support
sets therefore makes the centered numerator zero. This proves the claim
about realized actions in all mixed equilibria. Under uniform item
weights, maximizing or minimizing at a fixed disclosed count requires
the corresponding empirical tail, by a strict swap argument when there
are no ties.

For convergence, let $Q_{k,n}$ denote camp $k$'s empirical quantile
function, using its sample minimum and maximum at the endpoints. The
continuous, strictly increasing $F_\theta$ has a uniformly continuous
inverse on $[0,1]$. Uniform empirical distribution convergence therefore
gives
\[
 \max_k\sup_{v\in[0,1]}|Q_{k,n}(v)-Q(v)|\xrightarrow{p}0.
\]
For any $p>0$, the difference between the empirical and population top
tail means is at most this quantile error: divide the integral of their
quantile difference over $[1-p,1]$ by $p$. The same holds for bottom
tail means and for their continuous extensions at $p=0$. Thus the
centered empirical tail-contribution functions converge uniformly in
$p\in[0,1]$ and $a\in[\underline x,\overline x]$ to $g_\pm$.
The available nonempty disclosed-count fractions form increasingly
dense grids. Uniform continuity of $g_\pm$ then implies that their grid
optima converge uniformly to the population optima. Adding the actual
empty option does not change these limiting optima: for the upward
camp its value is bounded above by the top-tail endpoint extension, and
for the downward camp it is bounded below by the bottom-tail extension.
Consequently $H_n^\tau$ converges uniformly to $H_\tau(\cdot;r)$.
The latter has a unique strictly interior zero, proving convergence of
$c_n^\tau$.

For small $\tau$, both population tail optimizers are unique and
interior, as proved above. Uniform convergence of their objective
functions, including when evaluated at $c_n^\tau$, implies convergence
of every maximizing and minimizing grid fraction to its population
counterpart. Actual empty submissions are asymptotically excluded by the
strict gap between their contributions and the positive-mass optima.
The support characterization already proved applies to every finite
equilibrium, completing the uniform equilibrium-selection claim.
\end{proof}

%% file: private-learning-proof.tex
\subsection{Proof of Proposition~\ref{prop:unknown-state}}
\label{app:unknown-state}

Write $\ell=\underline x$, $u=\overline x$, $d=u-\ell$, and abbreviate
$c_n(\theta)=c_n^{\circ}(\theta)$. All expectations in this proof that
are not indexed by $\theta$ include the prior. Conditional independence
of the pools is maintained; unconditional or posterior independence is
neither assumed nor used.

\paragraph{Step 1: existence and tail disclosure.}
For each $n$, the joint type law has Lebesgue density
\[
 q_n(x_1,x_2)=\int_\Theta
  \prod_{i\in I_1}f_{1,\theta}(x_i)
  \prod_{j\in I_2}f_{2,\theta}(x_j)\,d\pi(\theta).
\]
The marginal type densities are strictly positive almost everywhere on
the respective products of the support interior. Consequently the joint
type law is absolutely continuous with respect to the product of its
marginal laws, even though the types are generally correlated. The type
spaces are compact, the disclosure-action spaces are finite, and payoffs
are bounded and continuous in the realized pools and disclosure actions
because $\eta_n>0$. The existence theorem of
\citet[Theorem~1]{milgromweber1985} therefore gives an equilibrium in
distributional strategies. Disintegration gives its behavioral
representation.

Fix a sender's realized pool and any opponent strategy. Replacing a
disclosed observation by a larger undisclosed observation, while holding
the number disclosed fixed, strictly raises the receiver action for
every realization of the opponent's disclosure set. The gain is
$v(x_j-x_i)/(B+kv)>0$, where $v$ is the sender's common item weight,
$k$ the number disclosed, and $B\geq\eta_n$ the anchor plus the
opponent's disclosed weight. Thus the upward sender's best responses
use top tails, and the downward sender's use bottom tails. This
pointwise argument remains valid when the conditional distribution of
the opponent's pool depends on the sender's own realization. Ties have
probability zero. Every action in the support of an equilibrium
behavioral strategy is therefore a tail, almost surely.

\paragraph{Step 2: uniform estimation and empirical approximation.}
Let $C(\eta,r,\theta)$ be the zero of
$\eta(a_0-c)+H(c;r,\theta)$. The assumed continuity of the CDFs and
bounded support imply joint continuity of all tail probabilities and
tail first moments. Strict monotonicity, the strict endpoint signs,
and compactness imply that $C$ is continuous whenever $\eta\geq0$ is
bounded and $r$ stays in a compact subset of $(0,1)$. Hence
\[
 \sup_{\theta\in\Theta}|c_n(\theta)
       -c_{\bar\eta}^{*}(r,\theta)|\longrightarrow0.
\]
For all sufficiently large $n$, the targets are uniformly separated
from both endpoints, and their own-camp upper and lower tail masses
are bounded below by a positive constant, uniformly on $\Theta$.

Projection onto $\Theta$ satisfies
$|\widehat\theta_{kn}-\theta|\leq
2|\overline X_{kn}-\theta|$. Since
$\operatorname{Var}_\theta(\overline X_{kn})\leq d^2/(4n_k)$,
the estimates are uniformly consistent in $\theta$. Uniform continuity
of $C$ on the compact parameter set therefore gives
\begin{equation}
 c_n(\widehat\theta_{kn})-c_n(\theta)\xrightarrow{p}0,
 \label{eq:unknown-estimation}
\end{equation}
uniformly over the conditional state laws.

The empirical upper and lower tail probabilities and first moments also
converge uniformly in the threshold, in conditional probability uniformly
over $\theta$. For completeness, uniform continuity of the CDFs on
$\mathcal X\times\Theta$ gives a common finite grid whose adjacent
population probability masses are arbitrarily small. A union bound and
the variance bound for bounded independent observations control the
empirical CDF at every grid point uniformly over $\theta$; monotonicity
then controls the full CDF. The identity
\[
 \E_\theta[X_k;X_k>t]
 =t\Pr_\theta(X_k>t)+\int_t^u\Pr_\theta(X_k>z)\,dz
\]
and its empirical counterpart give the same conclusion for first
moments; lower tails follow by subtraction. These uniform empirical
bounds apply at any sample-dependent thresholds and require no
independence between those thresholds.

For deterministic thresholds $s,t$, define
\[
 y_{n,\theta}(s,t)=
 \frac{\eta_n a_0+r_n\E_\theta[X_1;X_1>s]
       +(1-r_n)\E_\theta[X_2;X_2<t]}
      {\eta_n+r_n\Pr_\theta(X_1>s)
       +(1-r_n)\Pr_\theta(X_2<t)}.
\]
Whenever at least one threshold is within a sufficiently small fixed
neighborhood of its target $c_n(\theta)$, the corresponding own-camp
tail bounds the population denominator away from zero. The empirical
action then approximates $y_{n,\theta}$ in probability, uniformly
over the other threshold. Combined with \eqref{eq:unknown-estimation},
this implies that a sender using its estimated target can be replaced,
for purposes of population payoff approximation, by the target
$c_n(\theta)$. All approximation errors are bounded, so their
expectations also tend to zero.

\paragraph{Step 3: feasible security strategies and the equilibrium value.}
Consider the upward sender's estimated-target rule, writing
$s_n=c_n(\widehat\theta_{1n})$ and $c=c_n(\theta)$.
Against an arbitrary downward disclosure set $S_2$, the numerator of
the action minus $c$ is bounded below by
\[
 \eta_n(a_0-c)
 +\frac{r_n}{n_1}\sum_{i\in I_1}(X_i-c)\ind{X_i>s_n}
 -\frac{1-r_n}{n_2}\sum_{j\in I_2}(c-X_j)^+.
\]
Replacing $s_n$ by $c$ changes this expression in absolute value by at
most $r_n|s_n-c|$, because every affected observation lies between
$s_n$ and $c$. The resulting expression converges to
$\eta_n(a_0-c)+H(c;r_n,\theta)=0$ uniformly in conditional probability,
by Step~2. The denominator contains the upward empirical tail at $s_n$
and is bounded away from zero with probability tending to one. Hence
the upward estimated-target rule guarantees
$a_n\geq c_n(\theta)-o_p(1)$ uniformly over all opponent rules. The
downward argument is symmetric. The exceptional events and bounds do
not depend on the opponent's strategy, proving the two uniform
security statements in the proposition.

Boundedness yields the corresponding expected-payoff guarantees. If
$V_n=\E[a_n]$ is the payoff in any equilibrium, the two feasible
unilateral deviations imply
\begin{equation}
 \E[c_n(\theta)]-o(1)\leq V_n
       \leq\E[c_n(\theta)]+o(1).
 \label{eq:unknown-value}
\end{equation}
This is an ex ante comparison: the proof does not treat knowledge of
$\theta$ as a feasible deviation.

\paragraph{Step 4: all equilibrium cutoffs concentrate.}
Define
\[
 g_{2,n}(\theta,t)=y_{n,\theta}(c_n(\theta),t)-c_n(\theta),
 \qquad
 g_{1,n}(\theta,s)=c_n(\theta)
                 -y_{n,\theta}(s,c_n(\theta)).
\]
Both gaps are nonnegative. With $c=c_n(\theta)$, the numerator of the
first gap is exactly
\[
 (1-r_n)\int_{\min\{t,c\}}^{\max\{t,c\}}|x-c|\,dF_{2,\theta}(x),
\]
and that of the second is
$r_n\int_{\min\{s,c\}}^{\max\{s,c\}}|x-c|\,dF_{1,\theta}(x)$.
Each is strictly positive unless its threshold equals $c$, by positive
interior densities. The denominators are uniformly positive because
the other sender uses its target. Joint continuity and compactness
therefore imply that, for every $\varepsilon>0$, there are
$\delta_\varepsilon>0$ and $N_\varepsilon$ such that for
$n\geq N_\varepsilon$,
\begin{align*}
 \inf_{\theta\in\Theta,\ |t-c_n(\theta)|\geq\varepsilon}
       g_{2,n}(\theta,t)&\geq\delta_\varepsilon,\\
 \inf_{\theta\in\Theta,\ |s-c_n(\theta)|\geq\varepsilon}
       g_{1,n}(\theta,s)&\geq\delta_\varepsilon.
\end{align*}
An empty infimum imposes no restriction.

Against the equilibrium downward tail $T_2^n$, the upward sender's
feasible estimated-target deviation has expected payoff
\[
 \E[c_n(\theta)]+\E[g_{2,n}(\theta,T_2^n)]+o(1)
\]
by Step~2. Equilibrium optimality and the upper bound in
\eqref{eq:unknown-value} force
$\E[g_{2,n}(\theta,T_2^n)]\to0$. The downward deviation and the lower
bound similarly force $\E[g_{1,n}(\theta,T_1^n)]\to0$. The gaps cannot
cancel across states because they are everywhere nonnegative.
Their uniform separation consequently gives
$T_k^n-c_n(\theta)\xrightarrow{p}0$ for each $k$.

\paragraph{Step 5: actions.}
With probability tending to one both equilibrium thresholds lie near
$c_n(\theta)$, so the empirical denominator is bounded away from zero.
Step~2 then gives
$a_n-y_{n,\theta}(T_1^n,T_2^n)\xrightarrow{p}0$.
Uniform continuity near the target pair and
$y_{n,\theta}(c_n(\theta),c_n(\theta))=c_n(\theta)$ prove the action
convergence. Step~2 already proved uniform convergence of the targets.
When $\bar\eta=0$ their limit is $c^*(r,\theta)$. No step compares
$\eta_n$ to either sample size. \qed

%% file: main.bbl
\begin{thebibliography}{99}

\bibitem[Au and Kawai(2020)]{aukawai2020}
Au, Pak Hung, and Keiichi Kawai.
2020.
``Competitive Information Disclosure by Multiple Senders.''
\emph{Games and Economic Behavior} 119: 56--78.

\bibitem[Chen(2011)]{chen2011}
Chen, Ying.
2011.
``Perturbed Communication Games with Honest Senders and Naive Receivers.''
\emph{Journal of Economic Theory} 146(2): 401--424.

\bibitem[Dewatripont and Tirole(1999)]{dewatriponttirole1999}
Dewatripont, Mathias, and Jean Tirole.
1999.
``Advocates.''
\emph{Journal of Political Economy} 107(1): 1--39.

\bibitem[Di Tillio et al.(2021)]{ditillioottavianisorensen2021}
Di Tillio, Alfredo, Marco Ottaviani, and Peter Norman S{\o}rensen.
2021.
``Strategic Sample Selection.''
\emph{Econometrica} 89(2): 911--953.

\bibitem[Ehm et al.(2016)]{ehm2016}
Ehm, Werner, Tilmann Gneiting, Alexander Jordan, and Fabian Kr\"uger.
2016.
``Of Quantiles and Expectiles: Consistent Scoring Functions, Choquet
Representations, and Forecast Rankings.''
\emph{Journal of the Royal Statistical Society: Series B} 78(3): 505--562.

\bibitem[Eyster and Rabin(2005)]{eysterrabin2005}
Eyster, Erik, and Matthew Rabin.
2005.
``Cursed Equilibrium.''
\emph{Econometrica} 73(5): 1623--1672.

\bibitem[Farina et al.(2026)]{farina2024}
Farina, Agata, Guillaume R. Fr\'echette, Alessandro Ispano, Alessandro
Lizzeri, and Jacopo Perego.
2026.
``The Selective Disclosure of Evidence: An Experiment.''
\emph{Review of Economic Studies}, corrected proof, rdag076.
\url{https://doi.org/10.1093/restud/rdag076}.

\bibitem[Froeb and Kobayashi(1996)]{froebkobayashi1996}
Froeb, Luke M., and Bruce H. Kobayashi.
1996.
``Naive, Biased, yet Bayesian: Can Juries Interpret Selectively Produced Evidence?''
\emph{Journal of Law, Economics, and Organization} 12(1): 257--276.

\bibitem[Gao(2026)]{gao2026}
Gao, Ying.
2026.
``Inference from Selectively Disclosed Data.''
Working paper, July 31, author manuscript.
\url{https://sites.google.com/view/ying-gao/research}.

\bibitem[Gentzkow and Kamenica(2017)]{gentzkowkamenica2017}
Gentzkow, Matthew, and Emir Kamenica.
2017.
``Competition in Persuasion.''
\emph{Review of Economic Studies} 84(1): 300--322.

\bibitem[Grossman(1981)]{grossman1981}
Grossman, Sanford J.
1981.
``The Informational Role of Warranties and Private Disclosure about Product
Quality.''
\emph{Journal of Law and Economics} 24(3): 461--483.

\bibitem[Hirshleifer and Teoh(2003)]{hirshleiferteoh2003}
Hirshleifer, David, and Siew Hong Teoh.
2003.
``Limited Attention, Information Disclosure, and Financial Reporting.''
\emph{Journal of Accounting and Economics} 36(1--3): 337--386.

\bibitem[Jin et al.(2021)]{jin2021}
Jin, Ginger Zhe, Michael Luca, and Daniel Martin.
2021.
``Is No News (Perceived As) Bad News? An Experimental Investigation of
Information Disclosure.''
\emph{American Economic Journal: Microeconomics} 13(2): 141--173.

\bibitem[Kartik et al.(2026)]{kartikleesuen2026}
Kartik, Navin, Frances Xu Lee, and Wing Suen.
2026.
``Multi-Sender Disclosure with Costs.''
Working paper, January 15, arXiv:2601.10048v1.
\url{https://arxiv.org/abs/2601.10048v1}.

\bibitem[Lee et al.(2019)]{leeullahwang2019}
Lee, Tae-Hwy, Aman Ullah, and He Wang.
2019.
``The Second-Order Asymptotic Properties of Asymmetric Least Squares
Estimation.''
\emph{Sankhya B} 81(Suppl.~1): 201--233.

\bibitem[Lipman and Seppi(1995)]{lipmanseppi1995}
Lipman, Barton L., and Duane J. Seppi.
1995.
``Robust Inference in Communication Games with Partial Provability.''
\emph{Journal of Economic Theory} 66(2): 370--405.

\bibitem[Milgrom(1981)]{milgrom1981}
Milgrom, Paul R.
1981.
``Good News and Bad News: Representation Theorems and Applications.''
\emph{Bell Journal of Economics} 12(2): 380--391.

\bibitem[Milgrom and Roberts(1986)]{milgromroberts1986}
Milgrom, Paul, and John Roberts.
1986.
``Relying on the Information of Interested Parties.''
\emph{RAND Journal of Economics} 17(1): 18--32.

\bibitem[Milgrom and Weber(1985)]{milgromweber1985}
Milgrom, Paul R., and Robert J. Weber.
1985.
``Distributional Strategies for Games with Incomplete Information.''
\emph{Mathematics of Operations Research} 10(4): 619--632.

\bibitem[Mohseni et al.(2022)]{mohseni2022}
Mohseni, Aydin, Cailin O'Connor, and James Owen Weatherall.
2022.
``The Best Paper You'll Read Today: Media Biases and the Public
Understanding of Science.''
\emph{Philosophical Topics} 50(2): 127--153.

\bibitem[Newey and Powell(1987)]{neweypowell1987}
Newey, Whitney K., and James L. Powell.
1987.
``Asymmetric Least Squares Estimation and Testing.''
\emph{Econometrica} 55(4): 819--847.

\bibitem[Ottaviani and Squintani(2006)]{ottavianisquintani2006}
Ottaviani, Marco, and Francesco Squintani.
2006.
``Naive Audience and Communication Bias.''
\emph{International Journal of Game Theory} 35(1): 129--150.

\bibitem[Rappoport(2025)]{rappoport2025}
Rappoport, Daniel.
2025.
``Evidence and Skepticism in Verifiable Disclosure Games.''
\emph{Theoretical Economics} 20(4): 1213--1246.

\bibitem[Shin(1998)]{shin1998}
Shin, Hyun Song.
1998.
``Adversarial and Inquisitorial Procedures in Arbitration.''
\emph{RAND Journal of Economics} 29(2): 378--405.

\bibitem[van der Vaart and Wellner(1996)]{vanderVaartWellner1996}
van der Vaart, Aad W., and Jon A. Wellner.
1996.
\emph{Weak Convergence and Empirical Processes: With Applications to
Statistics}. New York: Springer.

\end{thebibliography}
